\documentclass[a4paper,twoside,reqno,11pt,tbtags]{article}

\usepackage{amssymb,amsmath,amsthm,mathrsfs}
\usepackage{graphicx}
\usepackage{times}

\leftmargini=\baselineskip

\newtheoremstyle{localthm}
	{5pt} 
	{5pt} 
	{\sl} 
	{} 
	{\bf} 
	{{\rm.}} 
	{.7em} 
	{} 

\theoremstyle{localthm}
\newtheorem{Theorem}{Theorem}[section]
\newtheorem{Corollary}[Theorem]{Corollary}
\newtheorem{Proposition}[Theorem]{Proposition}
\newtheorem{Lemma}[Theorem]{Lemma}

\newtheoremstyle{localrem}
	{5pt} 
	{5pt} 
	{\rm} 
	{} 
	{\bf} 
	{{\rm.}} 
	{.7em} 
	{} 

\theoremstyle{localrem}
\newtheorem{Definition}[Theorem]{Definition}
\newtheorem{Remark}[Theorem]{Remark}
\newtheorem{Example}[Theorem]{Example}

\def\Ex{\mathop{\mathrm{I\!E}}\nolimits}

\newcommand{\M}{\mathbb{M}}
\newcommand{\R}{\mathbb{R}}
\newcommand{\V}{\mathbb{V}}
\newcommand{\W}{\mathbb{W}}

\newcommand{\LL}{\mathcal{L}}
\newcommand{\NN}{\mathcal{N}}
\newcommand{\TT}{\mathcal{T}}

\newcommand{\argmin}{\mathop{\mathrm{arg\,min}}}

\newcommand{\tr}{\mathop{\mathrm{tr}}\nolimits}

\def\hat{\widehat}

\newcommand{\Rqq}{\R_{}^{q\times q}}
\newcommand{\Rqqns}{\R_{\rm ns}^{q\times q}}
\newcommand{\Rqqorth}{\R_{\rm orth}^{q\times q}}
\newcommand{\Rqqsym}{\R_{\rm sym}^{q\times q}}
\newcommand{\Rqqsympd}{\R_{{\rm sym},+}^{q\times q}}

\usepackage{color}

\begin{document}

\title{Geodesic Convexity and Regularized Scatter Estimators}
\author{Lutz D\"umbgen\thanks{The University of Bern, Switzerland} ~and  
David E. Tyler\thanks{Rutgers, The State University of New Jersey, USA. Research partially supported by the National Science Foundation Grant No.\  DMS-1407751. Any opinions, findings and conclusions or recommendations expressed in this material are those of the author(s) and do not necessarily reflect those of the National Science Foundation.}}

\date{July 2016}

\maketitle

\begin{abstract}
As observed by \nocite{Auderset_etal_2005}{Auderset et al.\ (2005)} and \nocite{Wiesel_2012}{Wiesel (2012)}, viewing covariance matrices as elements of a Riemannian manifold and using the concept of geodesic convexity provide useful tools for studying $M$-estimators of multivariate scatter. In this paper, we begin with a mathematically rigorous self-contained overview of Riemannian geometry on the space of symmetric positive definite matrices and of the notion of geodesic convexity. The overview contains both a review as well as new results. In particular, we introduce and utilize first and second order Taylor expansions with respect to geodesic parametrizations. This enables us to give sufficient conditions for a function to be geodesically convex. In addition, we introduce the concept of geodesic coercivity, which is important in establishing the existence of a minimum to a geodesic convex function. We also develop a general partial Newton algorithm for minimizing smooth and strictly geodesically convex functions. We then use these results to generate a fairly complete picture of the existence, uniqueness and computation of regularized $M$-estimators of scatter defined using additive geodescially convex penalty terms. Various such penalties are demonstrated which shrink an estimator towards the identity matrix or multiples of the identity matrix. Finally, we propose a cross-validation method for choosing the scaling parameter for the penalty function, and illustrate our results using a numerical example.
\end{abstract}



\newpage

\tableofcontents

\newpage

\section{Introduction}
\label{sec:Introduction}

High dimensional multivariate data is becoming increasingly prevalent, with the estimation of the covariance matrix for such data sets being an important fundamental problem. The classical estimator, i.e.\ the sample covariance matrix, though, is known to be highly non-robust under longer tailed alternatives to the multivariate normal distribution, as well as being highly non-resistant to outliers in the data. Consequently, there have been numerous proposals for robust alternatives to the sample covariance matrix, with one of the earliest alternatives being the $M$-estimators of multivariate scatter \cite{Maronna_1976, Huber_1981}.
As with the multivariate $M$-estimators of scatter, most of the subsequent proposals for robust estimators of multivariate scatter are affine equivariant. However, for sparse multivariate data, that is when the sample size $n$ is less than or not much larger than the dimension of the data $q$, such estimators of scatter do not differ greatly from the sample covariance matrix, and for the case $q \le n$, they are simply proportional to the sample covariance, see \cite{Tyler_2010}.

Even when the distribution is normal and there are no outliers in the data set, the sample covariance matrix can still be unreliable for sparse data sets due to the large number of parameters being estimated, namely $q(q+1)/2$. Consequently, one may wish to model the covariance matrix using less parameters, or one may wish to give preference to certain covariance structures and pull the estimator towards such structures via penalization or regularization techniques. Traditionally, research on robust estimators of multivariate scatter have not taken these concerns into account, and the statistics literature has focused primarily on the unrestricted robust estimation of the scatter matrix. Within the signal processing community, though, there has been an increasing interest in the $M$-estimators of multivariate scatter \cite{Abramovich_etal_2013, Besson_etal_2013, Conte_etal_2002, Gini_Greco_2002, Ollila_Koivunen_2003, Ollila_Koivunen_2009,  Ollila_Tyler_2012, Ollila_etal_2003, Ollila_etal_2012, Pascal_etal_2008, Soloveychik_Wiesel_2013, Wiesel_2012, Zhang_etal_2013} and more recently an interest in developing regularized versions of them \cite{Chen_etal_2011, Couillet_McKay_2014, Ollila_Tyler_2014, Pascal_etal_2014, Wiesel_2012, Wiesel_2012b}. An important mathematical contribution arising from the area of signal processing is the realization in \cite{Wiesel_2012} that treating the multivariate scatter matrices as elements in a Riemannian manifold and using the notion of geodesic convexity can be very useful, leading to elegant theory as well as new results. These concepts had been applied previously within the statistics literature \cite{Auderset_etal_2005}, but only for the specific case of the distribution free $M$-estimator of multivariate scatter. More recently they have been used in \cite{Sra_Hosseini_2013} and implicitly in the survey paper \cite{Duembgen_etal_2015} on $M$-functionals of multivariate scatter.

The purpose of the present paper is threefold. We first review the standard Riemannian geometry on the space of symmetric positive definite matrices and the notion of geodesic convexity in Section~\ref{sec:G-Convexity}. In particular we introduce and utilize first and second order Taylor expansions of such functions with respect to geodesic parametrizations. Such expansions allow us to introduce sufficient conditions for a function to be geodesically convex. In addition we introduce the concept of geodesic coercivity, which is important in establishing the existence of both the $M$-estimators of scatter and their regularized versions. As in classical convex analysis, a real valued function on the space of symmetric positive definite matrices which is continuous, strictly geodesically convex and coercive has a unique minimizer.

Our second contribution is a general analysis of regularized $M$-estimators of multivariate scatter with respect to geodesic convexity and coercivity in Section~\ref{sec:Regularized.scatter}. Our starting point are results of \cite{Wiesel_2012, Zhang_etal_2013} and \cite{Duembgen_etal_2015} which show that the log-likelihood type functions underlying $M$-estimators of multivariate scatter are geodesically convex under rather general conditions. We show that various penalty functions favoring matrices which are close to the identity matrix or to multiples of the identity matrix are geodesically convex. This leads to a rather complete picture concerning existence and uniqueness of regularized $M$-functionals of scatter. It also provides new results on regularized sample covariance matrices when using penalty functions which are geodesically convex but not convex in the inverse of the covariance matrix. Furthermore, we propose a cross-validation method for choosing a scaling parameter for the penalty function.

Finally, we present a general partial Newton algorithm to minimize a smooth and strictly geodesically convex function in Section~\ref{sec:Algorithm}. This algorithm is a generalization of the partial Newton method of \cite{Duembgen_etal_2016} with guaranteed convergence. We illustrate this method with a numerical example in Section~\ref{sec:Example}.

All proofs and some auxiliary results are deferred to Section~\ref{sec:Proofs} and to a supplement \ref{sec:Auxiliary}. We begin with some notation and a brief background review.

\section{Background and Notation}
\label{sec:Background}

Let the space of symmetric matrices in $\R^{q\times q}$ be denoted by $\Rqqsym$, and let $\Rqqsympd$ stand for its subset of positive definite matrices, i.e.\ symmetric matrices with eigenvalues in $\R_+ := (0,\infty)$. For a distribution $Q$ on $\R^q$ with given center $0$ and a function $\rho : [0,\infty) \to \R$, an $M$-functional of multivariate scatter can be defined as a matrix which minimizes the objective function
\begin{equation}
\label{eq:Lrho}
	L_\rho(\Sigma,Q)
	\ := \ \int \bigl[ \rho(x^\top\Sigma^{-1}x) - \rho(\|x\|^2) \bigr] \, Q(dx)
		+ \log \det(\Sigma)
\end{equation}
over $\Sigma \in \Rqqsympd$. When $Q = Q_n$ represents an empirical distribution, then the minimizer defines an $M$-estimator of scatter, and the objective function can be viewed as a generalization of the negative log-likelihood function arising from an elliptical distribution \cite{Maronna_1976}. The term $\rho(\|x\|^2)$ is not needed when working with empirical distributions. In general, though, this term allows us to be able to consider distributions $Q$ for which $\int | \rho(\|x\|^2) | \, Q(dx) = \infty$.

For continuous $\rho$ with sill $a_o > q$, defined below, a minimizer $\Sigma \in \Rqqsympd$ to $L_\rho(\Sigma,Q_n)$ is known to exist, provided no subspace contains too may data points, or specifically if the following condition holds for $Q=Q_n$ \cite{Kent_Tyler_1991}.

\paragraph{Condition~1.}
For all linear subspaces $\V \subset \R^q$ with $0 \le \dim(\V) < q$,
\[
	Q(\V) < 1 - \frac{\{q - \dim(\V)\}}{a_o} ,
\]
where $a_o = \sup\{a : s^a\exp\{-\rho(s)\} \to 0 \ \text{as} \ s \to \infty \}$. (Note that the function $\rho$ in the present paper corresponds to $2\rho$ in \cite{Kent_Tyler_1991} and other publications.)

If $\rho$ is differentiable, then the critical points, and hence any minimizer, of \eqref{eq:Lrho} satisfy the $M$-estimating equations
\begin{equation}
\label{eq:Mee}
	\Sigma \ = \ \int u(x^\top\Sigma^{-1}x) xx^\top \, Q_n(dx)
\end{equation}
where $u(s) := \rho'(s)$. Furthermore, if we define $\psi(s) := su(s)$, then the sill $a_o$ equals the limit $\psi(\infty) = \lim_{s \to \infty} \psi(s)$ whenever the latter exists.

To assure the uniqueness of a minimizer to $L_\rho(\Sigma,Q_n)$ or a unique solution to the $M$-estimating equations \eqref{eq:Mee}, further conditions on the function $\rho$ are needed. It has been know since the introduction of the $M$-estimators of scatter \cite{Maronna_1976, Huber_1981} that one such sufficient condition is the following.

\paragraph{Condition~2.} The function $\rho$ is differentiable, with $u(s)$ being non-increasing and $\psi(s)$ being non-decreasing and strictly increasing for $\psi(s) < \psi(\infty)$.
\smallskip

\noindent
The proof of uniqueness given in \cite{Maronna_1976, Huber_1981} assumes more restrictive conditions on the distribution $Q$ than that given by Condition~1, although it is shown in \cite{Kent_Tyler_1991} that Conditions~1 and 2 are sufficient for the existence of a unique solution to \eqref{eq:Mee}, i.e.\ for the existence and uniqueness of the $M$-estimator of scatter. Some common examples of $M$-estimators satisfying Condition~2 are Huber's $M$-estimator for which $\psi(s) = K\min(s/c,1)$ with tuning constants $c > 0$ and $K > p$, and the maximum likelihood estimators derived from an elliptical t-distribution on $\nu > 0$ degrees of freedom, for which $\psi(s) = (\nu+q)s/(\nu+s)$. 

The above conditions lack some intuition as to why \eqref{eq:Lrho} has a unique minimum. The proofs of uniqueness given in \cite{Maronna_1976, Huber_1981, Kent_Tyler_1991} are based on a study of the $M$-estimating equations \eqref{eq:Mee}. Recall that for the classical case when $L_\rho(\Sigma,Q_n)$ corresponds to the negative log-likelihood under a $q$-dimensional normal distribution with mean zero and covariance $\Sigma$, i.e.\ when $\rho(s) = s$, then $L_\rho(\Sigma,Q_n)$ is strictly convex in $\Sigma^{-1}$ and hence has a unique minimizer, namely the sample covariance matrix.  For general $\rho$, however, $L_\rho(\Sigma,Q_n)$ tends not to be convex in $\Sigma^{-1}$.

Important insight into the function $L_\rho(\Sigma,Q_n)$ has recently been given within the area of signal processing. In particular, it is shown in \cite{Zhang_etal_2013} that if the function $\rho(e^x)$ is convex in $x \in \R$, then $L_\rho(\Sigma,Q_n)$ is geodesically convex in $\Sigma \in \Rqqsympd$, and that if the function $\rho(e^x)$ is strictly convex in $x \in \R$, then $L_\rho(\Sigma,Q_n)$ is strictly geodesically convex in $\Sigma \in \Rqqsympd$ provided the data span $\R^q$. Consequently, when Condition~1 holds, then the minimizer set for $L_\rho(\Sigma,Q_n)$ is a geodesically convex set when $\rho(e^x)$ is convex, and the minimizer is unique when $\rho(e^x)$ is strictly convex. The results on geodesic convexity, or g-convexity, not only give a mathematically elegant insight into uniqueness, but they also yield more general results. For example, $\rho(s)$ need not be differentiable. Also, when $\rho(s)$ is differentiable, then $\rho(e^x)$ is (strictly) convex in $x \in \R$ if and only if $\psi(s)$ is (strictly) increasing, with no additional conditions on $u(s)$ being needed, i.e. $u(s)$ need not be non-increasing.

The notion of g-convexity also allows for the development of new results regarding minimizing $L_\rho(\Sigma,Q)$ over a g-convex subset of $\Rqqsympd$, as well as minimizing a penalized objective function when the penalty function is also g-convex. Before addressing these problems, though, we provide a thorough review and present some new results on the notion of geodesic convexity.

\begin{Remark}
Note that our objective function \eqref{eq:Lrho} assumes $0$ to be the center of the distribution $Q$. In various applications in signal processing the center of $Q$ is often known or
hypothesized, and consequently all the aforementioned signal processing references presume a known center. In more traditional location-scatter problems, one could embed the location-scatter problem in dimension $q$ into a scatter-only problem in dimension $q+1$ as explained in \cite{Kent_Tyler_1991, Duembgen_etal_2015}. But regularization in this setting is less clear. If the location parameter is merely a nuisance parameter, then one can first center the data using an auxiliary estimate of location. Alternatively, the location parameter can be removed by symmetrization, i.e.\ instead of $Q$ one considers the symmetrized distribution $\LL(X - X')$ with independent random vectors $X, X' \sim Q$; see \cite{Duembgen_1998, Duembgen_etal_2015} for further details.
\end{Remark}

\section{Geodesic Convexity}
\label{sec:G-Convexity}

\subsection{A Riemannian geometry for scatter matrices}

We collect a few basic ideas about positive definite matrices and their geometry. For a full treatment we refer to \cite{Bhatia_2007}. The Euclidean norm of a vector $v \in \R^p$ is denoted by $\|v\| = \sqrt{v^\top v}$. For matrices $A, B$ with identical dimensions we write
\[
	\langle A, B\rangle \ := \ \tr(A^\top B)
	\quad\text{and}\quad
	\|A\| \ := \ \sqrt{\langle A, A\rangle} ,
\]
so $\|A\|$ is the Frobenius norm of $M$.

Equipped with this inner product $\langle\cdot,\cdot\rangle$ and norm $\|\cdot\|$, the matrix space $\Rqqsym$ is a Euclidean space of dimension $q(q+1)/2$, and $\Rqqsympd$ is an open subset thereof. But in the context of scatter estimation an alternative geometry turns out to be useful. Let $\hat{\Sigma}_n$ be the sample covariance matrix of independent random vectors $X_1, X_2, \ldots, X_n$ with distribution $\mathcal{N}_q(\mu,\Sigma)$ with $\mu \in \R^q$ and $\Sigma \in \Rqqsympd$. It is well known that
\[
	\hat{\Sigma}_n \ =_{\LL}^{} \ \Sigma^{1/2} (I_q + A_n) \Sigma^{1/2}
\]
with the identity matrix $I_q \in \Rqq$ and a random matrix $A_n \in \Rqqsym$. The distribution of $A_n$ depends only on $n$ and is invariant under transformations $A_n \mapsto U A_n U^\top$ with $U \in \Rqqorth$, the set of orthogonal matrices in $\Rqq$. Moreover, $A_n \to_p 0$ as $n \to \infty$. Thus one could measure the distance between $\hat{\Sigma}_n$ and $\Sigma$ by
\[
	\|A_n\| \ = \ \|\hat{\Sigma}_n - \Sigma\|_{\Sigma}
\]
with the local norm
\[
	\|\Delta\|_\Sigma \ := \ \|\Sigma^{-1/2} \Delta \Sigma^{-1/2}\|
	\ = \ \sqrt{ \tr(\Delta \Sigma^{-1} \Delta \Sigma^{-1})}
\]
corresponding to the local inner product
\[
	\langle \Delta, \tilde{\Delta}\rangle_\Sigma
	\ := \ \langle \Sigma^{-1/2}\Delta\Sigma^{-1/2},
		\Sigma^{-1/2}\tilde{\Delta}\Sigma^{-1/2}\rangle
	= \tr(\Delta\Sigma^{-1}\tilde{\Delta}\Sigma^{-1})
\]
of matrices $\Delta, \tilde{\Delta} \in \Rqqsym$.

To define a distance between two arbitrary matrices $\Sigma_0, \Sigma_1 \in \Rqqsympd$, we consider a smooth path $M$ connecting them. That means, $M : [0,1] \to \Rqqsympd$ is piecewise continuously differentiable with $M(0) = \Sigma_0$ and $M(1) = \Sigma_1$. Then we define the length of $M$ to be
\[
	L(M) \ := \ \int_0^1 \|\dot{M}(t)\|_{M(t)} \, dt .
\]
Denoting with $\Rqqns$ the set of nonsingular matrices in $\Rqq$, one can easily verify that for any $B \in \Rqqns$, the new path
\[
	M_B(t) \ := \ B M(t) B^\top
\]
connects the matrices $B\Sigma_0B^\top$ and $B\Sigma_1B^\top$ and has length
\[
	L(M_B) \ = \ L(M) .
\]
Here is a well-known key result about shortest paths in $\Rqqsympd$. For the reader's convenience we provide a self-contained proof in Supplement~\ref{sec:Auxiliary}.

\begin{Theorem}
\label{thm:geodesics}
Let $M : [0,1] \to \Rqqsympd$ be a path connecting $M(0) = \Sigma_0$ and $M(1) = \Sigma_1$. Then
\[
	L(M) \ \ge \ \bigl\| \log(\Sigma_0^{-1/2} \Sigma_1^{} \Sigma_0^{-1/2}) \bigr\|
\]
with equality if, and only if,
\[
	M(t) \ = \ \Sigma_0^{1/2} \,
		(\Sigma_0^{-1/2} \Sigma_1^{} \Sigma_0^{-1/2})_{}^{u(t)}
		\, \Sigma_0^{1/2}
\]
for some non-decreasing, piecewise continuously differentiable function $u : [0,1] \to \R$ with $u(0) = 0$ and $u(1) = 1$.
\end{Theorem}

Note that for a shortest path $M$, its track $\{M(t) : t \in [0,1]\}$ does not depend on the function $u$ but is equal to $\{N(u) : u \in [0,1]\}$ with the special path $N : [0,1] \to \Rqqsympd$ given by $N(u) := \Sigma_0^{1/2} \, (\Sigma_0^{-1/2} \Sigma_1^{} \Sigma_0^{-1/2})_{}^{u} \, \Sigma_0^{1/2}$. Indeed $M(t) = N(u(t))$, and the path $N$ has constant geodesic speed in the sense that for all $u \in [0,1]$,
\[
	\|\dot{N}(u)\|_{N(u)} \ = \ L(N) = L(M) .
\]
 
The preceding considerations involve matrix powers and logarithms. In general, a real valued function $h : \R \to \R$ can be extended to a matrix-valued function $h: \Rqqsym \rightarrow \Rqqsym$ in the following manner: Let $A \in \Rqqsym$ have spectral decomposition $A = U D(\lambda) U^\top$ with a matrix $U \in \Rqqorth$ of orthonormal eigenvectors of $A$ and a diagonal matrix $D(\lambda)$ with diagonal elements given by $\lambda = (\lambda_i)_{i=1}^q \in \R^q$ , then
\[
	h(A) \ := \ U D(h(\lambda)) U^\top
	\quad\text{for} \ s \in \R,
\]
using the convention $h(\lambda) := (h(\lambda_i))_{i=1}^q$. If $h$ is defined only on $\R_+$, then we restrict $A$ to $\Rqqsympd$ and obtain a matrix-valued function $h: \Rqqsympd \rightarrow \Rqqsym$. So, for $\lambda \in \R_+^q$, 
\[
	A^s \ := \ U D(\lambda^s) U^\top
	\quad\text{for} \ s \in \R
\]
and
\[
	\log(A) \ := \ U D(\log \lambda) U^\top .
\]
Also, for $A \in \Rqqsym$, 
\[
	\exp(A) \ := \ U D(e^\lambda) U^\top .
\]
This is consistent with the more general definition of a matrix exponential
\[
	\exp(A) \ := \ \sum_{k=0}^\infty \frac{A^k}{k!}
\]
which is defined for any arbitrary matrix $A \in \Rqq$.

Analogous to the real setting,  $\exp : \Rqqsym \to \Rqqsympd$ is a bijection with inverse mapping $\log : \Rqqsympd \to \Rqqsym$. For $A \in \Rqqsympd$,
\[
	A^s \ = \ \exp(s \log(A)) .
\]
Hence Theorem~\ref{thm:geodesics} shows that a shortest path between two matrices $\Sigma_0, \Sigma_1 \in \Rqqsympd$ is given by
\[
	M(t) \ := \ \Sigma_0^{1/2}
		\exp \bigl( t \log(\Sigma_0^{-1/2} \Sigma_1^{} \Sigma_0^{-1/2}) \bigr)
		\Sigma_0^{1/2} , \quad t \in [0,1] .
\]
Sometimes it is convenient to consider other factorizations of $\Sigma_0$, i.e.\ other square roots. If we write $\Sigma_0 = BB^\top$ for some $B \in \Rqqns$, then
\[
	M(t) \ = \ B \exp(tA) B^\top
	\quad\text{with}\quad
	A \ := \ \log(B^{-1} \Sigma_1 B^{-\top})
\]
and $B^{-\top} := (B^\top)^{-1} = (B^{-1})^\top$. The function $M(t)$ does not depend on the particular choice for $B$ since $B = \Sigma_0^{1/2} V$ for some $V \in \Rqqorth$. In particular, let $\Sigma_0^{-1/2} \Sigma_1^{} \Sigma_0^{-1/2} = V D(\eta) V^\top$ with $V \in \Rqqorth$ and $\eta \in \R_+^q$ containing the eigenvalues of $\Sigma_1 \Sigma_0^{-1}$. Then $\Sigma_0 = BB^\top$ and $\Sigma_1 = B D(\eta) B^\top$ with $B = \Sigma_0^{1/2} V$. For this choice and $\gamma := \log\eta$ we obtain the expression
\begin{equation}
\label{eq:BDB}
	M(t) \ = \ B D(\eta)^t B^\top \ = \ B \exp(t D(\gamma) ) B^\top ,
\end{equation}
which leads to a simple interpretation of the geodesic path from $\Sigma_0$ to $\Sigma_1$. Namely, after jointly diagonalizing $\Sigma_0$ and $\Sigma_1$, the geodesic path corresponds to the linear path connecting the logs of the diagonal elements.

\begin{Lemma}[Geodesic curves and $q$-dimensional surfaces]
\label{lem:geodesic.curves}
Let $B$ be an arbitrary matrix in $\Rqqns$. For $A \in \Rqqsym$ and $t \in \R$ let
\[
	\Sigma(t) \ := \ B \exp(tA) B^\top .
\]
This defines a geodesic curve in the following sense: For arbitrary different numbers $t_0, t_1$, a shortest path connecting $\Sigma(t_0)$ and $\Sigma(t_1)$ is given by
\[
	[0,1] \ni u \ \mapsto \ \Sigma((1 - u) t_0 + u t_1) .
\]
For $x \in \R^q$ let
\[
	\Gamma(x) \ := \ B D(e^x) B^\top = B \exp(D(x)) B^\top .
\]
This defines a $q$-dimensional geodesic surface in the following sense: For arbitrary $x_0, x_1 \in \R^q$, a shortest path connecting $\Gamma(x_0)$ and $\Gamma(x_1)$ is given by
\[
	[0,1] \ni u \ \mapsto \ B D \bigl( \exp((1 - u) x_0 + u x_1) \bigr) B^\top .
\]
\end{Lemma}

\paragraph{Local geodesic parametrizations.}
Closely related to the geodesic paths just described are the following local parametrizations of subsets of $\Rqqsympd$. For any matrix $\Sigma = BB^\top$ with $B \in \Rqqns$ one may write
\[
	\Rqqsympd \ = \ \bigl\{ B \exp(A) B^\top : A \in \Rqqsym \bigr\} .
\]
These parametrizations are particularly useful in connection with first and second order Taylor expansions of smooth functions on $\Rqqsympd$.

\begin{Definition}[Geodesically convex sets]
\label{def:g-convex.sets}
A subset $C$ of $\Rqqsympd$ is called \textsl{geodesically convex} (\textsl{g-convex}) if for arbitrary $\Sigma_0, \Sigma_1 \in C$ the whole geodesic path connecting them is contained in $C$. That means, for $0 \le t \le 1$,
\[
	\Sigma_t^{} := \Sigma_0^{1/2}
	\bigl( \Sigma_0^{-1/2} \Sigma_1^{} \Sigma_0^{-1/2} \bigr)^t
	\Sigma_0^{1/2}
	\ \in \ C .
\]
\noindent
In other words, for arbitrary $B \in \Rqqns$ and $A \in \Rqqsym$ such that both $BB^\top$ and $B\exp(A)B^\top$ belong to $C$,
\[
	B \exp(tA) B^\top \ \in \ C
	\quad\text{for} \ 0 \le t \le 1 .
\]
\end{Definition}

\paragraph{Examples.}
Lemma~\ref{lem:geodesic.curves} implies that for arbitrary $B \in \Rqqns$ the following sets are g-convex:
\[
	\bigl\{ B \exp(tA) B^\top : t \in \mathcal{T}\}
\]
with $A \in \Rqqsym$ and an interval $\mathcal{T} \subset \R$, and
\[
	\bigl\{ B D(e^x) B^\top : x \in \mathcal{X}\}
\]
with a convex set $\mathcal{X} \subset \R^q$. Moreover, for any number $c > 0$, the set
\[
	\{\Sigma \in \Rqqsympd : \det(\Sigma) = c\}
\]
is easily shown to be g-convex.

\paragraph{Geodesic distance.}
The geodesic distance between two matrices $\Sigma_0, \Sigma_1 \in \Rqqsympd$ is defined to be the length of the geodesic path connecting them, i.e.
\[
	d_g(\Sigma_0,\Sigma_1)
	\ := \ \bigl\| \log(\Sigma_0^{-1/2} \Sigma_1^{} \Sigma_0^{-1/2}) \bigr\| .
\]
If, as in \eqref{eq:BDB}, we express $\Sigma_0 = BB^\top$ and $\Sigma_1 = B \exp(D(\gamma)) B^\top$, then
\[
	d_g(\Sigma_0,\Sigma_1)^2 \ = \ \|\gamma\|^2 \ = \ \sum_{i=1}^q \gamma_i^2 .
\]
Obviously $d_g(\Sigma_0,\Sigma_1) \ge 0$ with equality if, and only if, $\Sigma_0^{-1/2} \Sigma_1^{} \Sigma_0^{-1/2} = I_q$ which is equivalent to $\Sigma_0 = \Sigma_1$. The interpretation of $d_g(\Sigma_0,\Sigma_1)$ as the length of a shortest path between $\Sigma_0$ and $\Sigma_1$ implies that $d_g(\cdot,\cdot)$ is a metric on $\Rqqsympd$. As to symmetry, $d_g(\Sigma_1,\Sigma_0) = d_g(\Sigma_0,\Sigma_1)$, because any path $M$ from $\Sigma_0$ and $\Sigma_1$ defines a path $\tilde{M}(t) := M(1 - t)$ from $\Sigma_1$ to $\Sigma_0$ such that $L(\tilde{M}) = L(M)$. As to the triangle inequality, for a third matrix $\Sigma_2 \in \Rqqsympd$ let $M_{01}$ be a shortest path from $\Sigma_0$ to $\Sigma_1$ and let $M_{12}$ be a shortest path from $\Sigma_1$ to $\Sigma_2$. Then
\[
	M(t) \ := \ \begin{cases}
		M_{01}(2t)   & \text{for} \ 0 \le t \le 1/2 \\
		M_{12}(2t-1) & \text{for} \ 1/2 \le t \le 1
	\end{cases}
\]
defines a path from $\Sigma_0$ to $\Sigma_2$ such that $L(M) = L(M_{01}) + L(M_{12})$. Thus $d_g(\Sigma_0,\Sigma_2) \le L(M) = d_g(\Sigma_0,\Sigma_1) + d_g(\Sigma_1,\Sigma_2)$.

Tow additional facts are that
\[
	d_g(B \Sigma_0 B^\top, B \Sigma_1 B^\top)
	\ = \ d_g(\Sigma_0,\Sigma_1)
	\ = \ d_g(\Sigma_0^{-1},\Sigma_1^{-1}) .
\]
The first equality follows from the fact that any path $M$ from $\Sigma_0$ to $\Sigma_1$ gives rise to the path $M_B$ from $B\Sigma_0 B^\top$ to $B\Sigma_1 B^\top$ 
with $L(M_B) = L(M)$. Moreover, one can easily verify that $\tilde{M}(t) := M(t)^{-1}$ defines a path from $\Sigma_0^{-1}$ to $\Sigma_1^{-1}$ with $L(\tilde{M}) = L(M)$.

\paragraph{Matrices with determinant one.}
In connection with scale-invariant functionals, the submanifold
\[
	\M^{(q)} \ := \ \bigl\{ \Sigma \in \Rqqsympd : \det(\Sigma) = 1 \bigr\}
\]
of $\Rqqsympd$ plays a prominent role. Note that any $\Sigma \in \M^{(q)}$ may be represented as $\Sigma = BB^\top$ with $B \in \Rqq$ satisfying $\det(B) = \pm 1$, and then
\[
	\M^{(q)} \ = \ \bigl\{ B \exp(A) B^\top : A \in \W^{(q)} \bigr\}
\]
with the linear subspace
\[
	\W^{(q)} \ := \ \bigl\{ A \in \Rqqsym : \tr(A) = 0 \bigr\}
\]
of $\Rqqsym$.

An arbitrary matrix $\Sigma \in \Rqqsympd$ may be written as $\Sigma = a^a \Gamma$ with $a := q^{-1} \log\det(\Sigma) \in \R$ and $\Gamma := \det(\Sigma)^{-1/q} \Sigma \in \M^{(q)}$. Then indeed
\[
	\min_{G \in \M^{(q)}} d_G(\Sigma,G)
	\ = \ d_G(\Sigma,\Gamma) \ = \ q^{1/2} |a| \ = \ q^{-1/2} |\log \det(\Sigma)| .
\]
This follows from a more general observation: Let $\Sigma_0, \Sigma_1 \in \Rqqsympd$ be written as $\Sigma_j = e^{a_j}\Gamma_j$ with $a_j \in \R$ and $\Gamma_j \in \M^{(q)}$. Then
\[
	\log(\Sigma_0^{-1/2} \Sigma_1^{} \Sigma_0^{-1/2})
	\ = \ (a_1 - a_0) I_q + \log(\Gamma_0^{-1/2} \Gamma_1^{} \Gamma_0^{-1/2}) ,
\]
and it follows from $\langle I_q, \log(\Gamma_0^{-1/2} \Gamma_1^{} \Gamma_0^{-1/2})\rangle = \log \det(\Gamma_0^{-1} \Gamma_1) = 0$ that
\[
	d_g(\Sigma_0, \Sigma_1)^2 \ = \ q (a_1 - a_0)^2 + d_g(\Gamma_0,\Gamma_1)^2 .
\]

\subsection{Geodesically convex functions}
 
\begin{Definition}[Geodesically convex functions]
\label{def:g-convex.functions}
Let $C \subset \Rqqsympd$ be g-convex. A function $f : C \to \R$ is called \textsl{geodesically convex} (\textsl{g-convex}) if for arbitrary matrices $\Sigma_0,\Sigma_1 \in C$ and $0 < t < 1$,
\[
	f(\Sigma_t) \ \le \ (1 - t) f(\Sigma_0) + t f(\Sigma_1) ,
\]
where $\Sigma_t$ is defined as in Definition~\ref{def:g-convex.sets}. If the preceding inequality is strict whenever $\Sigma_0 \ne \Sigma_1$, the function $f$ is called \textsl{strictly geodesically convex} (\textsl{strictly g-convex}).

\noindent
Equivalently, $f : C \to \R$ is (strictly) g-convex if for arbitrary $B \in \Rqqns$ and $A \in \Rqqsym \setminus \{0\}$ such that both $BB^\top$ and $B \exp(A) B^\top$ belong to $C$,
\[
	f(B \exp(tA) B^\top) \ \text{is (strictly) convex in} \ t \in [0,1] . 
\]
\end{Definition}

\begin{Example}
\label{ex0}
The function $f(\Sigma) := \log \det(\Sigma)$ is geodesically convex on $\Rqqsympd$. It is even geodesically linear in the sense that
\[
	f(B \exp(A) B^\top)
	\ = \ f(BB^\top) + \tr(A)
	\ = \ f(BB^\top) + \langle I_q,A\rangle
\]
for arbitrary $B \in \Rqqns$ and $A \in \Rqqsym$.
\end{Example}

By means of Lemma~\ref{lem:geodesic.curves} one can easily derive the following result.

\begin{Lemma}
\label{lem:criteria.g-convexity}
For a function $f : \Rqqsympd \to \R$ the following three properties are equivalent:

\noindent
\textbf{(a)} \ $f$ is (strictly) geodesically convex;

\noindent
\textbf{(b)} \ For arbitrary $B \in \Rqqns$ and $A \in \Rqqsym \setminus \{0\}$, the function
\[
	\R \ni t \ \mapsto \ f(B \exp(tA) B^\top)
\]
is (strictly) convex;

\noindent
\textbf{(b')} \ For arbitrary $B \in \Rqqns$ and $x \in \R^q \setminus \{0\}$, the function
\[
	\R \ni t \ \mapsto \ f(B D(e^{tx}) B^\top)
\]
is (strictly) convex;

\noindent
\textbf{(c)} \ For arbitrary $B \in \Rqqns$, the function
\[
	\R^q \ni x \ \mapsto \ f(B D(e^x) B^\top)
\]
is (strictly) convex.
\end{Lemma}

Obviously, Property~(b') is a special case of Property~(b), because $D(e^{tx}) = \exp(t D(x))$. On the other hand we may write $A \in \Rqqsym \setminus \{0\}$ as $A = U D(x) U^\top$ for some $U \in \Rqqorth$ and $x \in \R^q \setminus \{0\}$. Then $B \exp(tA) B^\top = (BU) D(e^{tx}) (BU)^\top$, whence Property~(b') implies Property~(b).

\begin{Example}
\label{ex1}
For any vector $v \in \R^q \setminus \{0\}$, the function
\[
	\Sigma \ \mapsto \ v^\top \Sigma v
\]
is g-convex, and the function
\[
	\Sigma \ \mapsto \ \tr(\Sigma)
\]
is strictly g-convex. To verify these claims we use criterion (c) in Lemma~\ref{lem:criteria.g-convexity}: For $B \in \Rqqns$ and $x \in \R^q$,
\[
	v^\top B D(e^x) B^\top v
	\ = \ \sum_{i=1}^q e^{x_i} (B^\top v)_i^2
\]
is obviously convex in $x$, because $\exp : \R \to \R$ is convex. Similarly,
\[
	\tr(B D(e^x) B^\top)
	\ = \ \sum_{j=1}^q e^{x_j} w_j
	\quad\text{with} \ w_j = \sum_{i=1}^q B_{ij}^2 .
\]
This is even strictly convex in $x$, because $\exp : \R \to \R$ is strictly convex and all weights $w_j$ are strictly positive.
\end{Example}

\begin{Example}
\label{exlog}
For any vector $v \in \R^q \setminus \{0\}$, the function
\[
	\Sigma \ \mapsto \ \log (v^\top \Sigma v)
\]
is g-convex. To verify this claim we use criterion (b') in Lemma~\ref{lem:criteria.g-convexity}: For $B \in \Rqqns$, $x \in \R^q \setminus \{0\}$,
and $t \in \R$,
\[
	g(t) \ = \ \log\left( v^\top B D(e^{tx}) B^\top v \right)
	\ = \ \log \sum_{i=1}^q e^{t x_i} a_i.
\]
with $a_i = (B^\top v)_i^2 \ge 0$. Evaluating its second derivative gives
\[
	g^{\prime\prime}(t)
	\ = \ \frac{\sum_{i=1}^q e^{t x_i} a_i x_i^2}{\sum_{i=1}^q e^{t x_i} a_i}
	- \biggl\{ \frac{\sum_{i=1}^q e^{t x_i} a_i x_i}{\sum_{i=1}^q e^{t x_i} a_i} \biggr\}^2 ,
\]
and so by application of the Cauchy Schwartz inequality $g^{\prime\prime}(t) \ge 0 $, with equality if and only if all the $x_i$'s are equal for those $i$ for which
$a_i > 0$.

Furthermore, suppose that $\rho : \R_+ \to \R$ is g-convex, which is equivalent to $h(t) := \rho(e^t)$ being convex in $t \in \R$, and that $\rho$ is non-decreasing. Then the function
\[
	\Sigma \ \mapsto \ \rho(v^\top \Sigma^{-1} v)
\]
is g-convex. This follows by expressing $\rho(v^\top \Sigma^{-1} v) = h(f(\Sigma^{-1})$ with $f(\Sigma) := \log(v^\top \Sigma v)$ and then applying the two remarks given below.
\end{Example}

\begin{Remark}[G-convexity and inversion]
\label{rem:Inversion}
If $f : \Rqqsympd \to \R$ is geodesically convex, then $\tilde{f}(\Sigma) := f(\Sigma^{-1})$ defines a geodesically convex function, too. This follows essentially from the fact that
\[
	\bigl( B \exp(tA) B^\top \bigr)^{-1}
	\ = \ B^{-\top} \exp(- tA) B^{-1}
	\ = \ \tilde{B} \exp(t \tilde{A}) \tilde{B}^\top
\]
with $\tilde{B} := B^{-\top}$ and $\tilde{A} := -A$.
\end{Remark}

\begin{Remark}[G-convexity and compositions]
\label{rem:Transformation}
Let $f : \Rqqsympd \to \R$ be geodesically convex with values in an interval $\mathcal{T} \subset \R$, and let $h : \mathcal{T} \to \R$ be convex and non-decreasing. Then $\tilde{f}(\Sigma) := h(f(\Sigma))$ defines a geodesically convex function, too. For if $\Sigma_0, \Sigma_1, \Sigma_t$ as in Definition~\ref{def:g-convex.sets}, then
\begin{align*}
	\tilde{f}(\Sigma_t) \
	&= \ h(f(\Sigma_t)) \\
	&\le \ h \bigl( (1 - t)f(\Sigma_0) + t f(\Sigma_1) \bigr)
		\qquad(\text{g-convexity of} \ f, \ \text{monotonicity of} \ h) \\
	&\le \ (1 - t) h(f(\Sigma_0)) + t h(f(\Sigma_1))
		\qquad(\text{convexity of} \ h) \\
	&= \ (1 - t) \tilde{f}(\Sigma_0) + \tilde{f}(\Sigma_1) .
\end{align*}
The function $\tilde{f}$ is even strictly g-convex if $f$ is strictly g-convex and $h$ is strictly increasing.
\end{Remark}

\subsection{Minimizers and geodesic coercivity}
\label{subsec:Minimizers.g-coercivity}

Suppose we want to minimize a g-convex function $f : \Rqqsympd \to \R$. As in classical convex analysis, a minimizer of $f$ may be characterized by means of the one-sided directional derivatives
\[
	\lim_{t \to 0\,+} \frac{f(B \exp(tA) B^\top) - f(BB^\top)}{t}
\]
for $B \in \Rqqns$ and $A \in \Rqqsym$. The latter limit exists in $\R$, because g-convexity of $f$ implies convexity of $f(B \exp(tA) B^\top)$ in $t \in \R$.

\begin{Lemma}[Characterizing minimizers]
\label{lem:minimizer.g-convex}
A matrix $\Sigma = BB^\top$ with $B \in \Rqqns$ minimizes a g-convex function $f  : \Rqqsympd \to \R$ if, and only if,
\begin{equation}
\label{eq:minimizer.g-convex}
	\lim_{t \to 0\,+} \frac{f(B \exp(tA) B^\top) - f(BB^\top)}{t}
	\ \ge \ 0 \quad\text{for all} \ A \in \Rqqsym .
\end{equation}
\end{Lemma}

This lemma provides an explicit criterion to check whether a certain point $\Sigma$ is a minimizer of a differentiable and g-convex function on $\Rqqsympd$. But it is not clear under what conditions a minimizer has to exist. In this context a key property of $f$ is coercivity in the following sense.

\begin{Definition}[Geodesic coercivity]
A function $f : \Rqqsympd \to \R$ is called \textsl{geodesically coercive} (\textsl{g-coercive}) if
\[
	f(\Sigma) \ \to \ \infty
	\quad\text{as} \ \|\log(\Sigma)\| \to \infty .
\]

\noindent
In other words, a function $f : \Rqqsympd \to \R$ is g-coercive if, and only if, the function $\Rqqsym \ni A \mapsto f(\exp(A))$ is coercive in the usual sense, that is, $f(\exp(A)) \to \infty$ as $\|A\| \to \infty$.
\end{Definition}

Note that $\|\log(\Sigma)\| \to \infty$ is equivalent to $\|\Sigma\| + \|\Sigma^{-1}\| \to \infty$. Various authors have realized that any continuous function $f$ on $\Rqqsympd$ with the latter property has a compact set of minimizers, e.g.\ \cite{Sra_Hosseini_2015}. The following lemma and its corollary explain the relation between g-coercivity and the existence of minimizers in case of g-convex functions. In particular, the corollary shows that a continuous and strictly g-convex function has a unique minimizer if, and only if, it is g-coercive.

\begin{Lemma}[Existence of minimizers]
\label{lem:existence.minimizers}
Let $f : \Rqqsympd \to \R$ be a continuous and geodesically convex function.

\noindent
\textbf{(i)} \ The set $\mathcal{S}_*$ of its minimizers is a closed and geodesically convex subset of $\Rqqsympd$. It is possibly empty.

\noindent
\textbf{(ii)} \ If $f$ is g-coercive, then $\mathcal{S}_*$ is nonvoid and compact.

\noindent
\textbf{(iii)} \ If $f$ fails to be g-coercive but $\mathcal{S}_*$ is nonvoid, then $\mathcal{S}_*$ is geodesically unbounded, that means,
\[
	\sup_{\Sigma_1, \Sigma_2 \in \mathcal{S}_*} d_g(\Sigma_1,\Sigma_2) \ = \ \infty .
\]
\end{Lemma}

\begin{Corollary}[Existence of unique of minimizers]
\label{cor:uniqueness.minimizer}
Let $f : \Rqqsympd \to \R$ be a continuous and strictly geodesically convex function.

\noindent
\textbf{(i)} \ If $f$ is g-coercive, it has a unique minimizer.

\noindent
\textbf{(ii)} \ If $f$ fails to be g-coercive, it has no minimizer at all.
\end{Corollary}

Corollary~\ref{cor:uniqueness.minimizer} follows easily from Lemma~\ref{lem:existence.minimizers}. Note that a strictly g-convex function $f$ can have at most one minimizer. For if $\Sigma_0, \Sigma_1$ are two different matrices with $f(\Sigma_0) = f(\Sigma_1)$, then $f$ attains strictly smaller values along the geodesic path connecting $\Sigma_0$ and $\Sigma_1$. Since a geodesically unbounded set is necessarily infinite, a continuous and strictly g-convex function which is not g-coercive cannot have a minimizer.

The next lemma provides an equivalent characterization for g-coercivity:

\begin{Lemma}[Characterizing g-coercivity]
\label{lem:g-coercivity}
Let $f : \Rqqsympd \to \R$ be continuous and geodesically convex. Then $f$ is geodesically coercive if, and only if, for any fixed $A \in \Rqqsym \setminus \{0\}$,
\[
	\lim_{t \to \infty} \ \lim_{u \to t\,+} \frac{f(\exp(uA)) - f(\exp(tA))}{u - t}
	\ > \ 0 .
\]
\end{Lemma}

\subsection{Differentiability}

The next lemma establishes a connection between differentiability in the usual sense and differentiability with respect to local geodesic coordinates.

\begin{Lemma}[1st order smoothness]
\label{lem:smoothness1}
For a function $f : \Rqqsympd \to \R$ the following two conditions are equivalent:

\noindent
\textbf{(S1.i)} \ $f$ is differentiable with gradient $\nabla f : \Rqqsympd \to \Rqqsym$.

\noindent
\textbf{(S1.ii)} \ For each $B \in \Rqqns$ there exists a matrix $G(B) \in \Rqqsym$ such that for $A \in \Rqqsym$,
\[
	f(B \exp(A) B^\top) \ = \ f(BB^\top) + \langle A, G(B)\rangle + o(\|A\|)
	\quad\text{as} \ A \to 0 .
\]

\noindent
In case of (S1.i-ii),
\begin{align*}
	G(B) \
	&= \ B^\top \nabla f(BB^\top) B
		\quad\text{for} \ B \in \Rqqns , \\
	\nabla f(\Sigma) \
	&= \ \Sigma^{-1/2} G(\Sigma^{1/2}) \Sigma^{-1/2}
		\quad\text{for} \ \Sigma \in \Rqqsympd .
\end{align*}
\end{Lemma}

In particular, a function $f : \Rqqsympd \to \R$ is continuously differentiable if, and only if, its ``geodesic gradient (g-gradient)'' $G(B)$ is continuous in $B \in\Rqqns$.

It is well-known from convex analysis that a differentiable convex function $f$ on $\R^d$ is minimal at a certain point $x \in \R^d$ if, and only if, $\nabla f(x) = 0$. The same is true for differentiable g-convex functions:

\begin{Corollary}[Characterizing minimizers]
\label{cor:minimizers.g-convex}
Let $f : \Rqqsympd \to \R$ be differentiable and geodesically convex. Then for $\Sigma = BB^\top$, $B \in \Rqqns$, the following three conditions are equivalent:

\noindent
\textbf{(a)} \ $\Sigma$ is a minimizer of $f$;

\noindent
\textbf{(b)} \ $\nabla f(\Sigma) = 0$;

\noindent
\textbf{(b')} \ $G(B) = 0$.
\end{Corollary}

This corollary follows directly from Lemmas~\ref{lem:minimizer.g-convex} and \ref{lem:smoothness1}, noting that
\[
	\lim_{t \to 0\,+} \frac{f(B\exp(tA)B^\top) - f(BB^\top)}{t}
	\ = \ \langle A, G(B)\rangle
\]
for $A \in \Rqqsym$ and $B \in \Rqqns$. Moreover, for different real numbers $t, u$ and $B := \exp((t/2)A)$,
\[
	\frac{f(\exp(uA)) - f(\exp(tA))}{u - t}
	\ = \ \frac{f \bigl( B\exp((u - t)A) B \bigr) - f(BB)}{u - t}
	\ \to \ \langle A, G(B) \rangle
\]
as $u \to t$. Hence for differentiable and g-convex functions $f$ the criterion for g-coercivity in Lemma~\ref{lem:g-coercivity} can be reformulated as follows:

\begin{Corollary}[Characterizing g-coercivity]
\label{cor:g-coercivity}
Let $f : \Rqqsympd \to \R$ be differentiable and geodesically convex. Then $f$ is geodesically coercive if, and only if, for any fixed $A \in \Rqqsym \setminus \{0\}$,
\[
	\lim_{t \to \infty} \ \frac{d}{dt} f(\exp(tA)) \ > \ 0
\]
which is equivalent to
\[
	\lim_{t \to \infty} \langle A, G(\exp(tA)) \ > \ 0 .
\]
\end{Corollary}

\subsection{Second order smoothness}

Verifying g-convexity of a function $f$ on $\Rqqsympd$ is not trivial. Many authors use direct calculations case by case \cite{Wiesel_2012} or use advanced matrix inequalities \cite{Sra_Hosseini_2013, Sra_Hosseini_2015}. Convexity of functions can be easily characterized in terms of second derivatives. The same is true for g-convexity if one uses local geodesic coordinates.

\begin{Lemma}[Conditions for g-convexity]
\label{lem:g-convexity}
Let $f : \Rqqsympd \to \R$ satisfy the following condition: For each $B \in \Rqqns$ there exist a matrix $G(B) \in \Rqqsym$ and a quadratic form $H(\cdot,B)$ on $\Rqqsym$ such that for $A \in \Rqqsym$,
\begin{equation}
\label{eq:Taylor2}
	f(B \exp(A) B^\top) \ = \ f(BB^\top)
		+ \langle A, G(B)\rangle + 2^{-1} H(A,B) + o(\|A\|^2)
	\quad\text{as} \ A \to 0 .
\end{equation}
Then the function $f$ is geodesically convex if, and only if,
\begin{equation}
\label{eq:H.pos.semidefinite}
	H(A,B) \ \ge \ 0
	\quad\text{for all} \ B \in \Rqqns \ \text{and} \ A \in \Rqqsym .
\end{equation}
It is strictly geodesically convex if
\begin{equation}
\label{eq:H.pos.definite}
	H(A,B) \ > \ 0
	\quad\text{for all} \ B \in \Rqqns \ \text{and} \ A \in \Rqqsym \setminus \{0\} .
\end{equation}
\end{Lemma}

\begin{Example}
\label{ex2}
The function $\Sigma \mapsto \log \tr(\Sigma)$ is geodesically convex. For if $B \in \Rqqns$ and $A \in \Rqqsym$, then
\begin{align*}
	\log \tr(B \exp(A) B^\top) \
	= \ &\log \bigl( \tr(BB^\top) + \tr(BAB^\top) + 2^{-1} \tr(B A^2 B^\top)
		+ O(\|A\|^3) \bigr) \\
	= \ &\log \tr(\Sigma)
		+ \log \Bigl( 1 + \frac{\tr(BAB^\top)}{\tr(BB^\top)}
			+ 2^{-1} \frac{\tr(B A^2 B^\top)}{\tr(BB^\top)}
			+ O(\|A\|^3) \Bigr) \\
	= \ &\log \tr(\Sigma)
		+ \frac{\tr(BAB^\top)}{\tr(BB^\top)} \\
		&+ \ 2^{-1} \Bigl( \frac{\tr(B A^2 B^\top)}{\tr(BB^\top)}
			- \Bigl( \frac{\tr(BAB^\top)}{\tr(BB^\top)} \Bigr)^2 \Bigr)
			+ O(\|A\|^3) \Bigr)
\end{align*}
as $A \to 0$, so
\begin{align*}
	G(B) \
	&= \ \tr(BB^\top)^{-1} B^\top B
		\ = \ \|B\|^{-2} B^\top B , \\
	H(A,B) \
	&= \ \frac{\tr(B A^2 B^\top)}{\tr(BB^\top)}
			- \Bigl( \frac{\tr(BAB^\top)}{\tr(BB^\top)} \Bigr)^2
		\ = \ \langle A^2, G(B)\rangle - \langle A, G(B)\rangle^2 .
\end{align*}
Obviously, $H(I_q,B) = 0$. But $H(A,B) > 0$ for all $A \not\in \{t I_q : t \in \R\}$. To show this let $A = U D(x) U^\top$ with $U \in \Rqqorth$ and $x \in \R^q$. Then for any integer $s \ge 0$,
\[
	\tr(B A^s B^\top)
	\ = \ \tr(BU D(x^s) (BU)^\top)
	\ = \ \sum_{j=1}^q w_j x_j^s
\]
with $w_j := \sum_{i=1}^q (BU)_{ij}^2 > 0$. Consequently,
\[
	H(A,B)
	\ = \ \frac{\sum_{j=1}^q w_j x_j^2}{\sum_{j=1}^q w_j}
	- \biggl\{ \frac{\sum_{j=1}^q w_j x_j}{\sum_{j=1}^q w_j} \biggr\}^2
	\ > \ 0
\]
unless $x_1 = x_2 = \cdots = x_q$. But the latter condition would be equivalent to $A$ being a multiple of the identity matrix.
\end{Example}

\begin{Remark}[Smoothness and inversion]
\label{rem:Inversion2}
Suppose that $f : \Rqqsympd \to \R$ satisfies the second order smoothness assumption in Lemma~\ref{lem:g-convexity}. Then $\tilde{f}(\Sigma) := f(\Sigma^{-1})$ satisfies this assumption, too: For any $B \in \Rqqns$, as $\Rqqsym \ni A \to 0$,
\[
	\tilde{f} \bigl( B \exp(A) B^\top \bigr)
	\ = \ \langle A, \tilde{G}(B)\rangle
		+ 2^{-2} \tilde{H}(A, B) + o(\|A\|^2)
\]
with
\begin{align*}
	\tilde{G}(B)   \ &:= \ - G(B^{-\top}) , \\
	\tilde{H}(A,B) \ &:= \ H(A,B^{-\top}) .
\end{align*}
\end{Remark}

\begin{Remark}[Smoothness and exponential or power transformations]
\label{rem:Transformation2}
Suppose that a function $f : \Rqqsympd \to \R$ satisfies the second order smoothness assumption in Lemma~\ref{lem:g-convexity}. For $c > 0$ let
\[
	f_c(\Sigma) \ := \ \exp(c f(\Sigma))/c .
\]
Then for any $B \in \Rqqns$, as $\Rqqsym \ni A \to 0$,
\[
	f_c(B \exp(A) B^\top)
	\ = \ f_c(BB^\top)
		+ \langle A, G_c(B)\rangle
		+ 2^{-1} H_c(A, B) + o(\|A\|^2)
\]
with
\begin{align*}
	G_c(B) \
	&:= \ \exp(c f(BB^\top)) G(B) , \\
	H_c(A,B) \
	&:= \ \exp(c f(BB^\top))
		\bigl( H(A,B) + c \langle A, G(B)\rangle^2 \bigr) .
\end{align*}
Similarly, if $f > 0$ and
\[
	f_\gamma(\Sigma) \ = \ f(\Sigma)^\gamma/\gamma
\]
for $\gamma > 1$, then
\[
	f_\gamma(B \exp(A) B^\top)
	\ = \ f_\gamma(BB^\top)
		+ \langle A, G_\gamma(B)\rangle
		+ 2^{-1} H_\gamma(A, B) + o(\|A\|^2)
\]
with
\begin{align*}
	G_\gamma(B) \
	&:= \ f(BB^\top)^{\gamma-1} G(B) , \\
	H_\gamma(A,B) \
	&:= \ f(BB^\top)^{\gamma-1} H(A,B)
		+ (\gamma - 1) f(BB^\top)^{\gamma-2} \langle A, G(B)\rangle^2 .
\end{align*}
\end{Remark}

\begin{Remark}[Orthogonal transformations]
\label{rem:Orthogonal transformations}
For matrices $B, \tilde{B} \in \Rqqns$, the equation $BB^\top = \tilde{B}\tilde{B}^\top$ is equivalent to $\tilde{B} = BU$ for some orthogonal matrix $U \in \Rqq$. For any function $f : \Rqqsympd \to \R$ satisfying the second order smoothness assumption in Lemma~\ref{lem:g-convexity},
\begin{align*}
	G(BU) \ &= \ U^\top G(B) U
		\quad\text{and} \\
	H(A,BU) \ &= \ H(UAU^\top, B) \quad\text{for} \ A \in \Rqqsym .
\end{align*}
In particular, neither the eigenvalues of $G(B)$ nor the set $\bigl\{ H(A,B) : A \in \Rqqsym, \|A\| = 1 \bigr\}$ change when $B$ is replaced with $BU$.

The equations for $G(BU)$ and $H(\cdot,BU)$ follow from the fact that $BU \exp(A) (BU)^\top = B \exp(UAU^\top) B^\top$. Thus
\[
	f \bigl( BU \exp(A) (BU)^\top \bigr)
	\ = \ f(BB^\top) + \langle A, G(BU)\rangle
		+ 2^{-1} H(A, BU) + o(\|A\|^2)
\]
coincides with
\begin{align*}
	f \bigl( B \exp(UAU^\top) B^\top \bigr) \
	&= \ f(BB^\top) + \langle UAU^\top, G(B)\rangle
		+ 2^{-1} H(UAU^\top, B) + o(\|A\|^2) \\
	&= \ f(BB^\top) + \langle A, U^\top G(B)U\rangle
		+ 2^{-1} H(UAU^\top, B) + o(\|A\|^2) .
\end{align*}
\end{Remark}

As explained in Supplement~\ref{sec:Auxiliary}, existence of second order Taylor expansions alone does not imply twice differentiability. But this is true under an additional continuity requirement on the quadratic terms.

\begin{Lemma}[2nd order smoothness]
\label{lem:smoothness2}
For a function $f : \Rqqsympd \to \R$ the following two conditions are equivalent:

\noindent
\textbf{(S2.i)} \ $f$ is twice continuously differentiable with gradient $\nabla f(\Sigma) \in \Rqqsym$ and Hessian operator $D^2f(\Sigma) : \Rqqsym \to \Rqqsym$ at $\Sigma \in \Rqqsympd$.

\noindent
\textbf{(S2.ii)} \ For each $B \in \Rqqns$ there exist a matrix $G(B) \in \Rqqsym$ and a quadratic form $H(\cdot,B)$ on $\Rqqsym$ such that expansion \eqref{eq:Taylor2} is valid. Moreover, $H(A,B)$ is continuous in $B \in \Rqqsym$ for any fixed $A \in \Rqqsym$.

\noindent
In case of (S2.i-ii), for $A \in \Rqqsym$,
\begin{align*}
	H(A,B) \
	&= \ \langle A^2, G(B) \rangle + \langle BAB^\top, D^2f(BB^\top) BAB^\top\rangle
		\quad\text{for} \ B \in \Rqqns , \\
	\langle A, D^2f(\Sigma) A\rangle \
	&= \ H(\Sigma^{-1/2}A\Sigma^{-1/2}, \Sigma^{1/2})
		- \langle A\Sigma^{-1}A, \nabla f(\Sigma)\rangle
		\quad\text{for} \ \Sigma \in \Rqqsympd .
\end{align*}
\end{Lemma}

\subsection{Scale-invariant functions}
\label{subsec:Scale-invariance}

Sometimes we consider \textsl{scale-invariant} functions $f : \Rqqsympd \to \R$ in the sense that
\[
	f(c\Sigma) \ = \ f(\Sigma)
	\quad\text{for arbitrary} \ \Sigma \in \Rqqsympd \ \text{and} \ c > 0 .
\]
If the function $f$ is differentiable, this property is equivalent to the following condition on its g-gradients $G(B)$:
\[
	\tr(G(B)) \ = \ 0
	\quad\text{for all} \ B \in \Rqqns .
\]
This follows essentially from the fact that for $t \in \R$,
\begin{align*}
	f(e^tBB^\top) \ = \ f(B \exp(t I_q) B^\top) \
	&= \ f(BB^\top) + \langle t I_q, G(B)\rangle + o(t) \\
	&= \ f(BB^\top) + \tr(G(B)) t + o(t)
\end{align*}
as $t \to 0$. If $f$ does even satisfy the second order smoothness assumption in Lemma~\ref{lem:g-convexity}, then
\[
	H(I_q,B) \ = \ 0
	\quad\text{for all} \ B \in \Rqqns ,
\]
because
\[
	f(e^tBB^\top)
	\ = \ f(BB^\top) + \tr(G(B)) t + H(I_q,B) t^2/2 + o(t^2) .
\]

A scale-invariant function $f$ on $\Rqqsympd$ is geodesically convex if, and only if, $f$ is geodesically convex on the g-convex submanifold $\M^{(q)} = \{\Sigma \in \Rqqsympd : \det(\Sigma) = 1\}$ introduced earlier. For if $\Sigma_t = B \exp(tA) B^\top$ for $t \in \R$ with arbitrary $B \in \Rqqns$ and $A \in \Rqqsym$, then $\det(\Sigma_t) = \det(B)^2 \exp(t \tr(A))$, and
\[
	f(\Sigma_t)
	\ = \ f \bigl( \det(\Sigma_t)^{-1/q} \Sigma_t \bigr)
	\ = \ f(B_o^{} \exp(t A_o) B_o^\top)
\]
with $B_o := |\det(B)|^{-1/q} B$ satisfying $\det(B_o) = \pm 1$ and $A_o := A - (\tr(A)/q) I_q$ belonging to the subspace $\W^{(q)}$ of symmetric matrices with trace $0$.

To minimize a scale-invariant function $f$, one may restrict one's attention to matrices in $\M^{(q)}$. Then the previous considerations can be adapted as follows:

\paragraph{A criterion for strict g-convexity.}
Suppose that $f : \Rqqsympd \to \R$ is scale-invariant and satisfies the second order smoothness assumption of Lemma~\ref{lem:g-convexity}. Then it is strictly geodesically convex on $\M^{(q)}$ if $H(A,B) > 0$ for all $B \in \Rqqns$ and $A \in \W^{(q)} \setminus \{0\}$.

\paragraph{Minimizers and g-coercivity.}
All results of Section~\ref{subsec:Minimizers.g-coercivity} carry over with the following modifications: We restrict our attention to matrices $\Sigma \in \M^{(q)}$, to matrices $B \in \Rqqns$ with $\det(B) = \pm 1$ and to matrices $A \in \W^{(q)}$. In particular, a matrix $\Sigma = BB^\top \in \M^{(q)}$ minimizes a g-convex function $f$ on $\M^{(q)}$ if, and only if,
\[
	\lim_{t \to 0\,+} \frac{f(B\exp(tA) B^\top) - f(BB^\top)}{t}
	\ \ge \ 0
	\quad\text{for all} \ A \in \W^{(q)} .
\]
A function $f$ is said to be geodesically coercive on $\M^{(q)}$ if
\[
	f(\exp(A)) \ \to \ \infty
	\quad\text{as} \ \|A\| \to \infty, A \in \W^{(q)} .
\]
In case of a continuous and g-convex function $f$, a necessary and sufficient condition for this is
\[
	\lim_{t \to \infty} \ \lim_{u \to t\,+} \frac{f(\exp(uA)) - f(\exp(tA))}{u - t}
	\ > \ 0
	\quad\text{whenever} \ A \in \W^{(q)} \setminus \{0\} .
\]

\section{Regularized $M$-estimators of scatter}
\label{sec:Regularized.scatter}

\subsection{Scatter functionals}
\label{subsec:M-scatter}

We now apply the results of the previous section to the problem of regularized $M$-functionals and $M$-estimators of scatter. Before doing so, we first briefly consider the non-penalized case, i.e.\ minimizing 
\[
	L_\rho(\Sigma,Q)
	\ := \ \int \bigl[ \rho(x^\top\Sigma^{-1}x) - \rho(\|x\|^2) \bigr] \, Q(dx)
		+ \log \det(\Sigma) .
\]
In what follows we summarize various results from \cite{Zhang_etal_2013} and \cite{Duembgen_etal_2015} in a slightly more general setting. The former paper considered only empirical distributions $Q = Q_n$ whereas the latter survey paper considered general distributions $Q$ but only differentiable functions $\rho$ satisfying additional constraints.

Throughout we assume that $\rho(s)$ is non-decreasing and g-convex in $s > 0$, that means, $h(x) := \rho(e^x)$ is non-decreasing and convex in $x \in \R$. In particular, $\rho$ is continuous with left- and right-sided derivatives on $\R_+$, and
\[
	\psi(s) \ := \ \begin{cases}
		0 & \text{if} \ s = 0 , \\
		s \rho'(s\,+) & \text{if} \ s > 0
	\end{cases}
\]
defines a non-decreasing function on $[0,\infty)$. Note that $\psi(e^x) = h'(x\,+)$ for $x \in \R$. Thus strict g-convexity of $\rho$ on $\R_+$ is equivalent to $\psi$ being strictly increasing on $[0,\infty)$.

The next proposition clarifies under which conditions on $\rho$ and $Q$ the objective function $L_\rho(\Sigma,Q)$ is well-defined for arbitrary $\Sigma \in \Rqqsympd$. In particular, a sufficient condition for that is $\psi(\infty) < \infty$ or $Q$ having bounded support.

\begin{Proposition}
\label{prop:existence}
The integral $\int \bigl| \rho(x^\top\Sigma^{-1}x) - \rho(\|x\|^2) \bigr| \, Q(dx)$ is finite for arbitrary $\Sigma \in \Rqqsympd$ if, and only if,
\begin{equation}
\label{eq:existence}
	\int \psi(\lambda \|x\|^2) \, Q(dx) \ < \ \infty
	\quad\text{for arbitrary} \ \lambda \ge 1 .
\end{equation}
In case of $\rho'(s\,+)$ being non-increasing in $s > 0$, the latter condition is equivalent to
\[
	\int \psi(\|x\|^2) \, Q(dx) \ < \ \infty .
\]
\end{Proposition}

The following theorem regarding the g-convexity of $L_\rho(\Sigma,Q)$ follows essentially from examples \ref{ex0} and \ref{exlog} plus some extra arguments, see Supplement~\ref{sec:Auxiliary}. It is an extension of Theorem~1(a) of \cite{Zhang_etal_2013}, who considered the case $Q = Q_n$, and of Proposition~5.4 of \cite{Duembgen_etal_2015}, who considered differentiable functions $\rho$:

\begin{Theorem}
\label{thm:Mfunc}
Under Condition~\eqref{eq:existence}, $L_\rho(\Sigma,Q)$ is continuous and geodesically convex in $\Sigma \in \Rqqsympd$. Furthermore,

\noindent
\textbf{(a)} suppose that $\rho(s)$ is strictly g-convex in $s > 0$. Then $L_\rho(\cdot,Q)$ is strictly geodesically convex if, and only if,
\[
	Q(\V) \ < \ 1
\]
for any linear subspace $\V$ of $\R^q$ with $\dim(\V) < q$.

\noindent
\textbf{(b)} suppose that $\rho(s) = q \log s$ for $s > 0$. Then $L_\rho(\cdot,Q)$ is strictly geodesically convex on $\M^{(q)}$ if, and only if,
\[
	Q(\V \cup \W) \ < \ 1
\]
for arbitrary linear subspaces $\V, \W \subsetneq \R^q$ with $\V \cap \W = \{0\}$.
\end{Theorem}

The special function $\rho(s) = q \log s$ in part~(b) corresponds to the distribution-free $M$-estimator of scatter introduced in \cite{Tyler_1987a}, and it is the setting for which geodesic convexity was first applied to $M$-estimation \cite{Auderset_etal_2005,Wiesel_2012}. The corresponding objective function $L_\rho(\cdot,Q)$ is scale-invariant if $Q(\{0\}) = 0$.

Results on the g-coercivity of $L_\rho(\Sigma,Q)$ can be obtained by extending Lemma 2.2 of \cite{Kent_Tyler_1991} from $Q_n$ to general $Q$, see also Theorem~1(b) of \cite{Zhang_etal_2013} and Proposition~5.5 of \cite{Duembgen_etal_2015}. Lemma~\ref{lem:g-coercivity} allows for a complete answer in the present general framework, starting from the following proposition.

\begin{Proposition}
\label{prop:g-coercivity}
Let $A = U D(-\gamma) U^\top$ with $U = [u_1, u_2, \ldots, u_q] \in \Rqqorth$ and $\gamma \in \R^q$ satisfying $\gamma_1 \le \gamma_2 \le \cdots \le \gamma_q$. Then
\begin{align}
\nonumber
	\lim_{t \to \infty} \ &\lim_{u \to t\,+} \,
		\frac{L_\rho(\exp(uA),Q) - L_\rho(\exp(tA),Q)}{u - t} \\
\label{eq:g-coercivity}
	&= \ \sum_{j=1}^q Q(\V_j\setminus\V_{j-1})
		\bigl( \psi(\infty) \gamma_j^+ - \psi(0\,+) \gamma_j^- \bigr)
		- \sum_{j=1}^q \gamma_j ,
\end{align}
where $\V_0 := \{0\}$ and $\V_j = \mathrm{span}(u_1,u_2,\ldots,u_j)$ for $1 \le j \le q$. Furthermore, $a^{\pm} := \max\{\pm a, 0\}$ for $a \in \R$.

\noindent
\textbf{(a)} \ Specifically let $\psi(0\,+) = 0 < \psi(\infty)$. Then the previous limit may be rewritten as
\[
	\sum_{k=0}^{q-1} \bigl( (1 - Q(\V_k)) \psi(\infty) - q + k \bigr)
		(\gamma_{k+1}^+ - \gamma_k^+)
	+ \sum_{j=1}^q \gamma_j^- .
\]

\noindent
\textbf{(b)} \ Specifically let $\rho(s) := q \log s$ for $s > 0$. Then $\psi \equiv q$ on $\R_+$, and the previous limit may be rewritten as
\[
	q \sum_{k=1}^{q-1} (k/q - Q(\V_k)) (\gamma_{k+1} - \gamma_k)
	- q Q(\{0\}) \gamma_1 .
\]
\end{Proposition}

This proposition will be used later in connection with regularized scatter functionals. In the present context it implies necessary and sufficient conditions for g-coercivity in the following two settings:

\bigskip

\noindent
\textbf{Setting 0.} \ $\rho(s) = q \log s$ for $s > 0$, and $Q(\{0\}) = 0$.

\noindent
\textbf{Setting 1.} \ $\psi(0\,+) = 0$, $q < \psi(\infty) \le \infty$, and $Q$ satisfies \eqref{eq:existence}.

\bigskip

\begin{Theorem}\strut
\label{thm:g-coercivity}

\noindent
\textbf{(a)} \ In Setting~1, $L_\rho(\cdot,Q)$ is geodesically coercive if, and only if,
\begin{equation}
\label{eq:g-coercivity1}
	Q(\V) \ < \ 1 - \frac{\{q - \dim(\V)\}}{\psi(\infty)}
\end{equation}
for all linear subspaces $\V \subset \R^q$ with $0 \le \dim(\V) < q$. If in addition $\psi$ is strictly increasing on $\{s \ge 0: \psi(s) < \psi(\infty)\}$, then $L_\rho(\cdot,Q)$ has a unique minimizer.

\noindent
\textbf{(b)} \ In Setting~0, $L_\rho(\cdot,Q)$ is geodesically coercive on $\M^{(q)}$ if, and only if,
\begin{equation}
\label{eq:g-coercivity0}
	Q(\V) \ < \ \frac{\dim(\V)}{q}
\end{equation}
for all linear subspaces $\V \subset \R^q$ with $1 \le \dim(\V) < q$. In this case, $L_\rho(\cdot,Q)$ has a unique minimizer on $\M^{(q)}$.
\end{Theorem}

Note that the condition in part~(a) of Theorem~\ref{thm:g-coercivity} is precisely Condition~1 mentioned in Section~\ref{sec:Background}. The additional assumption for uniqueness of the minimizer covers $M$-estimators of scatter as proposed in \cite{Maronna_1976,Huber_1981} with functions $\rho$ which are not strictly g-convex on the whole positive half-line. In part~(b) the condition $Q(\{0\}) = 0$ can be eliminated by replacing $Q$ with $\LL(X \,|\, X \ne 0)$, $X \sim Q$. The conclusion of part~(b) is well known, see \cite{Duembgen_Tyler_2005} and \cite{Duembgen_etal_2015}.

In connection with the algorithms introduced later we need objective functions $L_\rho(\cdot,Q)$ which are twice continuously differentiable. In Setting~0 this is the case, but Setting~1 will be replaced with the following one:

\bigskip

\noindent
\textbf{Setting~2.} \ $\rho$ is twice continuously differentiable on $\R_+$ such that $\psi(s) \ := \ s \rho'(s)$ is strictly increasing in $s \in \R_+$ with limits $\psi(0) = 0$ and $\psi(\infty) \in (q, \infty]$. Moreover, for some constant $\kappa > 0$, $s \psi'(s) \le \kappa \psi(s)$ for all $s \in \R_+$.

\bigskip

\begin{Lemma}[cf.\ \cite{Duembgen_etal_2015}]
\label{lem:g-convexity.log.likelihood}
For $B \in \Rqqns$ and $A \in \Rqqsym$, under Settings~0 and 2,
\[
	L_\rho(B \exp(A) B^\top, Q) - L_\rho(BB^\top)
	\ = \ \langle A, G_\rho(Q_B)\rangle
		+ 2^{-1} H_\rho(A, Q_B) + o(\|A\|^2)
\]
as $A \to 0$, where
\[
	Q_B \ := \ \LL(B^{-1} X), X \sim Q ,
\]
and
\begin{align*}
	G_\rho(Q) \
	&:= \ I_q - \Psi_\rho(Q), \\
	\Psi_\rho(Q) \
	&:= \ \int \rho'(\|x\|^2) xx^\top \, Q(dx) , \\
	H_\rho(A,Q) \
	&:= \ \langle A^2, \Psi_\rho(Q)\rangle
		+ \int \rho''(\|x\|^2) (x^\top Ax)^2 \, Q(dx) .
\end{align*}
Moreover, $H_\rho(A,Q) \ge 0$ with equality if, and only if,
\[
	\begin{cases}
		Q \bigl( \bigcup_{j=1}^m \V_j \bigr) = 1
			& \text{in Setting~0} , \\
		Q(\mathcal{N}_A) = 1 & \text{in Setting~2} .
	\end{cases}
\]
Here $\V_1, \V_2, \ldots, \V_m$ are the different eigenspaces of $A$, and $\mathcal{N}_A := \{x \in \R^q : Ax = 0\}$.
\end{Lemma}

\subsection{Regularization}
\label{subsec:Regularization}

As noted in the introduction, most research on robust estimation of scatter has mainly centered on the unrestricted estimation of the scatter matrix. But the previous results imply that a unique minimizer of $L_\rho(\cdot,Q)$ can only exist if $Q(\V) < 1$ for any proper linear subspace $\V$ of $\R^q$. This excludes empirical distributions $Q_n$ with sample size $n < q$. Some previous work on regularization does exist, with one approach being to introduce a regularization or shrinkage term to the $M$-estimating equations \eqref{eq:Mee}, as is done for the special function $\rho(s) = q \log s$ in  \cite{Chen_etal_2011, Couillet_McKay_2014, Pascal_etal_2014, Wiesel_2012b} and for more general $M$-estimates in \cite{Abramovich_etal_2013, Besson_etal_2013}. Proving existence and/or uniqueness to regularized $M$-estimation equations, though, is not straightforward, and most of the work using this approach does not include conditions to insure such properties. 

Here, we consider a penalized objective function approach, that is we aim to minimize over $\Sigma \in \Rqqsympd$ the function
\begin{equation}
\label{eq:Lpen}
	f_\alpha(\Sigma) \ := \ L_\rho(\Sigma,Q) + \alpha \pi(\Sigma)
\end{equation}
for some tuning parameter $\alpha > 0$ and penalty function $\pi : \Rqqsympd \to \R$. For the special function $\rho(s) = q \log s$, the empirical version of this approach has been considered in \cite{Wiesel_2012b} for certain g-convex penalties, although coercivity is not treated and consequently conditions for existence are not given. The empirical version is also studied in \cite{Ollila_Tyler_2014} for general g-convex $\rho$-functions and general g-convex penalties, but conditions for coercivity are only given for the penalty function $\tr(\Sigma^{-1})$.

\begin{Remark}[The graphical lasso]
A popular penalty function is the $l_1$ penalty on the off-diagonal elements of $\Sigma^{-1}$, i.e.\ when $\pi(\Sigma) = \sum_{i<j} | (\Sigma^{-1})_{ij}|$. In the classical setting, i.e.\ when $L_\rho(\Sigma,Q)$ is taken to be proportional to the multivariate normal negative log-likelihood functional, the problem of minimizing \eqref{eq:Lpen} using this $l_1$ penalty is commonly referred to as a graphical lasso. For this case, as $\alpha$ increases the solutions produce a path of increasing zeros in the off-diagonal elements of $\Sigma^{-1}$. A robust graphical lasso can be constructed by considering general $L_\rho(\Sigma,Q)$, as has been proposed e.g.\ in \cite{Finegold_Drton_2011} for the case when $L_\rho(\Sigma,Q)$ is proportional to the negative log-likelihood of an elliptical t-distribution. One drawback to this approach is that when using $\rho$-functions which yield bounded influence estimators, the function $L_\rho(\Sigma,Q)$ is not convex in $\Sigma^{-1}$ and consequently as $\alpha$ increases the solution path may not yield increasing zeros in the off-diagonal elements of $\Sigma^{-1}$. Moreover, as shown in Supplement~\ref{sec:Auxiliary}, this $l_1$ penalty is not g-convex. So even when $L_\rho(\Sigma,Q)$ is strictly g-convex, the uniqueness of a solution to \eqref{eq:Lpen} is not guaranteed. 
\end{Remark}

Here, we are interested in considering \eqref{eq:Lpen} for the case when both $L_\rho(\Sigma,Q)$ and $\pi(\Sigma)$ are g-convex. Obviously this implies that the penalized objective function $f$ is g-convex, too. Moreover, if either $L_\rho(\Sigma,Q)$ or $\pi(\Sigma)$ are strictly g-convex, then $f$ is strictly g-convex as well.

Note that these considerations apply to the special case when $L_\rho(\Sigma,Q)$ is taken to be proportional to the multivariate normal negative log-likelihood functional, i.e.\ $\rho(s) = s$.  For this case, $L_\rho(\Sigma,Q)$ is not only strictly convex in $\Sigma^{-1}$, it is also strictly g-convex in $\Sigma^{-1}$ and hence in $\Sigma$. Thus, in this classical setting, in addition to penalty functions which are convex in $\Sigma^{-1}$, penalty functions which are g-convex in $\Sigma$ also ensure the uniqueness of a minimum to \eqref{eq:Lpen}, provided a minimum exists.

The existence of a minimizer to \eqref{eq:Lpen} depends on the geodesic coercivity of $f(\Sigma)$, which in turn depends of the behavior of $L_\rho(\Sigma,Q)$ and $\pi(\Sigma)$ as $\|\log(\Sigma)\| \to \infty$. For $L_\rho(\Sigma,Q)$, Proposition~\ref{prop:g-coercivity} provides a complete answer, so it remains to specify and investigate the penalties $\pi(\Sigma)$.

\paragraph{Shrinkage towards $I_q$.}
Functions which penalize deviations from $I_q$ are
\begin{align*}
	\Pi_0(\Sigma) \
	&:= \ \tr(\Sigma) + \tr(\Sigma^{-1})
		\ = \ \sum_{i=1}^q (\sigma_i + \sigma_i^{-1}) , \\
	\Pi_1(\Sigma) \
	&:= \ \log\det(\Sigma) + \tr(\Sigma^{-1})
		\ = \ \sum_{i=1}^q (\log \sigma_i + \sigma_i^{-1}) , \\
	\Pi_2(\Sigma) \
	&:= \ \| \log(\Sigma) \|^2 
		\ = \ \sum_{i=1}^q (\log \sigma_i)^2 ,
\end{align*}
where $\sigma_1 \ge \cdots \ge \sigma_q > 0$ are the eigenvalues of $\Sigma$. In all three cases, $\Sigma = I_q$ is the unique minimizer. Note that $\Pi_2(\Sigma)$ is just the square of the geodesic distance $d_g(I_p,\Sigma)$. While $\Pi_0$ and $\Pi_2$ satisfy the symmetry relation $\Pi(\Sigma^{-1}) = \Pi(\Sigma)$, the penalty $\Pi_1(\Sigma)$ is non-symmetric, penalizing very small eigenvalues more severely than very large ones. It corresponds to the Kullback-Leibler divergence between $\NN_q(0,\Sigma)$ and $\NN_q(0,I_q)$ and has been previously considered in \cite{Sun_etal_2014}. In principle one could also use the penalty $\Pi_1'(\Sigma) = \Pi_1(\Sigma^{-1})$, but from a statistical perspective this seems to be less reasonable.

The next lemma summarizes the essential properties of these penalties.

\begin{Lemma}
\label{lem:all.about.Pi}
For $k = 0,1,2$, the penalty function $\Pi_k$ is twice continuously differentiable and strictly geodesically convex on $\Rqqsympd$ with a unique minimum at $I_q$.

\noindent
Precisely, for any $B \in \Rqqns$, as $\Rqqsym \ni A \to 0$,
\[
	\Pi_k(B \exp(A) B^\top)
	\ = \ \Pi_k(BB^\top) + \langle A, G_k(B)\rangle
		+ 2^{-1} H_k(A,B) + o(\|A\|^2)
\]
with $G_k(B)$ and $H_k(A,B)$ given in the following table:
\[
	\begin{array}{|c||c|c|}
	\hline
	k & G_k(B) & H_k(A,B) \\
	\hline\hline
	0_{\strut}^{\strut}
		& B^\top B - B^{-1} B^{-\top}
		& \langle A^2, B^\top B + B^{-1} B^{-\top} \rangle \\
	\hline
	1_{\strut}^{\strut}
		& I_q - B^{-1} B^{-\top}
		& \langle A^2, B^{-1} B^{-\top} \rangle \\
	\hline
	2_{\strut}^{\strut}
		& 2 \log(B^\top B)
		& 2 \sum_{i,j=1}^q W_{ij}(\lambda) (v_i^\top A v_j)^2 \\
	\hline
	\end{array}
\]
Here $B^\top B = V D(e^\lambda) V^\top$ with $V = [v_1, v_2, \ldots, v_q] \in \Rqqorth$ and $\lambda \in \R^q$, and
\[
	W_{ij}(\lambda)
	\ := \ \frac{(\lambda_i - \lambda_j)/2}{ \tanh((\lambda_i - \lambda_j)/2)}
	\ \ge \ 1
\]
with the convention $0/\tanh(0) := 1$. In particular, $H_k(A,B) > 0$ whenever $A \ne 0$.

\noindent
Moreover, if $A = U D(-\gamma) U^\top$ with $U \in \Rqqorth$ and $\gamma \in \R^q \setminus \{0\}$ such that $\gamma_1 \le \gamma_2 \le \cdots \le \gamma_q$, then
\[
	\lim_{t \to \infty} \, \frac{d}{dt} \Pi_k(\exp(tA))
	\ = \ \begin{cases}
		\infty
			&\text{if} \ k = 0 , \\
		1_{[\gamma_q > 0]} \infty - \sum_{i=1}^q \gamma_i
			&\text{if} \ k = 1 , \\
		\infty
			&\text{if} \ k = 2 .
	\end{cases}
\]
\end{Lemma}

This lemma and Theorem~\ref{thm:Mfunc} together show that using any of the penalties $\Pi_0$, $\Pi_1$ or $\Pi_2$ together with a g-convex function $\rho$ yields an objective function $f$ in \eqref{eq:Lpen} which is strictly g-convex. In particular, by Corollary~\ref{cor:uniqueness.minimizer}, \eqref{eq:Lpen} has a unique minimizer or no minimizer. With $\Pi_0$ or $\Pi_2$ g-coercivity and thus existence of a unique minimizer is guaranteed, regardless of $Q$. This is in contrast to the non-regularized case for which conditions on $Q$ are needed to insure the existence of a minimizer.

Shrinkage towards a different given matrix $\Sigma_o \in \Rqqsympd$ is obtained by replacing $\Sigma$ in $\Pi_k(\Sigma)$ with $\Sigma_o^{-1/2} \Sigma \Sigma_o^{-1/2}$.

\paragraph{Shrinkage towards multiples of $I_q$.}
Functions which penalize large condition numbers $\sigma_1/\sigma_q$ of $\Sigma$ are given by
\begin{align*}
	\pi_0(\Sigma) \
	&:= \ \log\tr(\Sigma) + \log\tr(\Sigma^{-1})
		\ = \ \log \Bigl( \sum_{i=1}^q \sigma_i \Bigr)
			+ \log \Bigl( \sum_{i=1}^q \sigma_i^{-1} \Bigr) , \\
	\pi_1(\Sigma) \
	&:= \ q^{-1} \log\det(\Sigma) + \log\tr(\Sigma^{-1})
		\ = \ q^{-1} \sum_{i=1}^q \log \sigma_i
			+ \log \Bigl( \sum_{i=1}^q \sigma_i^{-1} \Bigr) , \\
	\pi_2(\Sigma) \
	&:= \ \Pi_2(\det(\Sigma)^{-1/q} \Sigma)
		\ = \ \sum_{i=1}^q \Bigl( \log \sigma_i
			- q^{-1} \sum_{j=1}^q \log \sigma_j \Bigr)^2 .
\end{align*}
All three functions are scale-invariant with $\Sigma$ minimizing $\pi_j(\Sigma)$ if, and only if, $\Sigma$ is a positive multiple of $I_q$. Moreover, $\pi_0$ and $\pi_2$ satisfy the symmetry relation $\pi(\Sigma^{-1}) = \pi(\Sigma)$, whereas $\pi_1(\Sigma)$ penalizes relatively small eigenvalues more severely than relatively large ones. Here are the main facts:

\begin{Lemma}
\label{lem:all.about.pi}
For $k = 0,1,2$, the penalty function $\pi_k$ is scale-invariant, twice continuously differentiable and geodesically convex. On $\M^{(q)}$ it is strictly geodesically convex with a unique minimum at $I_q$.

\noindent
Precisely, for any $B \in \Rqqns$, as $\Rqqsym \ni A \to 0$,
\[
	\pi_k(B \exp(A) B^\top)
	\ = \ \pi_k(BB^\top) + \langle A, G_k(B)\rangle
		+ 2^{-1} H_k(A,B) + o(\|A\|^2)
\]
with $G_k(B)$ and $H_k(A,B)$ given in the following table:
\[
	\begin{array}{|c||c|c|}
	\hline
	k & G_k(B) & H_k(A,B) \\
	\hline\hline
	0^{\strut}
		& N(B^\top B) - N(B^{-1} B^{-\top})
		& \langle A^2, N(B^\top B) \rangle
			- \langle A, N(B^\top B)\rangle^2 \\
	\strut_{\strut}
		&
		& + \ \langle A^2, N(B^{-1}B^{-\top})\rangle
			- \langle A, N(B^{-1} B^{-\top})\rangle^2 \\
	\hline
	1_{\strut}^{\strut}
		& q^{-1} I_q - N(B^{-1}B^{-\top})
		& \langle A^2, N(B^{-1} B^{-\top}) \rangle
			- \langle A, N(B^{-1} B^{-\top}) \rangle^2 \\
	\hline
	2_{\strut}^{\strut}
		& 2 \log(B^\top B)^o
		& 2 \sum_{i,j=1}^q W_{ij}(\lambda) (v_i^\top A^o v_j)^2 \\
	\hline
	\end{array}
\]
Here $N(\Sigma) := \tr(\Sigma)^{-1} \Sigma$, $C^o := C - q^{-1} \tr(C) I_q$ for $C \in \Rqqsym$, and $V = [v_1, \ldots, v_q]$, $\lambda$, $W_{ij}(\lambda)$ are defined as in Lemma~\ref{lem:all.about.Pi}. In particular, $H_k(A,B) > 0$ whenever $A \in \W^{(q)} \setminus \{0\}$.

\noindent
Moreover, if $A = U D(-\gamma) U^\top$ with $U \in \Rqqorth$ and $\gamma \in \R^q$ such that $\gamma_1 \le \gamma_2 \le \cdots \le \gamma_q$ and $\gamma_1 < \gamma_q$,
\[
	\lim_{t \to \infty} \, \frac{d}{dt} \pi_k(\exp(tA))
	\ = \ \begin{cases}
		\gamma_q - \gamma_1
			&\text{if} \ k = 0 \\
		\gamma_q - \bar{\gamma}
			&\text{if} \ k = 1 \\
		\infty
			&\text{if} \ k = 2
	\end{cases}
\]
with $\bar{\gamma} := q^{-1} \sum_{i=1}^q \gamma_i$.
\end{Lemma}

Of course one could replace any of these penalties $\pi_k(\Sigma)$ with a non-decreasing convex function thereof. As pointed out in Remark~\ref{rem:Transformation}, this would preserve geodesic convexity.

\paragraph{A scale-invariant example.} We consider the special case where $\rho(s) = q \log s$ for $s > 0$ and $Q(\{0\}) = 0$. Since $L_\rho(\Sigma,Q)$ is scale-invariant, it is natural to choose a penalty which is scale-invariant, too, and to treat $f$ as a function on $\M^{(q)}$. If $\pi$ is strictly g-convex on the latter set, then $f$ inherits this property.

As to g-coercivity, let $A = U D(-\gamma) U^\top$ with $U = [u_1, \ldots, u_q] \in \Rqqorth$ and $\gamma \in \R^q$ such that $\gamma_1 \le \cdots \le \gamma_q$ and $\gamma_1 < \gamma_q$. If $\pi = \pi_0$, then
\begin{align*}
	\lim_{t \to \infty} \frac{d}{dt} f(\exp(tA)) \
	&= \ q \sum_{k=1}^{q-1} (k/q - Q(\V_k))(\gamma_{k+1} - \gamma_k) 
		+ \alpha (\gamma_q - \gamma_1) \\
	&= \ q \sum_{k=1}^{q-1} \bigl( (k + \alpha)/q - Q(\V_k) \bigr)
		(\gamma_{k+1} - \gamma_k) .
\end{align*}
Thus $f$ is g-coercive on $\M^{(q)}$ if, and only if,
\[
	Q(\V) \ < \ (\dim(\V) + \alpha)/q
\]
for any subspace $\V$ of $\R^q$ with $1 \le \dim(\V) < q$. If $\pi = \pi_1$, then
\begin{align*}
	\lim_{t \to \infty} \frac{d}{dt} f(\exp(tA)) \
	&= \ q \sum_{k=1}^{q-1} (k/q - Q(\V_k))(\gamma_{k+1} - \gamma_k) 
		+ \alpha (\gamma_q - \bar{\gamma}) \\
	&= \ q \sum_{k=1}^{q-1} (k/q - Q(\V_k))(\gamma_{k+1} - \gamma_k) 
		+ \alpha \sum_{k=1}^{q-1} \frac{k}{q} (\gamma_{k+1} - \gamma_k) \\
	&= \ q \sum_{k=1}^{q-1} \bigl( (k/q) (1 + \alpha/q) - Q(\V_k) \bigr)
		(\gamma_{k+1} - \gamma_k) .
\end{align*}
Thus $f$ is g-coercive on $\M^{(q)}$ if, and only if,
\[
	Q(\V) \ < \ (1 + \alpha/q) \dim(\V) / q
\]
for any subspace $\V$ of $\R^q$ with $1 \le \dim(\V) < q$.

In case of
\[
	\lim_{t \to \infty} \frac{d}{dt} \pi(\exp(tA)) \ = \ \infty
\]
for any fixed $A \in \W^{(q)} \setminus \{0\}$, the function $f$ is g-coercive on $\M^{(q)}$ without further constraints on $Q$. This is the case, for instance, if $\pi(\Sigma) = \pi_2(\Sigma)$ or
\[
	\pi(\Sigma) \ = \ \psi \bigl( \pi_k(\Sigma) - \pi_k(I_q) \bigr)
\]
for $k = 0,1$ with a non-decreasing convex function $\psi : [0,\infty) \to [0,\infty)$ such that $\psi(t)/t \to \infty$ as $t \to \infty$. Explicit examples for such functions $\psi$ are
\begin{align*}
	\psi(s) \ &:= \ (1 + s)^\gamma/\gamma , \quad \gamma > 1 , \\
	\psi(s) \ &:= \ \exp(c s), \quad c > 0 .
\end{align*}

\subsection{Cross validation}

Rather than choose $\alpha$ in \eqref{eq:Lpen} beforehand, one can use data dependent methods for selecting $\alpha$. One possible approach is to use an oracle type estimator for $\alpha$, as is done in \cite{Chen_etal_2011,Ollila_Tyler_2014}. Such an approach is based upon minimizing the mean square error under a specific distribution with the method being dependent on the choice of the penalty $\pi$ and the $\rho$-function. A more universal approach is to use cross-validation. Here we propose a leave-one-out cross validation approach for the current problem as follows. Let $Q_{n,(i)}$ denoted the empirical distribution when the $i$th data point is removed, and for a given $\alpha$ define
\[
	\hat{\Sigma}_{\alpha,(i)}
	\ := \ \argmin \bigl\{ L_\rho(\Sigma,Q_{n,(i)}) + \alpha \pi(\Sigma) \bigr\} ,
\]
with the minimum being taken over $\Sigma \in \Rqqsympd$. Next, define an aggregate robust measure of how well 
$\hat{\Sigma}_{\alpha,(i)}$ reflects the left-out observation $x_i$ by
\[
	\mathrm{CV}(\alpha) := \ \sum_{i=1}^n
		\bigl\{ \rho(x_i^\top \hat{\Sigma}_{\alpha,(i)}^{-1} x_i^{})
		+ \log \det(\hat{\Sigma}_{\alpha,(i)}^{}) \bigr\} .
\]
The objective is to then minimize $\mathrm{CV}(\alpha)$ over $\alpha \ge 0$. In practice, this would be done over over some finite set of values for $\alpha$. Some examples are given in section \ref{sec:Example}. Since the cross validation approach can be computationally intensive, we first discuss algorithms for computing the regularized $M$-estimators of scatter.

\section{Algorithms}
\label{sec:Algorithm}

There is a rich literature on optimization on Riemannian manifolds, see \cite{Ring_Wirth_2012} and the references therein. For the special case of functions on $\Rqqsympd$, \cite{Sra_Hosseini_2013, Sra_Hosseini_2015} propose various fixed-point and gradient descent methods. Newton-Raphson algorithms would be another possibility but may be inefficient due to the high dimension of Hessian operators. For the minimization of a smooth and g-convex function we propose a partial Newton-Raphson algorithm which is similar to a method of \cite{Duembgen_etal_2016} for pure $M$-functionals of scatter. While the latter method has been designed for special settings in which a certain fixed-point algorithm serves as a fallback option with guaranteed convergence, the present approach is more general.

We consider a twice continuously differentiable function $f : \Rqqsympd \to \R$ such that
\[
	H(A,B) \ > \ 0
	\quad\text{for any} \ A \in \Rqqsym \setminus \{0\} \ \text{and} \ B \in \Rqqns .
\]
In particular, $f$ is strictly g-convex. Furthermore we assume that $f$ is g-coercive, so
\[
	\Sigma_* \ := \ \argmin_{\Sigma \in \Rqqsympd} f(\Sigma)
\]
exists. Finally we assume that $G(B)$ and $H(A,B)$ are continuous in $B \in \Rqqns$ for any fixed $A \in \Rqqsym$. Under these conditions on $f$ one can devise an iterative algorithm to compute the minimizer $\Sigma_*$. According to Lemma~\ref{lem:minimizer.g-convex}, this is equivalent to finding a matrix $B_* \in \Rqqns$ such that $G(B_*) = 0$.

\paragraph{Algorithmic mappings.}
To compute $\Sigma_*$ we iterate a certain mapping
\[
	\phi : \Rqqsympd \to \Rqqsympd
\]
such that $\phi(\Sigma_*) = \Sigma_*$ and $f(\phi(\Sigma)) < f(\Sigma)$ whenever $\Sigma \ne \Sigma_*$. If we replace the latter condition by a somewhat stronger constraint, iterating the mapping $\phi$ yields sequences with guaranteed converge to $\Sigma_*$.

\begin{Lemma}
\label{lem:algorithm}
Suppose that $\phi : \Rqqsympd \to \Rqqsympd$ satisfies $\phi(\Sigma_*) = \Sigma_*$ and
\[
	\limsup_{\Sigma \to \Sigma_o} f(\phi(\Sigma)) \ < \ f(\Sigma_o)
	\quad\text{for any} \ \Sigma_o \in \Rqqsympd \setminus \{\Sigma_*\} .
\]
Let $\Sigma_1 \in \Rqqsympd$ be an arbitrary starting point, and define inductively $\Sigma_{k+1} := \phi(\Sigma_k)$ for $k = 1,2,3,\ldots$. Then
\[
	\lim_{k \to \infty} \Sigma_k \ = \ \Sigma_* .
\]
\end{Lemma}

This lemma belongs to the folklore in optimization theory. For the reader's convenience we provide its short proof in Section~\ref{sec:Auxiliary}.

\paragraph{Construction of $\phi$.}
Let $\Sigma = BB^\top$ with $B \in \Rqqns$ be our current candidate for $\Sigma_*$. Note that the quadratic term $H(A,B)$ may be rewritten as
\[
	H(A,B) \ = \ \langle A, H_B A\rangle
\]
for a self-adjoint linear operator $H_B : \Rqqsym \to \Rqqsym$ with strictly positive eigenvalues. Thus a promising new candidate for $\Sigma_*$ would be
\[
	\phi_{\rm fN}(\Sigma) \ := \ B \exp(A_{\rm fN}) B^\top
\]
with
\[
	A_{\rm fN}
	\ := \ \argmin_{A \in \Rqqsym}
		\bigl( \langle A, G(B)\rangle + 2^{-1} H(A, B) \bigr)
	\ = \ - H_B^{-1} G(B) ,
\]
a full Newton step in local geodesic coordinates.

Computing $A_{\rm fN}$ would require substantial memory and computation time, though. Alternatively one could try a gradient descent step:
\[
	\phi_{\rm G}(\Sigma) \ := \ B \exp(A_{\rm G}) B^\top
\]
with
\[
	A_{\rm G}
	\ := \ \argmin_{A \in \{t G(B) : t \in \R\}}
		\bigl( \langle A, G(B)\rangle + 2^{-1} H(A, B) \bigr)
	\ = \ - \frac{\|G(B)\|^2}{H(G(B),B)} \, G(B) .
\]

As a compromise between a full Newton and a mere gradient step we propose a partial Newton step: To this end we consider a spectral decomposition
\[
	G(B) \ = \ U D(\lambda) U^\top
\]
with an orthogonal matrix $U = U(B) \in \Rqq$ and a vector $\lambda = \lambda(B) \in \R^q$. Then we define
\[
	\phi_{\rm pN}(\Sigma) \ := \ B \exp(A_{\rm pN}) B^\top
\]
with
\[
	A_{\rm pN} = A_{\rm pN}(B,U)
	\ := \ \argmin_{A \in \{U D(x) U^\top : x \in \R^q\}}
		\bigl( \langle A, G(B)\rangle + 2^{-1} H(A, B) \bigr) .
\]
This may be computed explicitly: Since
\[
	\langle U D(x) U^\top, G(B)\rangle + 2^{-1} H(U D(x) U^\top, B)
	\ = \ x^\top \lambda(B) + 2^{-1} x^\top \underline{H}(BU) x
\]
for a certain matrix $\underline{H}(BU) \in \Rqqsympd$, we may write
\[
	A_{\rm pN} \ = \ - U D \bigl( \underline{H}(BU)^{-1} \lambda(B) \bigr) U^\top .
\]

If $\Sigma = BB^\top$ is far from $\Sigma_*$, the matrix $\phi_{\rm pN}(\Sigma)$ need not be better than $\Sigma$ itself. To avoid poor steps we introduce a simple step size correction and define finally
\[
	\phi(\Sigma) \ := \ B \exp \bigl( 2^{-m(BU)} A_{\rm pN} \bigr) B^\top
	\ = \ BU D \bigl(
		\exp \bigl( - 2^{-m(BU)} \underline{H}(BU)^{-1} \lambda(B) \bigr) \bigr)
		(BU)^\top
\]
with $m(BU)$ being the smallest integer $m \ge 0$ such that
\[
	f \bigl( B \exp(2^{-m} A_{\rm pN}) B^\top \bigr) - f(\Sigma)
	\ \le \ 2^{-m} \langle A_{\rm pN}, G(B) \rangle / C
\]
for a given $C > 2$. The rationale behind this definition is the fact that
\[
	\min_{x \in \R^q}
		\bigl( \langle UD(x)U^\top, G(B)\rangle + 2^{-1} H(UD(x)U^\top, B) \bigr)
	\ = \ \langle A_{\rm pN}, G(B)\rangle / 2
\]
and
\[
	\lim_{m \to \infty}
		\frac{f \bigl( B \exp(2^{-m} A_{\rm pN}) B^\top \bigr) - f(\Sigma)}{2^{-m}}
	\ = \ \langle A_{\rm pN}, G(B)\rangle .
\]
Note that $\phi(\Sigma) = \Sigma$ whenever $G(B) = 0$, which is equivalent to $\Sigma = \Sigma_*$. Otherwise
\[
	\langle A_{\rm pN}, G(B)\rangle
	\ = \ - \lambda(B)^\top \underline{H}(BU)^{-1} \lambda(B)
	\ < \ 0 .
\]

This algorithmic mapping $\phi$ has the desired properties, no matter how the factor $B$ of $\Sigma = BB^\top$ and the orthogonal matrix $U$ in the spectral decomposition $G(B) = U D(\lambda) U^\top$ are chosen.

\begin{Theorem}
\label{thm:phi}
The algorithmic mapping just defined has the properties described in Lemma~\ref{lem:algorithm}. Moreover, if $\Sigma = BB^\top$ is sufficiently close to $\Sigma_*$, then the number $m(BU)$ in the step size correction equals $0$, whence $\phi(\Sigma) = \phi_{\rm pN}(\Sigma) = B \exp(A_{\rm pN}) B^\top$.
\end{Theorem}

\paragraph{Pseudo-code for $\phi(\cdot)$.}
One may interpret our algorithmic mapping $\phi$ such that the factor $B$ of our current candidate $\Sigma = BB^\top$ for $\Sigma_*$ is replaced with a new matrix
\[
	B_{\rm new} \ = \ B U \exp(2^{-m(BU)} A_{\rm pN}/2) ,
\]
and $\phi(\Sigma) = B_{\rm new}^{} B_{\rm new}^\top$. Here is corresponding pseudo-code for the computation of $B_{\rm new}$:

\begin{align*}
	&(U,\lambda) \ \leftarrow \ \text{eigen}(G(B)) \\
	&a \ \leftarrow \ H(BU)^{-1} g(BU) \\
	&\epsilon \ \leftarrow \ a^\top g(BU) \\
	&\text{while} \ f(BB^\top) - f(B D(\exp(-a)) B^\top) < \epsilon/C \ \text{do} \\
	&\quad a \ \leftarrow \ a/2 \\
	&\quad \epsilon \ \leftarrow \ \epsilon/2 \\
	&\text{end while} \\
	&B_{\rm new} \ \leftarrow \ B U D(\exp(- a/2))
\end{align*}

\section{Numerical Example}
\label{sec:Example}

We illustrate the proposed methods in case of $\rho(s) = q \log s$ and
\[
	\pi(\Sigma) \ := \ \exp(\pi_1(\Sigma) - \pi_1(I_q))
	\ = \ \det(\Sigma)^{1/q} \tr(\Sigma^{-1})/q .
\]
The resulting functional $f_\alpha(\Sigma) = L_\rho(\Sigma,Q) + \alpha \pi(\Sigma)$ is strictly g-convex and g-coercive on $\M^{(q)}$ for any value $\alpha > 0$.

Precisely, we chose $q = 50$ and simulated a random sample of size $n = 30$ from the multivariate Cauchy distribution with center $0$ and scatter matrix
\[
	\Sigma \ = \ D(10,5,3,2,1,1,\ldots,1)^2 .
\]
Then we computed the minimizer $\hat{\Sigma}(\alpha)$ of $f_\alpha$ with $Q$ being the empirical distribution of this sample for $\alpha = 2^z$ with $z = 1,2,\ldots,15$. Table~\ref{tab:CV.errors} shows the resulting values $\mathrm{CV}(\alpha)$ and the following estimation errors:
\begin{align*}
	\epsilon_0(\alpha) \ &: \quad
		\text{Euclidean distance between first eigenvectors of} \
		\Sigma, \hat{\Sigma}(\alpha) , \\
	\epsilon_1(\alpha) \ &: \quad
		\text{Euclidean distance between} \
		\log \lambda(S), \log \lambda(\hat{S}(\alpha)) , \\
	\epsilon_2(\alpha) \ &: \quad
		\text{geodesic distance between} \
		S, \hat{S}(\alpha) ,
\end{align*}
where $S := \det(\Sigma)^{-1/q} \Sigma$, $\hat{S}(\alpha) := \det(\hat{\Sigma}(\alpha))^{-1/q} \hat{\Sigma}(\alpha)$, and $\lambda(B)$ refers to the vector of the ordered eigenvalues of a symmetric matrix $B$. Note that our cross-validation criterion yields $\alpha = 2^7$, which is a reasonable choice in view of the estimation errors. Figure~\ref{fig0} shows a bar plot of the log-transformed eigenvalues of $S$ and of $\hat{S}(2^7)$.

\begin{table}[h!]
\[
	\begin{array}{|c||c||c|c|c|}
	\hline
	\log_2(\alpha)
		& \mathrm{CV}(\alpha)
			& \epsilon_0(\alpha)
				& \epsilon_1(\alpha)
					& \epsilon_2(\alpha) \\
	\hline\hline
	 1 & 11670.248 & 0.164 & 20.817 & 118.797 \\
	\hline
	 2 & 10658.798 & 0.164 & 16.985 &  74.729 \\
	\hline
	 3 &  9704.005 & 0.163 & 13.278 &  46.696 \\
	\hline
	 4 &  8883.141 & 0.160 &  9.793 &  28.871 \\
	\hline
	 5 &  8282.730 & 0.158 &  6.660 &  17.781 \\
	\hline
	 6 &  7933.924 & 0.158 &  4.141 &  11.518 \\
	\hline
	 7 &  7816.171 & 0.160 &  2.899 &   8.787 \\
	\hline
	 8 &  7883.674 & 0.165 &  3.307 &   8.098 \\
	\hline
	 9 &  8079.799 & 0.173 &  4.260 &   8.137 \\
	\hline
	10 &  8321.868 & 0.183 &  5.035 &   8.295 \\
	\hline
	11 &  8515.666 & 0.190 &  5.499 &   8.407 \\
	\hline
	12 &  8633.030 & 0.194 &  5.740 &   8.467 \\
	\hline
	13 &  8695.983 & 0.196 &  5.859 &   8.497 \\
	\hline
	14 &  8728.327 & 0.197 &  5.918 &   8.513 \\
	\hline
	15 &  8744.677 & 0.198 &  5.947 &   8.520 \\
	\hline
	\end{array}
\]
\caption{Cross-validation criterion and estimation errors for one data matrix. }
\label{tab:CV.errors}
\end{table}

\begin{figure}[b!]
\centering
\includegraphics[width=0.99\textwidth]{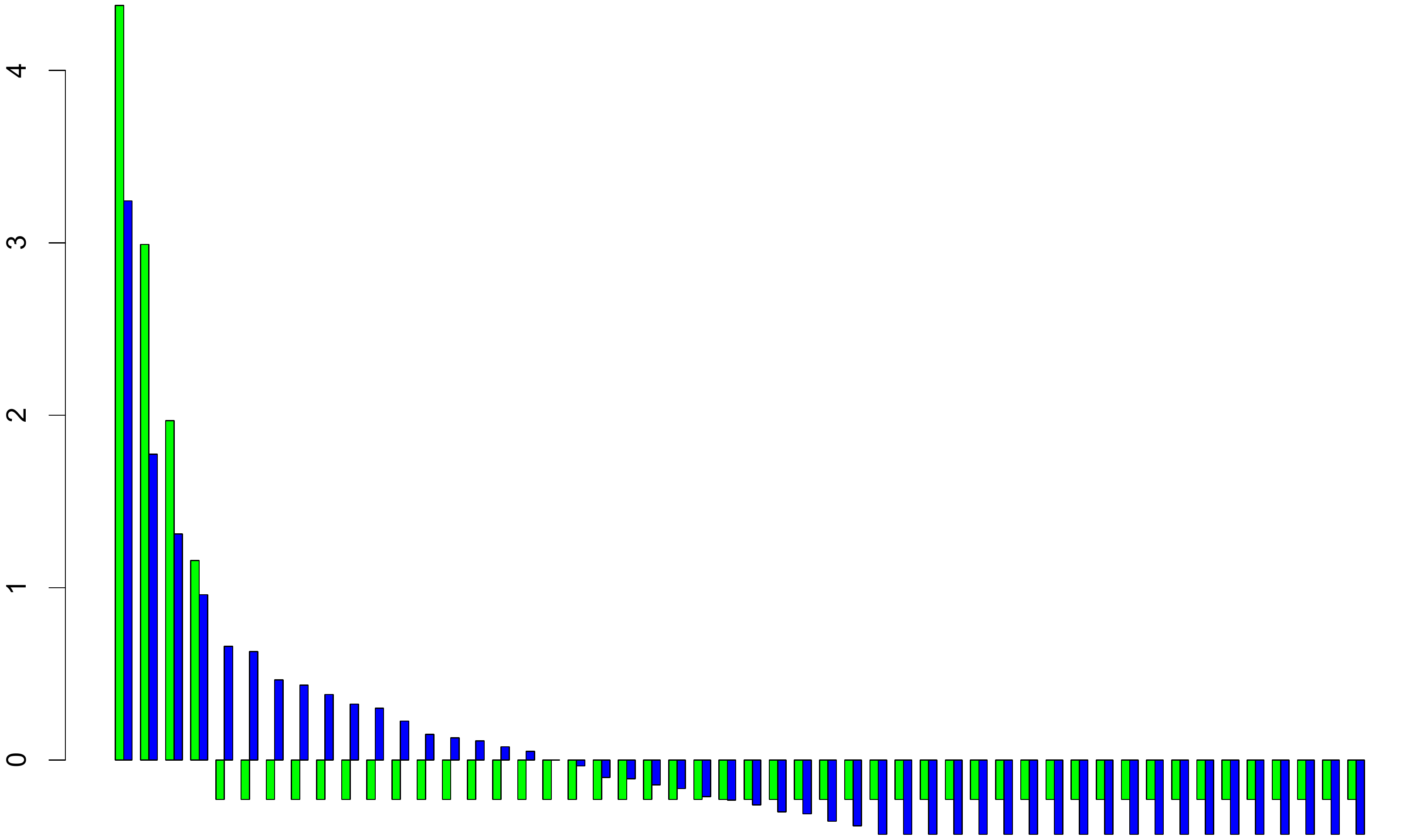}
\caption{Log-eigenvalues of $S$ (green) and $\hat{S}(2^7)$ (blue).}
\label{fig0}
\end{figure}

This simulation was repeated 100 times, and in all cases the minimizer of $\mathrm{CV}(\alpha)$ on the given grid turned out to be $2^7 = 128$. Figure~\ref{fig1} shows box plots of $\mathrm{CV}(\alpha)$ and the estimation errors $\epsilon_0(\alpha)$, $\epsilon_1(\alpha)$, $\epsilon_2(\alpha)$ for these simulations.

\begin{figure}[b!]
\includegraphics[width=0.49\textwidth]{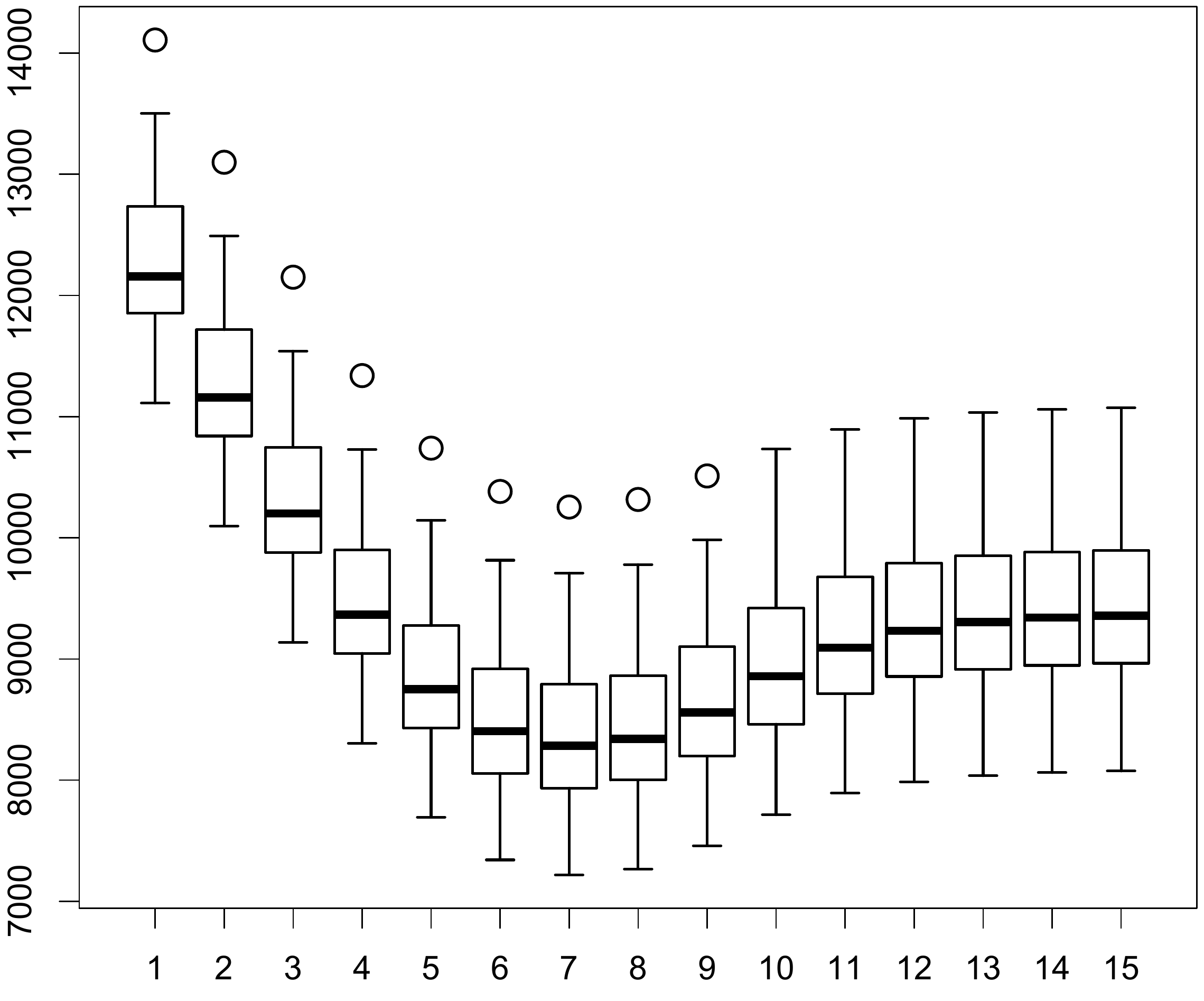}
\hfill
\includegraphics[width=0.49\textwidth]{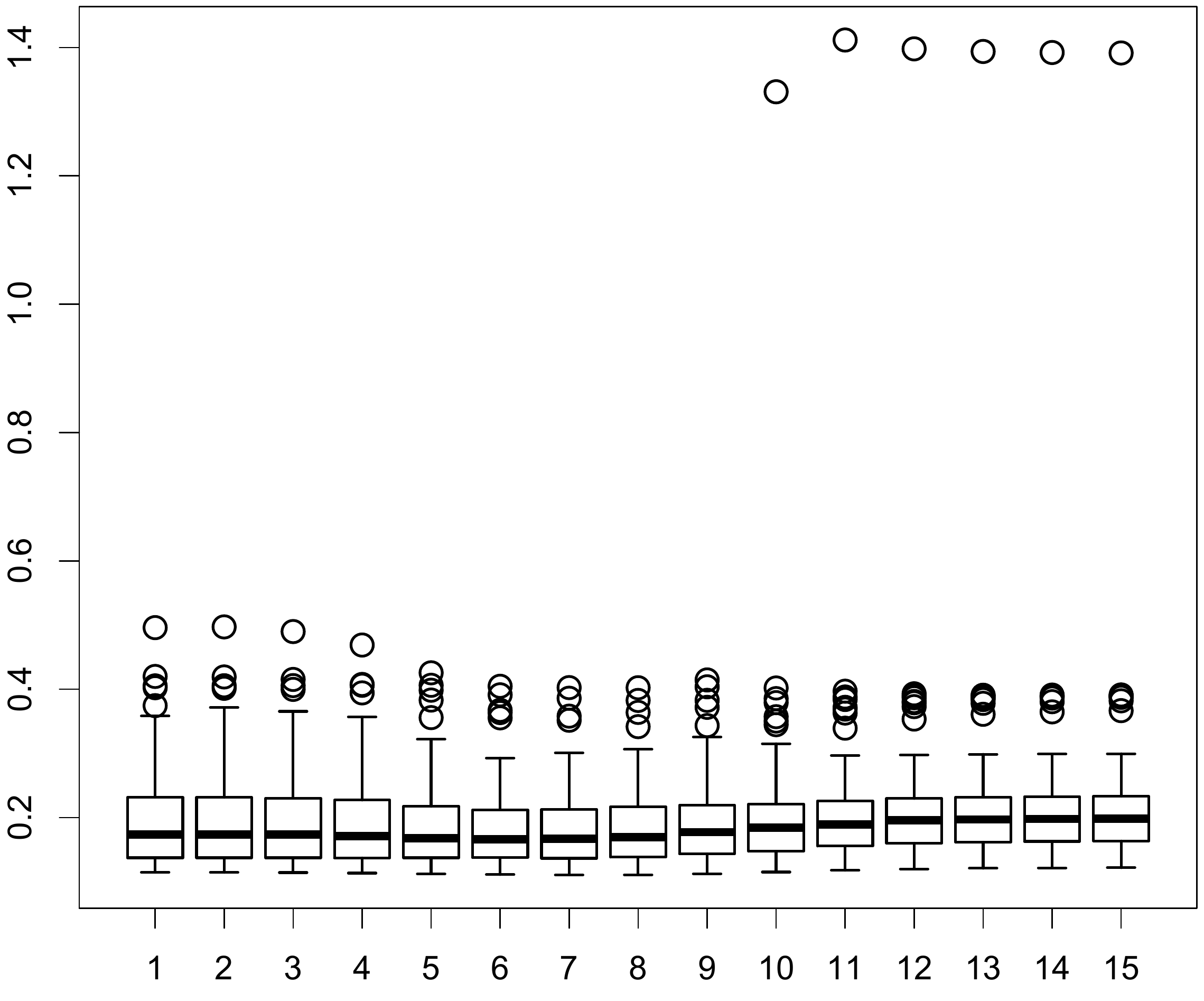}

\includegraphics[width=0.49\textwidth]{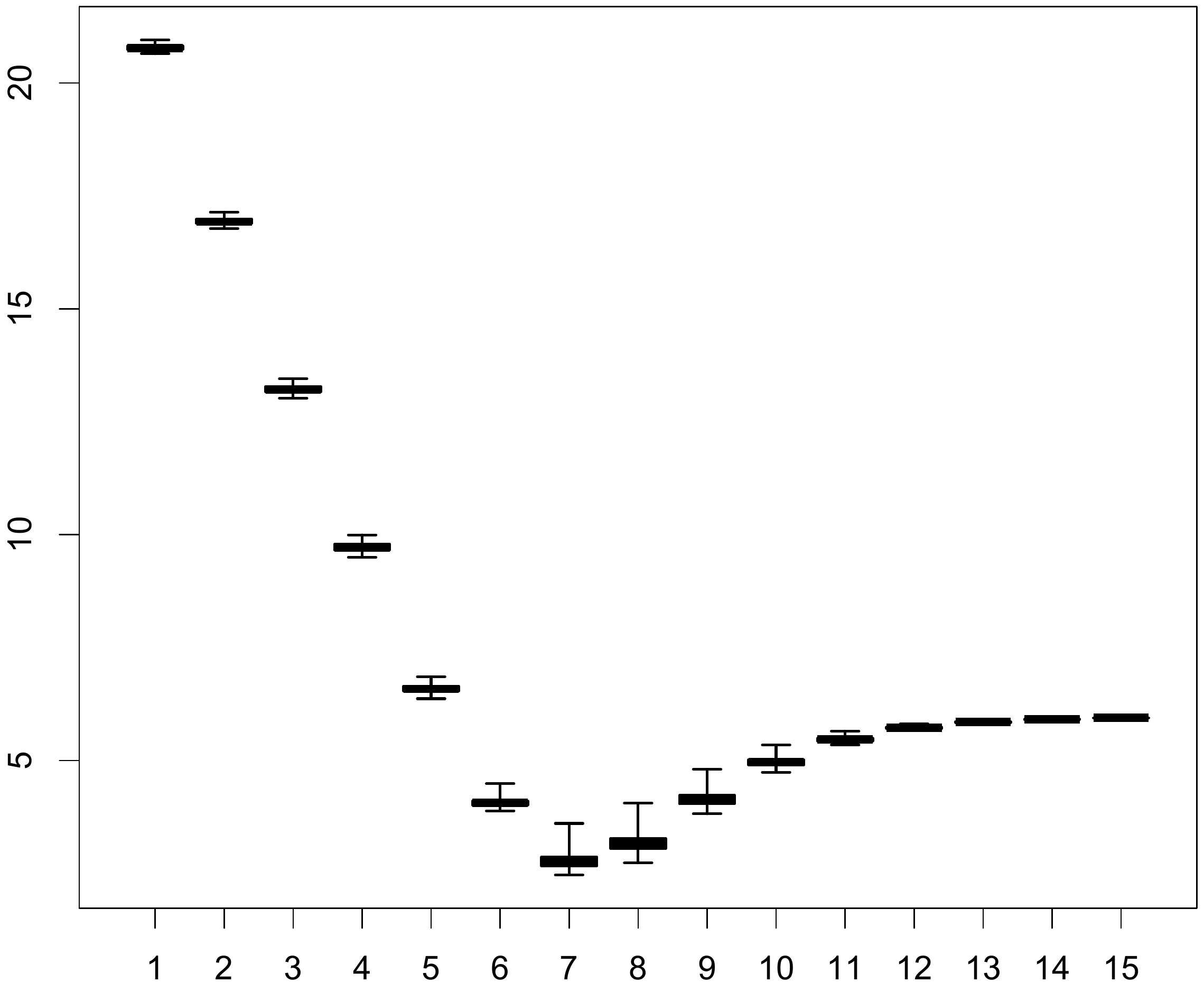}
\hfill
\includegraphics[width=0.49\textwidth]{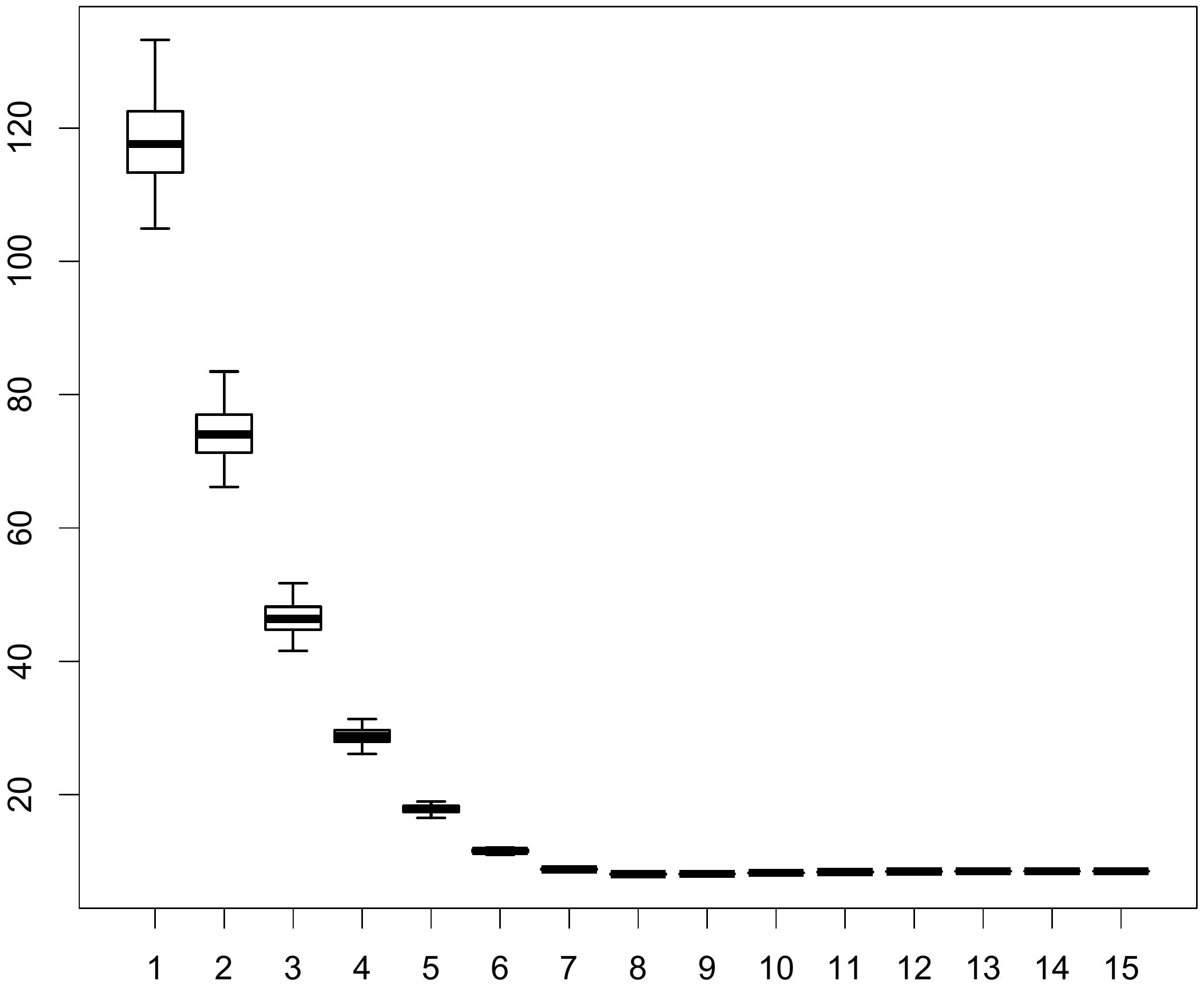}

\caption{Cross-validation measures $\mathrm{CV}(\alpha)$ (upper left) and estimation errors $\epsilon_0(\alpha)$ (upper right), $\epsilon_1(\alpha)$ (lower left), $\epsilon_2(\alpha)$ (lower right) versus $\log_2(\alpha)$.}
\label{fig1}
\end{figure}

\section{Proofs}
\label{sec:Proofs}

\subsection{Proofs for Section~\ref{sec:G-Convexity}}

\begin{proof}[\bf Proof of Lemma~\ref{lem:geodesic.curves}]
For $B \in \Rqqns$ and $A_0,A_1 \in \Rqqsym$ define $\Sigma_j := B \exp(A_j) B^\top$. Then $\Sigma_0 = B_0^{} B_0^\top$ with $B_0 := B \exp(A_0/2)$, and this implies that
\[
	\Sigma_0^{1/2} \ = \ B_0 V \ = \ V^\top B_0^\top
\]
for some $V \in \Rqqorth$. Hence
\[
	\Sigma_0^{-1/2} \ = \ V^\top B_0^{-1} \ = \ B_0^{-\top} V ,
\]
and for $u \in [0,1]$,
\begin{align*}
	\Sigma_0^{1/2}
	& (\Sigma_0^{-1/2} \Sigma_1 \Sigma_0^{-1/2})_{}^u \Sigma_0^{1/2} \\
	&= \ B_0 V \,
		(V^\top B_0^{-1} \Sigma_1^{} B_0^{-\top} V)_{}^u \,
			V^\top B_0^{-1} \\
	&= \ B_0 \, (B_0^{-1} \Sigma_1^{} B_0^{-\top})_{}^u \, B_0^{-1} \\
	&= \ B \exp(A_0/2)
		\bigl( \exp(- A_0/2) \exp(A_1) \exp(-A_0/2) \bigr)_{}^u
		\exp(A_0/2) B^\top .
\end{align*}
If $A_0A_1 = A_1A_0$, the right hand side may be simplified further and we obtain
\begin{align*}
	\Sigma_0^{1/2}
		(\Sigma_0^{-1/2} \Sigma_1 \Sigma_0^{-1/2})_{}^u \Sigma_0^{1/2} \
	&= \ B \exp(A_0/2) \exp(A_1 - A_0)^u \exp(-A_0/2) B^\top \\
	&= \ B \exp(A_0/2) \exp(uA_1 - uA_0) \exp(A_0/2) B^\top \\
	&= \ B \exp((1 - u)A_0 + u A_1) B^\top .
\end{align*}
This may be applied to the curve $t \mapsto \Sigma(t)$ with $A_j = t_j A$ as well as to the surface $x \mapsto \Gamma(x)$ with $A_j = D(x_j)$.
\end{proof}

\begin{proof}[\bf Proof of Lemma~\ref{lem:minimizer.g-convex}]
If $\Sigma = BB^\top$ minimizes $f$, then obviously \eqref{eq:minimizer.g-convex} has to hold true. On the other hand, suppose that $\Sigma = BB^\top$ is not a minimizer of $f$. That means, $f(B \exp(A) B^\top) < f(BB^\top)$ for some $A \in \Rqqsym$. But $h(t) := f(B \exp(t A) B^\top)$ is a convex function of $t \in \R$, so
\[
	\lim_{t \to 0\,+} \frac{h(t) - h(0)}{t} \ \le \ h(1) - h(0) \ < \ 0 .
\]\\[-5ex]
\end{proof}

\begin{proof}[\bf Proof of Lemma~\ref{lem:g-coercivity}]
The result and its proof generalize Proposition~5.5 in \cite{Duembgen_etal_2015}. Recall first that for any $A \in \Rqqsym$, the function $\R \ni t \mapsto f(\exp(tA))$ is convex with right-sided derivative
\[
	g(t,A) \ := \ \lim_{u \to t\,+} \frac{f(\exp(uA)) - f(\exp(tA))}{u - t} .
\]
Moreover, $g(t,A)$ is non-decreasing in $t \in \R$ with limit $g(\infty,A) \in (-\infty,\infty]$ as $t \to \infty$. Thus we have to show that $f$ is g-coercive if, and only if, $g(\infty,A) > 0$ for any $A \in \Rqqsym \setminus \{0\}$.

Suppose that $f$ is not g-coercive. Then there exists a sequence $(A_k)_k$ in $\Rqqsym$ such that $\lim_{k \to \infty} \|A_k\| = \infty$ but $f(\exp(A_k)) \le C$ for all indices $k$ and some real constant $C$. Writing $A_k = \|A_k\| N_k$ for a matrix $N_k$ with norm one, we may even assume that $\lim_{k \to \infty} N_k = N$ with $N \in \Rqqsym$, $\|N\| = 1$. Now for any fixed $t > 0$,
\begin{align*}
	g(t,N) \
	&\le \ f(\exp((t+1)N)) - f(\exp(tN)) \\
	&= \ \lim_{k \to \infty}
		\bigl( f(\exp((t+1)N_k)) - f(\exp(tN_k)) \bigr) \\
	&\le \ \limsup_{k \to \infty}
		\frac{ f(\exp(\|A_k\|N_k)) - f(\exp(t N_k))}{\|A_k\| - t} \\
	&\le \ \limsup_{k \to \infty} \frac{ C - f(\exp(t N_k))}{\|A_k\| - t} \\
	&\le \ 0 .
\end{align*}
In the first and third step we used convexity of $f(\exp(tN_{(k)}))$ in $t \in \R$, the second and last step rely on continuity of $f$ and the choice of $(A_k)_k$. These considerations show that $g(\infty,N) \le 0$.

On the other hand, suppose that $f$ is g-coercive. Then for any $A \in \Rqqsym \setminus \{0\}$ and sufficiently large $r > 0$,
\[
	0
	\ < \ \frac{f(\exp(rA)) - f(I_q)}{r}
	\ = \ \frac{f(\exp(rA)) - f(\exp(0A))}{r}
	\ \le \ g(r,A)
	\ \le \ g(\infty,A) .
\]\\[-5ex]
\end{proof}

\begin{proof}[\bf Proof of Lemma~\ref{lem:existence.minimizers}]
By continuity of $f$, the set $\mathcal{S}_*$ is closed, and by g-convexity of $f$ it is g-convex.

Obviously, the set $\mathcal{S}_*$ is identical with the set of minimizers of $f$ on the closed set $\mathcal{K} := \bigl\{ \Sigma \in \Rqqsympd : f(\Sigma) \le f(I_q) \bigr\}$. If $f$ is also g-coercive, the set $\mathcal{K}$ is even compact, and $\mathcal{S}_*$ is a nonvoid and closed subset of $\mathcal{K}$, so it is compact itself.

Now suppose that $f$ has a minimizer $\Sigma_* = BB^\top$, $B \in \Rqqns$. Note that g-coercivity is equivalent to
\[
	f(B \exp(A) B^\top) \ \to \ \infty
	\quad\text{as} \ \|A\| \to \infty .
\]
This follows from the inequality
\begin{equation}
\label{eq:triangle?}
	\bigl| \|\log(B\exp(A)B^\top)\| - \|A\| \bigr|
	\ \le \ \|\log(\Sigma_*)\|
\end{equation}
which will be proved later. Now suppose that $f$ is minimal at $\Sigma_*$ but not g-coercive. That means, there exists a sequence $(A_k)_k$ in $\Rqqsym$ with $\lim_{k \to \infty} \|A_k\| = \infty$ but $f(B \exp(A_k) B^\top) \le C$ for all indices $k$ and some real constant $C$. Writing $A_k = \|A_k\| N_k$ for a matrix $N_k$ with norm one, we may even assume that $\lim_{k \to \infty} N_k = N$ with $N \in \Rqqsym$, $\|N\| = 1$. Since $h_k(t) := f(B \exp(t N_k) B^\top)$ is convex in $t \in \R$, we may conclude that for any fixed $t > 0$,
\begin{align*}
	\frac{f(B \exp(tN) B^\top) - f(\Sigma_*)}{t} \
	&= \ \lim_{k \to \infty} \frac{f(B \exp(tN_k) B^\top) - f(\Sigma_*)}{t} \\
	&= \ \lim_{k \to \infty} \frac{h_k(t) - h_k(0)}{t} \\
	&\le \ \limsup_{k \to \infty} \frac{h_k(\|A_k\|) - h_k(0)}{\|A_k\|} \\
	&= \ \limsup_{k \to \infty} \frac{f(B \exp(A_k) B^\top) - f(\Sigma_*)}{\|A_k\|} \\
	&\le \ 0 .
\end{align*}
This implies that $f(B \exp(tN) B^\top) = f(\Sigma_*)$ for all $t > 0$, so $\mathcal{S}_*$ is geodesically unbounded.

It remains to prove inequality \eqref{eq:triangle?} which is related to geodesic distances. On the one hand,
\begin{align*}
	\| \log(B \exp(A) B^\top) \| \
	&= \ d_g(I_q, B \exp(A) B^\top) \\
	&\le \ d_g(I_q,BB^\top) + d_g(BB^\top, B \exp(A) B^\top) \\
	&= \ \|\log(\Sigma_*)\| + d_g(I_q, \exp(A)) \\
	&= \ \|\log(\Sigma_*)\| + \|A\| .
\end{align*}
On the other hand,
\begin{align*}
	\|A\| \
	&= \ d_q(I_q, \exp(A)) \\
	&\le \ d_q(I_q, (B^\top B)^{-1}) + d_g((B^\top B)^{-1}, \exp(A)) \\
	&= \ d_q(I_q, B^\top B) + d_g(B^{-1} B^{-\top}, \exp(A)) \\
	&= \ \|\log(B^\top B)\| + d_g(I_q, B\exp(A)B^\top) \\
	&= \ \|\log(\Sigma_*)\| + \|\log(B\exp(A)B^\top)\| .
\end{align*}
In the last step we utilized that $B^\top B$ and $BB^\top = \Sigma_*$ have the same eigenvalues, which follows from the singular value decomposition of $B$.
\end{proof}

\begin{proof}[\bf Proof of Lemma~\ref{lem:g-convexity}]
This criterion follows from the fact that for $t,\delta \in \R$,
\[
	B \exp((t+\delta)A) B^\top
	\ = \ B_t^{} \exp(\delta A) B_t^\top
	\quad\text{with} \ B_t := B \exp((t/2)A) ,
\]
so
\[
	f \bigl( B \exp((t+\delta)A) B^\top \bigr)
	\ = \ f(B_t^{}B_t^\top)
		+ \langle A, G(B_t)\rangle \delta
		+ 2^{-1} H(A,B_t) \delta^2 + o(\delta^2)
\]
as $\delta \to 0$. By means of Lemma~\ref{lem:convexity} in Supplement~\ref{sec:Auxiliary}, this shows that $f(B \exp(tA) B^\top)$ is convex in $t \in \R$, provided that $H(A,B_t) \ge 0$ for all $t \in \R$. This convexity is strict if $H(A,B_t) > 0$ for all $t \in \R$.

If $H(A,B) < 0$ for some $B \in \Rqqns$ and $A \in \Rqqsym$, then for sufficiently small $\delta > 0$,
\[
	f(B \exp(\pm \delta A) B^\top)
	\ < \ f(BB^\top) \pm \delta \langle A, G(B)\rangle .
\]
Hence
\[
	f(BB^\top) = f(B \exp(0A) B^\top)
	\ > \ 2^{-1} f(B \exp(-\delta A) B^\top)
		+ 2^{-1} f(B \exp(\delta A) B^\top) .
\]
Thus $f(B \exp(tA) B^\top)$ is not convex in $t \in \R$, so $f$ is not geodesically convex.
\end{proof}

\subsection{Proofs for Section~\ref{subsec:Regularization}}

\begin{proof}[\bf Proof of Lemma~\ref{lem:all.about.Pi}]
That $\Sigma = I_q$ is the unique minimizer of $\Pi_k(\Sigma)$ follows from the fact that $x + x^{-1} > 2$, $\log x + x^{-1} > 1$, $x - \log x > 1$ and $(\log x)^2 > 0$ for $x \in \R_+ \setminus \{1\}$.

Note first that $f(\Sigma) := \tr(\Sigma)$ satisfies the expansion
\begin{align*}
	f(B \exp(A) B^\top) \
	&= \ f(BB^\top) + \tr(BAB^\top) + 2^{-1} \tr(BA^2 B^\top)
		+ o(\|A\|^2) \\
	&= \ f(BB^\top) + \langle A, B^\top B\rangle + 2^{-1} \langle A^2, B^\top B\rangle
		+ o(\|A\|^2) .
\end{align*}
This and Remark~\ref{rem:Inversion2} implies that $G_0(B) = B^\top B - B^{-1} B^{-\top}$ while $H_0(A,B)$ is given by $\langle A^2, B^\top B + B^{-1} B^{-\top}\rangle$. The inequality $H_0(A,B) > 0$ for $A \ne 0$ can be proved similarly as the inequality $H(A,B) \ge 0$ in Example~\ref{ex2}. In case of $A = U D(-\gamma) U^\top$ with an orthogonal matrix $U$ and a vector $\gamma \in \R^q$ with non-decreasing componnents,
\[
	\frac{d}{dt} \Pi_0(\exp(tA))
	\ = \ \sum_{i=1}^q \gamma_i (e^{t\gamma_i} - e^{-t \gamma_i})
	\ \to \ \infty
\]
as $t \to \infty$, unless $\gamma = 0$.

As to $\Pi_1$, it follows from the previous considerations and Example~\ref{ex0} that $G_1(B) = I_q - B^{-1} B^{-\top}$ and $H_1(A,B) = \langle A^2, B^{-1} B^{-\top}\rangle$. Again $H_1(A,B) > 0$ for $A \ne 0$. Moreover, if $A = U D(-\gamma) U^\top$ as before, as $t \to \infty$,
\[
	\frac{d}{dt} \Pi_1(\exp(tA))
	\ = \ \sum_{i=1}^q \gamma_i (e^{t \gamma_i} - 1)
	\ \to \ 1_{[\gamma_q > 0]}^{} \infty - \sum_{i=1}^q \gamma_i .
\]

For $\Pi_2$ the expansion is a consequence of Corollary~\ref{cor:penalties} in Supplement~\ref{sec:Auxiliary}. Just note that we may write $B = U D(\mu)^{1/2} V^\top$ with $U, V \in \Rqqorth$ and $\mu = e^\lambda$, $\lambda \in \R^q$, and
\[
	\Pi_2(B \exp(A) B^\top)
	\ = \ \Pi_2 \bigl( D(\mu)^{1/2} \exp(V^\top A V) D(\mu)^{1/2} \bigr) .
\]
Moreover, $\Pi_2(\exp(tA)) = t^2 \|A\|^2$, so $d \Pi_2(\exp(tA)) / dt = 2 t \|A\|^2$.
\end{proof}

\begin{proof}[\bf Proof of Lemma~\ref{lem:all.about.pi}]
Elementary considerations reveal that all penalty functions $\pi_k$ are scale-invariant. Next we show that a matrix $\Sigma \in \Rqqsympd$ with eigenvalues $\sigma_1 \ge \cdots \ge \sigma_q > 0$ minimizes $\pi_k(\Sigma)$ if, and only if, $\sigma_1/\sigma_q = 1$. On the one hand,
\[
	\pi_0(\Sigma)
	\ = \ \log \Bigl( \sum_{i=1}^q \sigma_i^{} \sum_{j=1}^q \sigma_j^{-1} \Bigr)
	\ = \ \log \Bigl( \frac{1}{2}
		\sum_{i,j=1}^q
			\bigl( \frac{\sigma_i}{\sigma_j} + \frac{\sigma_j}{\sigma_i} \Bigr)
		\Bigr)
	\ \ge \ \log(q^2)
\]
with equality if, and only if, $\sigma_i/\sigma_j = 1$ for all indices $i,j$. This follows from $x + x^{-1} > 2$ for arbitrary $x \in \R_+ \setminus \{1\}$. In case of $\pi_1(\Sigma)$, note that by Jensen's inequality and strict concavity of $\log$ on $\R_+$,
\[
	\pi_1(\Sigma)
	\ = \ - q^{-1} \sum_{i=1}^q \log(\sigma_i^{-1})
		+ \log \Bigl( q^{-1} \sum_{i=1}^q \sigma_i^{-1} \Bigr) + \log q
	\ \ge \ \log(q)
\]
with strict inequality unless all $\sigma_i$ are identical. Finally,
\[
	\pi_2(\Sigma)
	\ = \ \sum_{i=1}^q
		\Bigl( \log \sigma_i - q_{}^{-1} \sum_{j=1}^q \log \sigma_j \Bigr)^2
	\ \ge \ 0
\]
with equality if, and only if, all $\sigma_i$ are identical.

Next we verify the geodesic second order Taylor expansions of $\pi_k(\Sigma)$. It follows from Examples~\ref{ex0} and \ref{ex2} and Remark~\ref{rem:Inversion2} that
\begin{align*}
	G_0(B) \
	&= \ N(B^\top B) - N(B^{-1} B^{-\top}) , \\
	G_1(B) \
	&= \ q^{-1} I_q - N(B^{-1} B^{-\top}) ,
\end{align*}
and
\begin{align*}
	H_0(A,B) \
	&= \ \langle A^2, N(B^\top B)\rangle
			- \langle A, N(B^\top B)\rangle^2
		+ \langle A^2, N(B^{-1}B^{-\top})\rangle
			- \langle A, N(B^{-1}B^{-\top})\rangle^2 , \\
	H_1(A,B) \
	&= \ \langle A^2, N(B^{-1}B^{-\top})\rangle
			- \langle A, N(B^{-1}B^{-\top})\rangle^2
\end{align*}
with $N(\Sigma) := \tr(\Sigma)^{-1} \Sigma$. The considerations to Example~\ref{ex2} reveal that both $H_0(A,B)$ and $H_1(A,B)$ are strictly positive whenever $A \not\in \{t I_q : t \in \R\}$. The expansion for $\pi_2$ follows from Corollary~\ref{cor:penalties} with the same arguments as in the proof of Lemma~\ref{lem:all.about.Pi}. In particular,
\[
	H_2(A,B) \ = \ \sum_{i,j=1}^q W_{ij}(\lambda) (v_i^\top A^o v_j)^2
	\ \ge \ \|A^o\|^2
\]
with $A^o = A - q^{-1}\tr(A) I_q$.

Concerning coercivity, let $A = V D(-\gamma) V^\top$ with $\gamma_1 \le \ldots \le \gamma_q$ and $\gamma_q > \gamma_1$. Then for $\xi = \pm 1$,
\[
	\frac{d}{dt} q^{-1} \log \det(\exp(tA)^{\xi})
	\ = \ - \xi \bar{\gamma}
\]
and
\[
	\frac{d}{dt} \log \tr(\exp(tA)^{\xi})
	\ = \ - \xi \sum_{i=1}^q \gamma_i e^{- \xi t\gamma_i}
		\Big/ \sum_{i=1}^q e^{- \xi t\gamma_i}
		\ \to \ \begin{cases}
			- \gamma_1 & \text{if} \ \xi = +1 , \\
			\ \gamma_q & \text{if} \ \xi = -1 ,
		\end{cases}
\]
as $t \to \infty$. This implies for $k = 0,1$ the asserted limits of $d \pi_k(\exp(tA)) / dt$. For $k = 2$ the claim follows from
\[
	\pi_2(\exp(tA))
	\ = \ t^2 \sum_{i=1}^q (\gamma_i - \bar{\gamma})^2 .
\]\\[-5ex]
\end{proof}

\subsection{Proofs for Section~\ref{sec:Algorithm}}

Our proof of Theorem~\ref{thm:phi} is based on two elementary inequalities for the accuracy of Taylor expansions of $f$ which are derived in Supplement~\ref{sec:Auxiliary}:

\begin{Lemma}
\label{lem:remainders}
For $\Sigma \in \Rqqsympd$ and $\delta > 0$ let
\begin{align*}
	\Lambda_{\rm max}(\Sigma,\delta) \
	&:= \ \max_{A, C \in \Rqqsym \,:\, \|A\| \le 1, \|C\| \le \delta}
		H(A, \Sigma^{1/2} \exp(C/2)) , \\
	N(\Sigma,\delta) \
	&:= \ \max_{A, C \in \Rqqsym \,:\, \|A\| \le 1, \|C\| \le \delta}
		\bigl| H(A, \Sigma^{1/2} \exp(C/2)) - H(A, \Sigma^{1/2}) \bigr| .
\end{align*}
For arbitrary $\Sigma = BB^\top$ with $B \in \Rqqns$ and $A \in \Rqqsym \setminus \{0\}$,
\[
	f(B \exp(A) B^\top) - f(\Sigma) - \langle A, G(B)\rangle
	\ \le \ 2^{-1} \|A\|^2 \Lambda_{\rm max}(\Sigma,\|A\|)
\]
and
\[
	\bigl| f(B \exp(A) B^\top) - f(\Sigma)
		- \langle A, G(B)\rangle - 2^{-1} H(A,B) \bigr|
	\ \le \ 2^{-1} \|A\|^2 N(\Sigma,\|A\|) .
\]
\end{Lemma}

\begin{proof}[\bf Proof of Theorem~\ref{thm:phi}]
One can deduce from continuity of $H(A,B)$ in $B \in \Rqqns$ for fixed $A \in \Rqqsym$ and $\Rqqsym$ being finite-dimensional that both $\Lambda_{\rm max}(\Sigma,\delta)$ and $N(\Sigma,\delta)$ are continuous in $(\Sigma,\delta) \in \Rqqsympd \times [0,\infty)$, where $N(\Sigma,0) = 0$. Additional quantities we shall use repeatedly are
\[
	\Lambda_{\rm min}(\Sigma)
	\ := \ \min \bigl\{ H(A, \Sigma^{1/2}) : A \in \Rqqsym, \|A\| = 1 \bigr\} > 0
\]
and $\|G(\Sigma^{1/2})\|$. Both are continuous in $\Sigma$.

For arbitrary $\Sigma = BB^\top$, $B \in \Rqqns$, we can say that
\[
	\|A_{\rm pN}\|
	\ = \ \| \underline{H}(BU)^{-1} \lambda(B) \|
	\ \le \ \frac{\|\lambda(B)\|}{\lambda_{\rm min}(\underline{H}(BU))}
	\ \le \ \frac{\|G(\Sigma^{1/2})\|}{\Lambda_{\rm min}(\Sigma)}
	\ =: \ R_1(\Sigma) ,
\]
because $\|\lambda(B)\| = \|G(B)\| = \|G(\Sigma^{1/2})\|$ and
\[
	\lambda_{\rm min}(\underline{H}(BU))
	\ \ge \ \min \bigl\{ H(A, BU) : A \in \Rqqsym, \|A\| = 1 \bigr\}
	\ = \ \Lambda_{\rm min}(\Sigma) .
\]
On the other hand,
\[
	\langle A_{\rm pN}, G(B)\rangle
	\ = \ - \lambda(B)^\top \underline{H}(BU)^{-1} \lambda(B)
	\ \le \ - \frac{\|\lambda(B)\|^2}{\lambda_{\rm max}(\underline{H}(BU))}
	\ \le \ - \frac{\|G(\Sigma^{1/2})\|^2}{\Lambda_{\rm max}(\Sigma,0)} .
\]
Hence it follows from Lemma~\ref{lem:remainders} that for any fixed integer $m \ge 0$,
\begin{align*}
	f(B & \exp(2^{-m}A_{\rm pN}) B^\top) - f(\Sigma)
		- \langle 2^{-m} A_{\rm pN}, G(B)\rangle/C \\
	&\le \ 2^{-2m-1} \|A_{\rm pN}\|^2 \Lambda_{\rm max}(\Sigma,2^{-m} \|A_{\rm pN}\|)
		+ (1 - C^{-1}) \langle 2^{-m} A_{\rm pN}, G(B)\rangle \\
	&\le \ 2^{-m} \|G(\Sigma^{1/2})\|^2 \Bigl(
		\frac{\Lambda_{\rm max}(\Sigma, 2^{-m} R_1(\Sigma))}
			{2^{m+1} \Lambda_{\rm min}(\Sigma)^2}
		- \frac{1}{\Lambda_{\rm max}(\Sigma,0)} \Bigr) \\
	&=: \ R_{2,m}(\Sigma) .	
\end{align*}
Note that $R_{2,m}(\Sigma)$ is continuous in $\Sigma$. Moreover, for any fixed $\Sigma_o \ne \Sigma_*$ there is an integer $m_o \ge 0$ such that $R_{2,m_o}(\Sigma_o) < 0$. Consequently, if $\Sigma$ is sufficiently close to $\Sigma_o$, then the integer $m(BU)$ in $\phi(\Sigma)$ satisfies $m(BU) \le m_o$, and
\[
	f(\phi(\Sigma)) - f(\Sigma)
	\ \le \ 2_{}^{-m_o} \langle A_{\rm pN}, G(B)\rangle / C
	\ \le \ - \frac{\|G(\Sigma^{1/2})\|^2}{2_{}^{m_o} \Lambda_{\rm max}(\Sigma,0) C} .
\]
This shows that
\[
	\limsup_{\Sigma \to \Sigma_o} f(\phi(\Sigma)) - f(\Sigma_o)
	\ \le \ - \frac{\|G(\Sigma_o^{1/2})\|^2}{2_{}^{m_o} \Lambda_{\rm max}(\Sigma_o,0) C}
	\ < \ 0 .
\]

For $\Sigma$ close to $\Sigma_*$ we only consider $m = 0$ and utilize the second bound in Lemma~\ref{lem:remainders}. Namely,
\begin{align*}
	f(B \exp(A_{\rm pN}) B^\top) - f(\Sigma) \
	&= \ 2^{-1} \langle A_{\rm pN}, G(B)\rangle +
		2^{-1} \|A_{\rm pN}\|^2 N(\Sigma, \|A_{\rm pN}\|) \\
	&\le \ 2^{-1} \langle A_{\rm pN}, G(B)\rangle
		+ \|G(\Sigma^{1/2})\|^2 \frac{N(\Sigma, R_1(\Sigma))}
			{2 \Lambda_{\rm min}(\Sigma)^2} .
\end{align*}
Consequently,
\begin{align*}
	f(B & \exp(A_{\rm pN}) B^\top) - f(\Sigma)
		- \langle A_{\rm pN}, G(B)\rangle/C \\
	&\le \ (2^{-1} - C^{-1}) \langle A_{\rm pN}, G(B)\rangle
		+ \|G(\Sigma^{1/2})\|^2 \frac{N(\Sigma, R_1(\Sigma))}
			{2 \Lambda_{\rm min}(\Sigma)^2} \\
	&\le \ \|G(\Sigma^{1/2})\|^2 \Bigl(
		\frac{N(\Sigma, R_1(\Sigma))}{2 \Lambda_{\rm min}(\Sigma)^2}
		- \frac{2^{-1} - C^{-1}}{\Lambda_{\rm max}(\Sigma,0)} \Bigr) .
\end{align*}
But $R_1(\Sigma) \to 0$ as $\Sigma \to \Sigma_*$ and $N(\Sigma_*,0) = 0$, so
\[
	\lim_{\Sigma \to \Sigma_*} \Bigl(
		\frac{N(\Sigma, R_1(\Sigma))}{2 \Lambda_{\rm min}(\Sigma)^2}
		- \frac{2^{-1} - C^{-1}}{\Lambda_{\rm max}(\Sigma,0)} \Bigr)
	\ = \ - \frac{2^{-1} - C^{-1}}{\Lambda_{\rm max}(\Sigma_*,0)}
	\ < \ 0 .
\]
Consequently, $m(BU) = 0$ if $\Sigma$ is sufficiently close to $\Sigma_*$.
\end{proof}

\addcontentsline{toc}{section}{References}


\begin{thebibliography}{10}

\bibitem{Abramovich_etal_2013}
{\sc Y.~Abramovich, O.~Besson, et~al.}, {\em Regularized covariance matrix
  estimation in complex elliptically symmetric distributions using the expected
  likelihood approach—part 1: The over-sampled case}, Signal Processing, IEEE
  Transactions on, 61 (2013), pp.~5807--5818.

\bibitem{Auderset_etal_2005}
{\sc C.~Auderset, C.~Mazza, and E.~A. Ruh}, {\em Angular {G}aussian and
  {C}auchy estimation}, J. Multivar. Anal., 93 (2005), pp.~180--197.

\bibitem{Besson_etal_2013}
{\sc O.~Besson, Y.~Abramovich, et~al.}, {\em Regularized covariance matrix
  estimation in complex elliptically symmetric distributions using the expected
  likelihood approach—part 2: The under-sampled case}, Signal Processing,
  IEEE Transactions on, 61 (2013), pp.~5819--5829.

\bibitem{Bhatia_1997}
{\sc R.~Bhatia}, {\em Matrix analysis}, vol.~169 of Graduate Texts in
  Mathematics, Springer-Verlag, New York, 1997.

\bibitem{Bhatia_2007}
\leavevmode\vrule height 2pt depth -1.6pt width 23pt, {\em Positive definite
  matrices}, Princeton Series in Applied Mathematics, Princeton University
  Press, Princeton, NJ, 2007.

\bibitem{Chen_etal_2011}
{\sc Y.~Chen, A.~Wiesel, and A.~O. Hero~III}, {\em Robust shrinkage estimation
  of high-dimensional covariance matrices}, Signal Processing, IEEE
  Transactions on, 59 (2011), pp.~4097--4107.

\bibitem{Conte_etal_2002}
{\sc E.~Conte, A.~D. Maio, and G.~Ricci}, {\em Recursive estimation of the
  covariance matrix of a compound-gaussian process and its application to
  adaptive {CFAR} detection}, {IEEE} Transactions on Signal Processing, 50
  (2002), pp.~1908--1915.

\bibitem{Couillet_McKay_2014}
{\sc R.~Couillet and M.~McKay}, {\em Large dimensional analysis and
  optimization of robust shrinkage covariance matrix estimators}, Journal of
  Multivariate Analysis, 131 (2014), pp.~99--120.

\bibitem{Duembgen_1998}
{\sc L.~D{\"u}mbgen}, {\em {O}n {T}yler's {$M$}-functional of scatter in high
  dimension}, Ann. Inst. Statist. Math., 50 (1998), pp.~471--491.

\bibitem{Duembgen_etal_2016}
{\sc L.~D{\"u}mbgen, K.~Nordhausen, and H.~Schuhmacher}, {\em New algorithms
  for {M}-estimation of multivariate scatter and location}, Journal of
  Multivariate Analysis, 144 (2016), pp.~200--217.

\bibitem{Duembgen_etal_2015}
{\sc L.~D{\"u}mbgen, M.~Pauly, and T.~Schweizer}, {\em {M}-functionals of
  multivariate scatter}, Stat. Surv., 9 (2015), pp.~32--105.

\bibitem{Duembgen_Tyler_2005}
{\sc L.~D{\"u}mbgen and D.~E. Tyler}, {\em On the breakdown properties of some
  multivariate {M}-functionals}, Scand. J. Statist., 32 (2005), pp.~247--264.

\bibitem{Finegold_Drton_2011}
{\sc M.~Finegold and M.~Drton}, {\em Robust graphical modeling of gene networks
  using classical and alternative t-distributions}, The Annals of Applied
  Statistics,  (2011), pp.~1057--1080.

\bibitem{Gini_Greco_2002}
{\sc F.~Gini and M.~Greco}, {\em Covariance matrix estimation for cfar
  detection in correlated heavy tailed clutter}, Signal Processing, 82 (2002),
  pp.~1847--1859.

\bibitem{Huber_1981}
{\sc P.~J. Huber}, {\em Robust statistics}, John Wiley \& Sons, Inc., New York,
  1981.
\newblock Wiley Series in Probability and Mathematical Statistics.

\bibitem{Kent_Tyler_1991}
{\sc J.~T. Kent and D.~E. Tyler}, {\em Redescending {$M$}-estimates of
  multivariate location and scatter}, Ann. Statist., 19 (1991), pp.~2102--2119.

\bibitem{Maronna_1976}
{\sc R.~A. Maronna}, {\em Robust {$M$}-estimators of multivariate location and
  scatter}, Ann. Statist., 4 (1976), pp.~51--67.

\bibitem{Ollila_Koivunen_2003}
{\sc E.~Ollila and V.~Koivunen}, {\em Robust antenna array processing using
  m-estimators of pseudo-covariance}, in Personal, Indoor and Mobile Radio
  Communications, 2003. PIMRC 2003. 14th IEEE Proceedings on, vol.~3, IEEE,
  2003, pp.~2659--2663.

\bibitem{Ollila_Koivunen_2009}
\leavevmode\vrule height 2pt depth -1.6pt width 23pt, {\em Influence function
  and asymptotic efficiency of scatter matrix based array processors: Case mvdr
  beamformer}, Signal Processing, IEEE Transactions on, 57 (2009),
  pp.~247--259.

\bibitem{Ollila_etal_2003}
{\sc E.~Ollila, L.~Quattropani, and V.~Koivunen}, {\em Robust space-time
  scatter matrix estimator for broadband antenna arrays}, in Vehicular
  Technology Conference, 2003. VTC 2003-Fall. 2003 IEEE 58th, vol.~1, IEEE,
  2003, pp.~55--59.

\bibitem{Ollila_Tyler_2014}
{\sc E.~Ollila and D.~Tyler}, {\em Regularized m -estimators of scatter
  matrix}, Signal Processing, IEEE Transactions on, 62 (2014), pp.~6059--6070.

\bibitem{Ollila_Tyler_2012}
{\sc E.~Ollila and D.~E. Tyler}, {\em Distribution-free detection under complex
  elliptically symmetric clutter distribution}, in Sensor Array and
  Multichannel Signal Processing Workshop (SAM), 2012 IEEE 7th, IEEE, 2012,
  pp.~413--416.

\bibitem{Ollila_etal_2012}
{\sc E.~Ollila, D.~E. Tyler, V.~Koivunen, and H.~V. Poor}, {\em Complex
  elliptically symmetric distributions: Survey, new results and applications},
  Signal Processing, IEEE Transactions on, 60 (2012), pp.~5597--5625.

\bibitem{Pascal_etal_2008}
{\sc F.~Pascal, Y.~Chitour, J.-P. Ovarlez, P.~Forster, and P.~Larzabal}, {\em
  Covariance structure maximum-likelihood estimates in compound gaussian noise:
  Existence and algorithm analysis}, Signal Processing, IEEE Transactions on,
  56 (2008), pp.~34--48.

\bibitem{Pascal_etal_2014}
{\sc F.~Pascal, Y.~Chitour, and Y.~Quek}, {\em Generalized robust shrinkage
  estimator and its application to {STAP} detection problem}, Signal
  Processing, IEEE Transactions on, 62 (2014), pp.~5640--5651.

\bibitem{Ring_Wirth_2012}
{\sc W.~Ring and B.~Wirth}, {\em Optimization methods on {R}iemannian manifolds
  and their application to shape space}, SIAM J. Optim., 22 (2012),
  pp.~596--627.

\bibitem{Soloveychik_Wiesel_2013}
{\sc I.~Soloveychik and A.~Wiesel}, {\em Group symmetry and non-gaussian
  covariance estimation}, in Global Conference on Signal and Information
  Processing (GlobalSIP), 2013 IEEE, IEEE, 2013, pp.~1105--1108.

\bibitem{Sra_Hosseini_2013}
{\sc S.~Sra and R.~Hosseini}, {\em Geometric optimisation on positive definite
  matrices for elliptically contoured distributions}, Adv. Neural Inf. Proc.
  Sys., 26 (2013), pp.~2562--2570.

\bibitem{Sra_Hosseini_2015}
\leavevmode\vrule height 2pt depth -1.6pt width 23pt, {\em Conic geometric
  optimization on the manifold of positive definite matrices}, SIAM J. Opt., 25
  (2015), pp.~713--739.

\bibitem{Sun_etal_2014}
{\sc Y.~Sun, P.~Babu, and D.~P. Palomar}, {\em Regularized tyler's scatter
  estimator: Existence, uniqueness, and algorithms}, Signal Processing, IEEE
  Transactions on, 62 (2014), pp.~5143--5156.

\bibitem{Tyler_1987a}
{\sc D.~E. Tyler}, {\em A distribution-free {$M$}-estimator of multivariate
  scatter}, Ann. Statist., 15 (1987), pp.~234--251.

\bibitem{Tyler_2010}
\leavevmode\vrule height 2pt depth -1.6pt width 23pt, {\em A note on
  multivariate location and scatter statistics for sparse data sets},
  Statistics and Probability Letters, 80 (2010), pp.~1409--1413.

\bibitem{Wiesel_2012}
{\sc A.~Wiesel}, {\em Geodesic convexity and covariance estimation}, IEEE
  Trans. Signal Process., 60 (2012), pp.~6182--6189.

\bibitem{Wiesel_2012b}
\leavevmode\vrule height 2pt depth -1.6pt width 23pt, {\em Unified framework to
  regularized covariance estimation in scaled gaussian models}, Signal
  Processing, IEEE Transactions on, 60 (2012), pp.~29--38.

\bibitem{Zhang_etal_2013}
{\sc T.~Zhang, A.~Wiesel, and M.~S. Greco}, {\em Multivariate generalized
  gaussian distribution: Convexity and graphical models}, Signal Processing,
  IEEE Transactions on, 61 (2013), pp.~4141--4148.

\end{thebibliography}

\clearpage

\appendix

\section{Further Proofs and Auxiliary Results}
\label{sec:Auxiliary}

\subsection{Various expansions for matrix exponentials and logarithms}

The next three lemmas provide expansions and inequalities for matrix exponentials and logarithms. They involve the auxiliary function $J : \R \times \R \to \R$ given by
\[
	J(x,y) \ := \ \int_0^1 e_{}^{(1 - u)x + uy} \, du
	\ = \ \begin{cases}
		(e^y - e^x)/(y - x) & \text{if} \ x \ne y , \\
		e^x                 & \text{if} \ x = y .
	\end{cases}
\]
One may also write $J(x,y) = \Ex e^{(1 - U)x + Uy}$ with a random variable $U$ which is uniformly distributed on $[0,1]$. Convexity of the exponential function on $\R$ and Jensen's inequality imply that
\begin{equation}
\label{ineq:J(x,y)}
	e_{}^{(x + y)/2} \ \le \ J(x,y) \ \le \ (e^x + e^y)/2 .
\end{equation}

\begin{Lemma}[1st order Taylor expansions of matrix exponentials and logarithms]
\label{lem:derivatives.exp.log}
For a vector $\lambda \in \R^q$ and a matrix $V = [v_1,v_2,\ldots,v_q] \in \Rqqorth$ let $A = V D(\lambda) V^\top$ and $B = \exp(A) = V D(e^\lambda) V^\top$. Then as $\Rqqsym \ni \Delta \to 0$,
\begin{align*}
	\exp(A + \Delta) \
	&= \ \exp(A) + V \,
		\bigl( J(\lambda_i,\lambda_j) \, v_i^\top \Delta_{ij}^{} v_j^{} \bigr)_{i,j=1}^q
		\, V^\top
		+ o(\|\Delta\|)
\intertext{and}
	\log(B + \Delta) \
	&= \ \log(B)
		+ V \, \Bigl( \frac{v_i^\top \Delta v_j^{}}{J(\lambda_i,\lambda_j)}
			\Bigr)_{i,j=1}^q \, V^\top
		+ o(\|\Delta\|) .
\end{align*}
\end{Lemma}

These expansions may be viewed as special cases of the Daleckii-Krein formula; cf.\ Chapter~V of \cite{Bhatia_1997} and Chapter~2 of \cite{Bhatia_2007}. We provide a more direct proof starting from a particular series expansion of matrix exponentials in \cite{Duembgen_etal_2015}. The explicit formula for the derivative of the exponential transform of $\Rqqsym$ implies local Lipschitz constants.

\begin{Lemma}[Lipschitz properties of matrix exponentials and logarithms]
\label{lem:Lipschitz.exp.log}
For arbitrary different matrices $A,B \in \Rqqsym$,
\[
	\frac{\bigl\| \exp(B) - \exp(A) \bigr\|}{\|B - A\|} \
	\begin{cases}
		\le \ J \bigl( \lambda_{\rm max}(A), \lambda_{\rm max}(B) \bigr)
		\ \le \ \max \bigl\{ e_{}^{\lambda_{\rm max}(A)},
			e_{}^{\lambda_{\rm max}(B)} \bigr\} ,
		\\[1ex]
		\ge \ J \bigl( \lambda_{\rm min}(A), \lambda_{\rm min}(B) \bigr)
		\ \ge \ \min \bigl\{ e_{}^{\lambda_{\rm min}(A)},
			e_{}^{\lambda_{\rm min}(B)} \bigr\} .
	\end{cases}
\]
For arbitrary different matrices $A, B \in \Rqqsympd$,
\[
	\frac{\bigl\| \log(B) - \log(A) \bigr\|}{\|B - A\|} \
	\begin{cases}
		\displaystyle
		\le \ \frac{1}{J \bigl( \log\lambda_{\rm min}(A), \log\lambda_{\rm min}(B) \bigr)}
		\ \le \ \max \Bigl\{ \frac{1}{\lambda_{\rm min}(A)},
			\frac{1}{\lambda_{\rm min}(B)} \Bigr\} ,
		\\[2.5ex]
		\displaystyle
		\ge \ \frac{1}{J \bigl( \log\lambda_{\rm max}(A), \log\lambda_{\rm max}(B) \bigr)}
		\ \ge \ \min \Bigl\{ \frac{1}{\lambda_{\rm max}(A)},
			\frac{1}{\lambda_{\rm max}(B)} \bigr\} .
	\end{cases}
\]
\end{Lemma}

In connection with two particular penalties we need second order Taylor expansions of matrix exponentials and logarithms. In addition to the bivariate function $J(\cdot,\cdot)$ these involve the trivariate function $J : \R \times \R \times \R \to \R$ with
\[
	J(x,y,z)
	\ := \ \int_{\{u \in [0,1]^2 : u_1 + u_2 \le 1\}}
		\exp(u_0 x + u_1 y + u_2 z) \, du \qquad (\text{with} \ u_0 := 1 - u_1 - u_2) .
\]
One may also write $J(x,y,z) = 2^{-1} \Ex e^{U_0 x + U_1 y + U_2 z}$, where $(U_0,U_1,U_2)$ is uniformly distributed on the unit simplex of all triples $(u_0,u_1,u_2) \in [0,1]^3$ with $u_0 + u_1 + u_2 = 1$. Again one can deduce from convexity of the exponential function and Jensen's inequality that
\begin{equation}
\label{ineq:J(x,y,z)}
	e_{}^{(x+y+z)/3}/2 \ \le \ J(x,y,z) \ \le \ (e^x + e^y + e^z)/6 .
\end{equation}
Another useful identity which will be used later is
\begin{equation}
\label{eq:J(x,y,z)}
	J(x,y,z) \ = \ \frac{J(x,z) - J(y,z)}{x - y}
	\quad\text{if} \ x \ne y .
\end{equation}
For
\begin{align*}
	J(x,y,z) \
	&= \ \int_0^1 \int_0^{1-u} \exp((1 - u - v) x + v y + u z) \, dv
		\, du \\
	&= \ \int_0^1 \Bigl( \frac{\exp((1 - u - v)x + v y + uz)}{y - x} \Bigr) \Big|_{v=0}^{1-u}
		\, du \\
	&= \ \int_0^1 \frac{\exp((1 - u)y + uz) - \exp((1 - u)x + uz)}{y - x}
		\, du \\
	&= \ \frac{J(y,z) - J(x,z)}{y - x} .
\end{align*}

\begin{Lemma}[2nd order Taylor expansions of matrix exponentials and logarithms]
\label{lem:derivatives.exp.log.2}
Let $\lambda \in \R^q$ and $\mu = e^\lambda \in \R_+^q$. Then, as $\Rqqsym \ni \Delta \to 0$,
\begin{align*}
	\exp( &D(\lambda) + \Delta) \\
	&= \ D(\mu) + \bigl( J(\lambda_i,\lambda_j) \, \Delta_{ij} \bigr)_{i,j=1}^q
		+ \sum_{z=1}^q  \bigl( J(\lambda_i,\lambda_z,\lambda_j)
		\, \Delta_{iz} \Delta_{zj} \bigr)_{i,j=1}^q
		+ O(\|\Delta\|^3)
\intertext{and}
	\log( &D(\mu) + \Delta) \\
	&= \ D(\lambda) + \Bigl( \frac{\Delta_{ij}}{J(\lambda_i,\lambda_j)} \Bigr)_{i,j=1}^q
		- \sum_{z=1}^q \Bigl(
			\frac{J(\lambda_i,\lambda_z,\lambda_j) \, \Delta_{iz} \Delta_{zj}}
				{J(\lambda_i,\lambda_j) J(\lambda_i,\lambda_z) J(\lambda_z,\lambda_j)}
			\Bigr)_{i,j=1}^q
		+ O(\|\Delta\|^3) .
\end{align*}
\end{Lemma}

\begin{Corollary}[Geodesic 2nd order Taylor expansion of matrix logarithms]
\label{cor:g-derivatives.log}
Let $\lambda \in \R^q$ and $\mu = e^\lambda \in \R_+^q$. Then, as $\Rqqsym \ni A \to 0$,
\begin{align*}
	\log \bigl( D(\mu)^{1/2} & \exp(A) D(\mu)^{1/2} \bigr) \\
	\ = \ D(\lambda) \
	&+ \ \biggl( \frac{\sqrt{\mu_i\mu_j} \, A_{ij}}{J(\lambda_i,\lambda_j)}
			\biggr)_{i,j=1}^q \\
	&+ \ \sum_{z=1}^q \biggl( \frac{\sqrt{\mu_i\mu_j} \, A_{iz}A_{zj}}
			{J(\lambda_i,\lambda_j)}
			\Bigl( \frac{1}{2} - \frac{J(\lambda_i,\lambda_z,\lambda_j) \mu_z}
				{J(\lambda_i,\lambda_z) J(\lambda_z,\lambda_j)} \Bigr)
		\biggr)_{i,j=1}^q + O(\|A\|^3) .
\end{align*}
\end{Corollary}

\begin{Corollary}[Two particular penalties]
\label{cor:penalties}
For $\Sigma \in \Rqqsympd$ let $\Pi(\Sigma) := \|\log(\Sigma)\|^2$ and $\pi(\Sigma) := \Pi \bigl( (\det \Sigma)^{-1/q} \Sigma \bigr) = \Pi(\Sigma) - (\log \det(\Sigma))^2/q$. For arbitrary vectors $\mu = e^\lambda$ with $\lambda \in \R^q$ and matrices $A \in \Rqqsym$, as $A \to 0$,
\begin{align*}
	\Pi \bigl( D(\mu)^{1/2} \exp(A) D(\mu)^{1/2} \bigr) \
	&= \ \|\lambda\|^2 + 2 \sum_{i=1}^q \lambda_i A_{ii}
		+ \sum_{i,j=1}^q W_{ij}(\lambda) \, A_{ij}^2
		+ O(\|A\|^3) \\
\intertext{and}
	\pi \bigl( D(\mu)^{1/2} \exp(A) D(\mu)^{1/2} \bigr) \
	&= \ \|\lambda^o\|^2 + 2 \sum_{i=1}^q \lambda_i^o A_{ii}^o
		+ \sum_{i,j=1}^q W_{ij}(\lambda) (A_{ij}^o)^2
		+ O(\|A\|^3) ,
\end{align*}
where $\lambda^o := (\lambda_i - \bar{\lambda})_{i=1}^q$ with $\bar{\lambda} := q^{-1} \sum_{i=1}^q \lambda_i$, $A^o := A - (\tr(A)/q) I_q$, and
\[
	W_{ij}(\lambda)
	\ := \ \frac{\mu_i + \mu_j}{2 J(\lambda_i,\lambda_j)}
	\ \ge \ 1 .
\]
\end{Corollary}

\noindent
An alternative expression for $W_{ij}(\lambda)$ is
\[
	W_{ij}(\lambda)
	\ = \ \frac{(\lambda_i - \lambda_j)/2}{\tanh((\lambda_i - \lambda_j)/2)}
\]
with the convention $0/\tanh(0) := 1$.

\begin{proof}[\bf Proof of Lemma~\ref{lem:derivatives.exp.log}]
It is wellknown that the mapping $\exp : \Rqqsym \to \Rqqsympd$ is bijective with inverse function $\log : \Rqqsympd \to \Rqqsym$. Moreover, the exponential mapping is continuously differentiable with derivative $G_A(\cdot)$ at $A \in \Rqqsym$, where $G_A(\cdot)$ denotes the linear mapping
\[
	\Rqqsym \ni \Delta
	\ \mapsto \ G_A(\Delta) := \int_0^1 \exp((1-u)A) \Delta \exp(A) \, du ,
\]
see \cite{Duembgen_etal_2015}. By means of the spectral representation $A = V D(\lambda) V^\top$ one may write
\begin{align*}
	G_A(\Delta) \
	&= \ \int_0^1
		V \exp((1-u)D(\lambda)) V^\top \Delta V \exp(u D(\lambda)) V^\top \, du \\
	&= \ V \, \int_0^1 \bigl( e_{}^{(1 - u)\lambda_i + u \lambda_j} \,
		v_i^\top \Delta v_j^{} \bigr)_{i,j=1}^q \, du \, V^\top \\
	&= \ V \,
		\bigl( J(\lambda_i,\lambda_j) \, v_i^\top\Delta v_j^\top \bigr)_{i,j=1}^q
		\, V^\top .
\end{align*}
Since $J(x,y) > 0$ for arbitrary $x,y \in \R$, this representation shows that $G_A(\cdot)$ is a non-singular linear transformation of $\Rqqsym$ with inverse
\[
	G_A^{-1}(\Delta) \ = \ V \, \Bigl(
		\frac{v_i^\top \Delta v_j^{}}{J(\lambda_i,\lambda_j)}
			\Bigr)_{i,j=1}^q \, V^\top .
\]
By the inverse function theorem, the function $\log : \Rqqsympd \to \Rqqsym$ is also continuously differentiable with
\[
	\log(B + \Delta) \ = \ \log(B)
		+ G_{\log(B)}^{-1}(\Delta) + o(\|\Delta\|)
	\quad\text{as} \ \Delta \to 0 ,
\]
and $G_{\log(B)}^{-1}(\Delta) = G_{A}^{-1}(\Delta)$.
\end{proof}

\begin{proof}[\bf Proof of Lemma~\ref{lem:Lipschitz.exp.log}]
We first prove the inequalities for $\exp(B) - \exp(A)$. With $\Delta := B - A$ it follows from Lemma~\ref{lem:derivatives.exp.log} and its proof that
\[
	\exp(B) - \exp(A)
	\ = \ \int_0^1 \frac{d}{dt} \exp(A + t\Delta) \, dt
	\ = \ \int_0^1 G_{A + t\Delta}(\Delta) \, dt .
\]
Writing $A + t\Delta = V D(\lambda) V^\top$ with a vector $\lambda \in \R^q$ and a matrix $V = [v_1, v_2, \ldots, v_q] \in \Rqqorth$,
\[
	G_{A + t\Delta}(\Delta)
	\ = \ V \, \bigl( J(\lambda_i,\lambda_j) \tilde{\Delta}_{ij} \bigr)_{i,j=1}^q \, V^\top
\]
with $\tilde{\Delta} := V^\top \Delta V$. On the one hand, the latter representation of $G_{A + t\Delta}(\Delta)$ and \eqref{ineq:J(x,y)} imply that
\[
	\|G_{A + t\Delta}(\Delta)\|^2
	\ = \ \sum_{i,j=1}^q J(\lambda_i,\lambda_j)^2 \tilde{\Delta}_{ij}^2
	\ \le \ e_{}^{2\lambda_{\rm max}(A + t\Delta)} \|\tilde{\Delta}\|^2
	\ = \ e_{}^{2\lambda_{\rm max}(A + t\Delta)} \|\Delta\|^2 ,
\]
and
\[
	\lambda_{\rm max}(A + t\Delta)
	\ = \ \max_{v \in \R^q \,:\, \|v\|=1} v^\top ((1 - t) A + tB) v
	\ \le \ (1 - t) \lambda_{\rm max}(A) + t \lambda_{\rm max}(B) .
\]
Consequently,
\begin{align*}
	\|\exp(B) - \exp(A)\| \
	&\le \ \int_0^1 \|G_{A + t\Delta}(\Delta)\| \, dt \\
	&\le \ \int_0^1 e_{}^{(1 - t)\lambda_{\rm max}(A) + t \lambda_{\rm max}(B)} \, dt
		\, \|\Delta\| \\
	&= \ J(\lambda_{\rm max}(A),\lambda_{\rm max}(B)) \|\Delta\|
		\ \le \ e_{}^{\max\{\lambda_{\rm max}(A),\lambda_{\rm max}(B)\}} \|\Delta\| .
\end{align*}
On the other hand, the explicit representation of $G_{A + t\Delta}(\Delta)$ and \eqref{ineq:J(x,y)} imply that
\[
	\langle G_{A + t\Delta}(\Delta), \Delta\rangle
	\ = \ \sum_{i,j=1}^q J(\lambda_i,\lambda_j) \tilde{\Delta}_{ij}^2
	\ \ge \ e_{}^{\lambda_{\rm min}(A + t\Delta)} \|\tilde{\Delta}\|^2
	\ = \ e_{}^{\lambda_{\rm min}(A + t\Delta)} \|\Delta\|^2 ,
\]
and
\[
	\lambda_{\rm min}(A + t\Delta)
	\ = \ \min_{v \in \R^q \,:\, \|v\|=1} v^\top ((1 - t) A + tB) v
	\ \ge \ (1 - t) \lambda_{\rm min}(A) + t \lambda_{\rm min}(B) .
\]
Hence
\begin{align*}
	\|\exp(B) - \exp(A)\| \
	&\ge \ \|\Delta\|^{-1} \langle \exp(B) - \exp(A), \Delta\rangle \\
	&= \ \|\Delta\|^{-1} \int_0^1 \langle G_{A + t\Delta}(\Delta), \Delta\rangle \, dt \\
	&\ge \ \int_0^1 e_{}^{(1 - t)\lambda_{\rm min}(A) + t \lambda_{\rm min}(B)} \, dt
		\, \|\Delta\| \\
	&= \ J(\lambda_{\rm min}(A),\lambda_{\rm min}(B)) \|\Delta\|
		\ \ge \ e_{}^{\min\{\lambda_{\rm min}(A),\lambda_{\rm min}(B)\}} \|\Delta\| .
\end{align*}

The inequalities for $\exp(B) - \exp(A)$ imply the inequalities for $\log(B) - \log(A)$, because $\tilde{A} := \log(A)$ and $\tilde{B} := \log(B)$ satisfy $A = \exp(\tilde{A})$, $\lambda_{{\rm min}/{\rm max}}(\tilde{A}) = \log \lambda_{{\rm min}/{\rm max}}(A)$ and $B = \exp(\tilde{B})$, $\lambda_{{\rm min}/{\rm max}}(\tilde{B}) = \log \lambda_{{\rm min}/{\rm max}}(B)$.
\end{proof}

\begin{proof}[\bf Proof of Lemma~\ref{lem:derivatives.exp.log.2}]
As shown in \cite{Duembgen_etal_2015},
\[
	\exp(A + \Delta) \ = \ \exp(A) + G_A(\Delta) + H_A(\Delta) + O(\|\Delta\|^3) ,
\]
where $G_A(\Delta)$ is defined as in the proof of Lemma~\ref{lem:derivatives.exp.log}, and
\[
	H_A(\Delta) \ := \ \int_{\{u \in [0,1]^2 : u_1 + u_2 \le 1\}}
		\exp(u_0 A) \Delta \exp(u_1 A) \Delta \exp(u_2 A) \, du
\]
with $u_0 := 1 - u_1 - u_2$. In the special case of a diagonal matrix $A = D(\lambda)$, the matrix $G_A(\Delta)$ equals $\bigl( J(\lambda_i,\lambda_j) \Delta_{ij} \bigr)_{i,j=1}^q$, and the matrix $\exp(u_0 A) \Delta \exp(u_1 A) \Delta \exp(u_2 A)$ may be written as
\[
	\sum_{z=1}^q \bigl(
		\exp(u_0 \lambda_i + u_1 \lambda_z + u_2 \lambda_j)
			\Delta_{iz} \Delta_{zj} \bigr)_{i,j=1}^q ,
\]
so
\[
	H_A(\Delta) \ = \ \sum_{z=1}^q \bigl( J(\lambda_i,\lambda_z,\lambda_j)
		\Delta_{iz} \Delta_{zj} \bigr)_{i,j=1}^q .
\]
This proves the second order Taylor expansion for $\exp(A + \Delta)$.

Concerning the expansion of $\log(B + \Delta)$ with $B = \exp(A) = D(\mu)$, we determine a matrix $E = E(A,\Delta) \in \Rqqsym$ such that
\[
	\exp(A + E) \ = \ B + \Delta + O(\|\Delta\|^3) .
\]
To this end, recall that
\[
	\exp(A + E)
	\ = \ B + G_A(E) + H_A(E) + O(\|E\|^3)
\]
as $E \to 0$. Thus we set
\[
	E \ := \ G_A^{-1}(\Delta) - G_A^{-1} \bigl( H_A(G_A^{-1}(\Delta)) \bigr)
\]
and note that
\begin{align*}
	G_A^{-1}(\Delta) \
	&= \ \Bigl(
		\frac{\Delta_{ij}}{J(\lambda_i,\lambda_j)} \Bigr)_{i,j=1}^q
		\ = \ O(\|\Delta\|) , \\
	G_A^{-1} \bigl( H_A(G_A^{-1}(\Delta)) \bigr) \
	&= \ \sum_{z=1}^q \Bigl(
		\frac{J(\lambda_i,\lambda_z,\lambda_j) \, \Delta_{iz}\Delta_{zj}}
			{J(\lambda_i,\lambda_j) J(\lambda_i,\lambda_z) J(\lambda_z,\lambda_j)}
				\Bigr)_{i,j=1}^q
		\ = \ O(\|\Delta\|^2) ,
\end{align*}
so
\[
	E \ = \ G_A^{-1}(\Delta) + O(\|\Delta\|^2) \ = \ O(\|\Delta\|) .
\]
Moreover, one can easily verify that $H_A \bigl( G_A^{-1}(\Delta) + O(\|\Delta\|^2) \bigr) = H_A(G_A^{-1}(\Delta)) + O(\|\Delta\|^3)$, whence
\begin{align*}
	\exp(A + E) \
	&= \ B + G_A(E) + H_A(E) + O(\|\Delta\|^3) \\
	&= \ B + G_A(E) + H_A(G_A^{-1}(\Delta)) + O(\|\Delta\|^3) \\
	&= \ B + \Delta + O(\|\Delta\|^3) .
\end{align*}
In other words,
\[
	\exp(A + E) - \exp(\log(B + \Delta)) \ = \ O(\|\Delta\|^3) .
\]
But now it follows from Lemma~\ref{lem:Lipschitz.exp.log} and the continuity of eigenvalues that
\[
	\log(B + \Delta) \ = \ A + E + O(\|\Delta\|^3) .
\]\\[-5ex]
\end{proof}

\begin{proof}[\bf Proof of Corollary~\ref{cor:g-derivatives.log}]
This expansion follows essentially from Lemma~\ref{lem:derivatives.exp.log.2} with
\begin{align*}
	\Delta \
	&= \ D(\mu)^{1/2} (\exp(A) - I_q) D(\mu)^{1/2} \\
	&= \ D(\mu)^{1/2}A D(\mu)^{1/2} + 2^{-1} D(\mu)^{1/2} A^2 D(\mu)^{1/2} + O(\|A\|^3) \\
	&= \ \bigl( \sqrt{\mu_i\mu_j} \, A_{ij} \bigr)_{i,j=1}^q
		+ 2_{}^{-1} \sum_{z=1}^q \bigl( \sqrt{\mu_i\mu_j} \, A_{iz}A_{zj} \bigr)_{i,j=1}^q
		+ O(\|A\|^3) \\
	&= \ \bigl( \sqrt{\mu_i\mu_j} \, A_{ij} \bigr)_{i,j=1}^q
		+ O(\|A\|^2) \ = \ O(\|A\|) .
\end{align*}
So $D(\mu)^{1/2} \exp(A) D(\mu)^{1/2} = D(\mu) + \Delta$, and the Taylor expansion in Lemma~\ref{lem:derivatives.exp.log.2} involves matrices with entries
\begin{align*}
	\frac{\Delta_{ij}}{J(\lambda_i,\lambda_j)} \
	&= \ \frac{\sqrt{\mu_i\mu_j} \, A_{ij}}{J(\lambda_i,\lambda_j)}
		+ \sum_{z=1}^q
			\frac{\sqrt{\mu_i\mu_j} \, A_{iz}A_{zj}}
				{J(\lambda_i,\lambda_j)}
			\cdot \frac{1}{2}
		+ O(\|A\|^3) , \\
	\frac{J(\lambda_i,\lambda_z,\lambda_j) \, \Delta_{iz}\Delta_{zj}}
			{J(\lambda_i,\lambda_j) J(\lambda_i,\lambda_z) J(\lambda_z,\lambda_j)} \
	&= \ \frac{\sqrt{\mu_i\mu_j} \, A_{iz}A_{zj}}
			{J(\lambda_i,\lambda_j)}
		\cdot \frac{J(\lambda_i,\lambda_z,\lambda_j) \mu_z}
			{J(\lambda_i,\lambda_z) J(\lambda_z,\lambda_j)}
		+ O(\|A\|^3) .
\end{align*}\\[-5ex]
\end{proof}

\begin{proof}[\bf Proof of Corollary~\ref{cor:penalties}]
According to Corollary~\ref{cor:g-derivatives.log},
\[
	\log \bigl( D(\mu)^{1/2} \exp(A) D(\mu)^{1/2} \bigr)
	\ = \ D(\lambda) + L(\lambda,A) + Q(\lambda,A) + O(\|A\|^3)
\]
with
\[
	L(\lambda,A)_{ij}
	\, = \, \frac{\sqrt{\mu_i\mu_j} \, A_{ij}}{J(\lambda_i,\lambda_j)}
	\quad\text{and}\quad
	Q(\lambda,A)_{ij}
	\, = \, \sum_{z=1}^q \frac{\sqrt{\mu_i\mu_j} \, A_{iz}A_{zj}}
			{J(\lambda_i,\lambda_j)}
			\Bigl( \frac{1}{2} - \frac{J(\lambda_i,\lambda_z,\lambda_j) \mu_z}
				{J(\lambda_i,\lambda_z) J(\lambda_z,\lambda_j)} \Bigr) .
\]
In particular, since $J(\lambda_i,\lambda_i) = \mu_i$,
\[
	L(\lambda,A)_{ii}
	\ = \ A_{ii}
	\quad\text{and}\quad
	Q(\lambda,A)_{ii}
	\ = \ \sum_{z=1}^q A_{iz}^2
		\Bigl( \frac{1}{2} - \frac{J(\lambda_i,\lambda_z,\lambda_i) \mu_z}
				{J(\lambda_i,\lambda_z)^2} \Bigr) .
\]
Hence
\begin{align*}
	\Pi \bigl( & D(\mu)^{1/2} \exp(A) D(\mu)^{1/2} \bigr) \\
	&= \ \|D(\lambda)\|^2
		+ 2 \langle D(\lambda), L(\lambda,A)\rangle
		+ \|L(\lambda,A)\|^2
		+ 2 \langle D(\lambda), Q(\lambda,A)\rangle
		+ O(\|A\|^3) \\
	&= \ \|\lambda\|^2
		+ 2 \sum_{i=1}^q \lambda_i A_{ii}
		+ \|L(\lambda,A)\|^2
		+ 2 \langle D(\lambda), Q(\lambda,A)\rangle
		+ O(\|A\|^3) .
\end{align*}
Moreover,
\begin{align*}
	\|L(\lambda,A)\|^2
		+ 2 \langle D(\lambda), Q(\lambda,A)\rangle
	&= \ \sum_{i,j=1}^q \frac{\mu_i\mu_j A_{ij}^2}{J(\lambda_i,\lambda_j)^2}
		+ \sum_{i,z=1}^q A_{iz}^2 \lambda_i
			\Bigl( 1 - \frac{2 J(\lambda_i,\lambda_z,\lambda_i) \mu_z}
				{J(\lambda_i,\lambda_z)^2} \Bigr) \\
	&= \ \sum_{i,j=1}^q W_{ij}(\lambda) A_{ij}^2
\end{align*}
with
\begin{align*}
	W_{ij}(\lambda) \
	:=& \ \frac{\mu_i\mu_j}{J(\lambda_i,\lambda_j)^2}
		+ \lambda_i \Bigl( \frac{1}{2} - \frac{J(\lambda_i,\lambda_j,\lambda_i) \mu_j}
				{J(\lambda_i,\lambda_j)^2} \Bigr)
		+ \lambda_j \Bigl( \frac{1}{2} - \frac{J(\lambda_j,\lambda_i,\lambda_j) \mu_i}
				{J(\lambda_i,\lambda_j)^2} \Bigr)
\end{align*}
Now we have to show that
\begin{equation}
\label{eq:Wij}
	W_{ij}(\lambda)
	\ = \ \frac{\mu_i + \mu_j}{2 J(\lambda_i,\lambda_j)}
	\ \ge \ 1 .
\end{equation}
The inequality is just a consequence of \eqref{ineq:J(x,y)}. In case of $\lambda_i = \lambda_j$, the equation in \eqref{eq:Wij} follows from $J(\lambda_i,\lambda_i) = \mu_i$ and $J(\lambda_i,\lambda_i,\lambda_i) = \mu_i/2$, and here $W_{ij}(\lambda) = 1$. In case of $\lambda_i \ne \lambda_j$ we use \eqref{eq:J(x,y,z)} and obtain
\begin{align*}
	W_{ij}(\lambda) \
	&= \ \frac{\mu_i\mu_j}{J(\lambda_i,\lambda_j)^2}
		+ \lambda_i \Bigl( \frac{1}{2} - \frac{(\mu_i - J(\lambda_i,\lambda_j)) \mu_j}
				{(\lambda_i - \lambda_j) J(\lambda_i,\lambda_j)^2} \Bigr)
		+ \lambda_j \Bigl( \frac{1}{2} - \frac{(J(\lambda_i,\lambda_j) - \mu_j) \mu_i}
				{(\lambda_i - \lambda_j) J(\lambda_i,\lambda_j)^2} \Bigr) \\
	&= \ \lambda_i \Bigl( \frac{1}{2} + \frac{\mu_j}
				{(\lambda_i - \lambda_j) J(\lambda_i,\lambda_j)} \Bigr)
		+ \lambda_j \Bigl( \frac{1}{2} - \frac{\mu_i}
				{(\lambda_i - \lambda_j) J(\lambda_i,\lambda_j)} \Bigr) \\
	&= \ \lambda_i \Bigl( \frac{1}{2} + \frac{\mu_j}{\mu_i - \mu_j} \Bigr)
		+ \lambda_j \Bigl( \frac{1}{2} - \frac{\mu_i}{\mu_i - \mu_j} \Bigr) \\
	&= \ \frac{\mu_i + \mu_j}{2 J(\lambda_i,\lambda_j)} .
\end{align*}

Concerning the function $\pi(\cdot)$, note first that
\[
	\log \det \bigl( D(\mu)^{1/2} \exp(A) D(\mu)^{1/2} \bigr)
	\ = \ \sum_{i=1}^q \lambda_i + \tr(A)
	\ = \ q \bar{\lambda} + \tr(A) ,
\]
so
\begin{align*}
	\pi \bigl( & D(\mu)^{1/2} \exp(A) D(\mu)^{1/2} \bigr) \\
	&= \ \Pi \bigl( D(\mu)^{1/2} \exp(A) D(\mu)^{1/2} \bigr)
		- q \bar{\lambda}^2 - 2 \bar{\lambda} \tr(A) - \tr(A)^2/q \\
	&= \ \|\lambda\|^2 + 2 \sum_{i=1}^q \lambda_i A_{ii}
		+ \sum_{i,j=1}^q W_{ij}(\lambda) A_{ij}^2
		- q \bar{\lambda}^2 - 2 \bar{\lambda} \tr(A) - \tr(A)^2/q + O(\|A\|^3) \\
	&= \ \|\lambda^o\|^2 + 2 \sum_{i=1}^q \lambda_i^o A_{ii}^o
		+ \sum_{i,j=1}^q W_{ij}(\lambda) (A_{ij}^o)^2
		+ O(\|A\|^3) .
\end{align*}
The last step follows from elementary algebra and the facts that $W_{ii}(\lambda) = 1$ and $A_{ij}^o = A_{ij}$ whenever $i \ne j$.
\end{proof}

\subsection{Proof of Theorem~\ref{thm:geodesics}}

The following arguments are similar to the ones of \cite{Bhatia_2007}. In case of $\Sigma_0 = \Sigma_1$, the assertion is trivial, so we only consider the case $\Sigma_0 \ne \Sigma_1$. Without loss of generality let $\Sigma_0 = I_q$, otherwise consider the path $M_B$ with $B = \Sigma_0^{-1/2}$. Now let
\[
	A \ := \ \|\log(\Sigma_1)\|^{-1} \log(\Sigma_1) .
\]
Then we may write
\begin{align*}
	\|\log(\Sigma_1)\| \
	&= \ \bigl\langle A, \log(M(1)) - \log(M(0)) \bigr\rangle \\
	&= \ \int_0^1 \Bigl\langle A, \frac{d}{dt} \log(M(t)) \Bigr\rangle \, dt \\
	&\le \ \int_0^1 \Bigl\| \frac{d}{dt} \log(M(t)) \Bigr\| \, dt
\end{align*}
by virtue of the Cauchy-Schwarz inequality.

Equality holds in the latter display if, and only if, the derivative of $\log(M(t))$ is a non-negative multiple of $A$ for almost all $t \in [0,1]$. Since $\log(M(t))$ is continuously differentiable by assumption, we may rephrase this as
\[
	\frac{d}{dt} \log(M(t))
	\ = \ \dot{u}(t) \log(\Sigma_1)
\]
for some bounded function $\dot{u} : [0,1] \to [0,\infty)$ with at most finitely many discontinuities. Since $\log(M(0)) = 0$ and $\log(M(1)) = \log(\Sigma_1)$, we know that $u(t) := \int_0^t \dot{u}(s) \, ds$ defines a nondecreasing, piecewise continuously differentiable function $u : [0,1] \to \R$ with $u(0) = 0$, $u(1) = 1$ and $M(t) = \exp(u(t) \log(\Sigma_1))$ for $t \in [a,b]$. Note also that in this special case
\[
	\dot{M}(t) \ = \ \dot{u}(t) \log(\Sigma_1) M(t)
	\ = \ M(t)_{}^{1/2} \bigl( \dot{u}(t) \log(\Sigma_1) \bigr) M(t)_{}^{1/2} ,
\]
so $L(M) = \|\log(\Sigma_1)\|$.

Hence it suffices to show that for a general path $M$ and any $t \in [0,1]$,
\[
	\Bigl\| \frac{d}{dt} \log(M(t)) \Bigr\|
	\ \le \ \|\dot{M}(t)\|_{M(t)} .
\]
To this end we write $M(t) = V D(\mu) V^\top$ with an orthogonal matrix $V = [v_1,v_2,\ldots,v_q] \in \Rqq$ and a vector $\mu \in \R_+^q$. Then it follows from Lemma~\ref{lem:derivatives.exp.log} that
\[
	\frac{d}{dt} \log(M(t))
	\ = \ V
		\Bigl( \frac{v_i^\top \dot{M}(t) v_j^{}}{J(\log\mu_i,\log\mu_j)}
			\Bigr)_{i,j=1}^q V^\top
\]
with $\sqrt{\mu_i\mu_j} \le J(\log\mu_i,\log\mu_j) \le (\mu_i + \mu_j)/2$. On the other hand,
\[
	M(t)^{-1/2} \dot{M}(t) M(t)^{-1/2}
	\ = \ V
		\Bigl( \frac{v_i^\top \dot{M}(t) v_j^{}}{\sqrt{\mu_i \mu_j}}
			\Bigr)_{i,j=1}^q V^\top .
\]
Consequently,
\begin{align*}
	\Bigl\| \frac{d}{dt} \log(M(t)) \Bigr\|^2 \
	&= \ \sum_{i,j=1}^q
		\frac{(v_i^\top \dot{M}(t) v_j^{})^2}{J(\log\mu_i,\log\mu_j)^2} \\
	&\le \ \sum_{i,j=1}^q
		\frac{(v_i^\top \dot{M}(t) v_j^{})^2}{\sqrt{\mu_i\mu_j}^2} \\
	&= \ \bigl\| M(t)^{-1/2} \dot{M}(t) M(t)^{-1/2} \bigr\|^2
		\ = \ \|\dot{M}(t)\|_{M(t)}^2 .
\end{align*}\\[-5ex]
\strut \hfill $\Box$

\subsection{Basic considerations about convexity and smoothness}

\begin{Lemma}[A criterion for convexity]
\label{lem:convexity}
Let $\TT$ be a real interval and $f : \TT \to \R$ such that for any fixed $t \in \TT$ there exist real numbers $g(t), h(t)$ such that
\[
	f(t + \delta) \ = \ f(t) + g(t) \delta + h(t) \delta^2/2 + o(\delta^2)
	\quad\text{as} \ \delta \to 0 .
\]
If $h(t) \ge 0$ for all $t \in \TT$, then $f$ is convex. If $h(t) > 0$ for all $t \in \TT$, then $f$ is strictly convex.
\end{Lemma}

\begin{Remark}
The second order Taylor expansion in Lemma~\ref{lem:convexity} implies that $f(t + \delta) = f(t) + g(t)\delta + o(\delta)$ as $\delta \to 0$. Thus $f$ is differentiable with $f' = g$. However, it does \textsl{not} imply that $f$ is twice differentiable. As a counterexample consider $\TT = \R$ and
\[
	f(x) \ := \ \begin{cases}
		0 & \text{for} \ x = 0 , \\
		x^3 \sin(1/x^2) & \text{for} \ x \ne 0 .
	\end{cases}
\]
This function $f$ is obviously infinitely often differentiable on $\R \setminus \{0\}$, and $f(\delta) = O(\delta^3)$ as $\delta \to 0$, so $g(0) = h(0) = 0$. But for $x \ne 0$, the first derivative $f'(x) = 3 x^2 \sin(1/x^2) - 2 \cos(1/x^2)$ has no limit as $x \to 0$.
\end{Remark}

\begin{proof}[\bf Proof of Lemma~\ref{lem:convexity}].
Since $f$ is continuous, it suffices to show that for arbitrary points $t_0 < t_2$ in $\TT$ and their midpoint $t_1 := (t_0 + t_2)/2$, the value $f(t_1)$ is not greater than (strictly smaller than) $(f(t_0) + f(t_2))/2$. Note that there exists a unique quadratic function $g = g_{t_0,t_1,t_2} : \R \to \R$ such that $g(t_j) = f(t_j)$ for $j = 0,1,2$, namely,
\[
	g(t) \ = \ f(t_0) + (t - t_0) \frac{f(t_2) - f(t_0)}{t_2 - t_0}
		- (t - t_0)(t_2 - t) h(t_0,t_1,t_2) / 2
\]
with
\[
	h(t_0,t_1,t_2) \ := \ \frac{8}{(t_2 - t_0)^2}
		\Bigl( \frac{f(t_0) + f(t_2)}{2} - f(t_1) \Bigr) .
\]
Note also that $g''(t) = h(t_0,t_1,t_2)$ for all $t$. But $h(t_0,t_1,t_2)$ is greater or equal to the minimum of $h(t_0',t_1',t_2')$ when $(t_0',t_1',t_2')$ runs through $(t_0,t_{0.5},t_1)$, $(t_{0.5},t_1,t_{1.5})$ and $(t_1, t_{1.5}, t_2)$ with the midpoints $t_{0.5} := (t_0+t_1)/2$ and $t_{1.5} := (t_1 + t_2)/2$. For if $f(t_{0.5}) > g(t_{0.5})$, then $h(t_0,t_{0.5},t_1) < h(t_0,t_1,t_2)$, and if $f(t_{1.5}) > g(t_{1.5})$, then $h(t_1,t_{1.5},t_2) < h(t_0,t_1,t_2)$. But $f(t_{0.5}) \le g(t_{0.5})$ and $f(t_{1.5}) \le g(t_{1.5})$ together imply that $h(t_{0.5},t_1,t_{1.5}) \le h(t_0,t_1,t_2)$.

Consequently there exist triplets $(t_{n,0}, t_{n,1}, t_{n,2})$ for $n = 0,1,2,\ldots$ such that $(t_{0,0},t_{0,1},t_{0,2}) = (t_0,t_1,t_0)$, and
\[
	\left.\begin{array}{cl}
		h(t_{n,0},t_{n,1},t_{n,2})
			& \text{is non-increasing} \\
		t_{n,0}
			& \text{is non-decreasing} \\
		t_{n,2}
			& \text{is non-increasing}
	\end{array}\right\}
	\ \text{in} \ n
\]
with $t_{n,1} = (t_{n,0} + t_{n,2})/2$ and $t_{n,2} - t_{n,0} = 2^{-n} (t_2 - t_0)$. In particular, the three sequences $(t_{n,0})_n$, $(t_{n,1})_n$ and $(t_{n,2})_n$ converge to the same point $t_* \in [t_0,t_2]$, and
\[
	f(t_{n,j}) \ = \ f(t_*) + g(t_*) (t_{n,j} - t_*) + h(t_*) (t_{n,j} - t_*)^2/2
		+ o((t_{n,2} - t_{n,0})^2)
\]
for $j=0,1,2$. But then elementary calculations show that
\[
	\lim_{n \to \infty} h(t_{n,0},t_{n,1},t_{n,2})
	\ = \ h(t_*) ,
\]
whence $h(t_0,t_1,t_2) \ge h(t_*)$.
\end{proof}

Existence of second order Taylor expansions is equivalent to twice continuous differentiability, provided that the quadratic term depends continuously on the location:

\begin{Lemma}[2nd order Taylor expansions and differentiability]
\label{lem:2nd.order.Taylor}
Let $\Omega$ be an open subset of $\R^d$, and let $f : \Omega \to \R$ have the following property: For each $x \in \Omega$ there exist a vector $g(x) \in \R^d$ and a matrix $H(x) \in \R^{d\times d}_{\rm sym}$ such that
\[
	f(x + v) \ = \ f(x) + g(x)^\top v + 2^{-1} v^\top H(x) v + o(\|v\|^2)
	\quad\text{as} \ v \to \infty .
\]
Further suppose that $H : \Omega \to \R^{d\times d}_{\rm sym}$ is continuous. Then $f$ is twice continuously differentiable with $g_i(x) = \partial f(x) / \partial x_i$ and $H_{ij}(x) = \partial^2 f(x) / (\partial x_i \partial x_j)$.
\end{Lemma}

\begin{proof}[\bf Proof of Lemma~\ref{lem:2nd.order.Taylor}]
We start with dimension $d = 1$. For $t \in \Omega$ and $\delta > 0$ let $c_0$ be the infimum and $c_1$ the supremum of $h$ on $\Omega(t,\delta) := [t \pm \delta] \cap \Omega$. Now we apply Lemma~\ref{lem:convexity} to $\tilde{f}_j(x) := f(x) - c_j (x-t)^2/2$ and $\Omega(t,\delta)$ in place of $f$ and $\TT$, respectively. Note that
\[
	\tilde{f}_j(x+s) \ = \ \tilde{f}_j(x) + \tilde{g}_j(x)s + \tilde{h}_j(x)s^2/2
		+ o(s^2) \quad\text{as} \ s \to 0 ,
\]
where $\tilde{g}_j(x) = g(x) + c_j(x - t)$ and $\tilde{h}_j(x) = h(x) - c_j$. This shows that $\tilde{f}_0$ is convex and $\tilde{f}_1$ is concave on $\Omega(t,\delta)$. In particular, $\tilde{g}_0(x)$ is non-decreasing and $\tilde{g}_1(x)$ is non-increasing in $x \in \Omega(t,\delta)$. Thus we may conclude that
\[
	\frac{g(x) - g(t)}{x - t}
	\ = \ \frac{\tilde{g}_j(x) - \tilde{g}_j(t) + c_j(x - t)}{x - t}
	\ = \ \frac{\tilde{g}_j(x) - \tilde{g}_j(t)}{x - t} + c_j
	\ \in \ [c_0,c_1]
\]
for $x \in \Omega(t,\delta) \setminus \{t\}$. Letting $\delta \downarrow 0$ shows that
\[
	g'(t) = f''(t) = h(t) .
\]

Now we consider dimension $d \ge 2$. We have to show that for any point $x \in \Omega$ and any fixed unit vector $u \in \R^d$,
\[
	u^\top g(x + v) - u^\top g(x) \ = \ u^\top H(x) v + o(\|v\|)
	\quad\text{as} \ v \to 0 .
\]
Our assumption on $f$ and the result for the one-dimentional case imply that for arbitrary $y \in \Omega$ and $w \in \R^d$, the function $t \mapsto f(y + tw)$ is twice continuously differentiable on the set $\{t \in \R : y + tw \in \Omega\}$. Now for our given $x \in \Omega$ and $y,w \in \R^d$ with sufficiently small norms $\|y - x\|$ and $\|w\|$ we may write
\[
	f(y + w) \ = \ f(y) + g(y)^\top w + 2^{-1} w^\top H(x) w + \rho(x,y,w)
\]
with
\[
	\rho(x,y,w) \ := \ \int_0^1 w^\top \bigl( H(y + sw) - H(x) \bigr) w \, ds .
\]
Note that
\[
	|\rho(x,y,w)| \ \le \ \|w\|^2 R(x,\|y - x\| + \|w\|)
\]
where
\[
	R(x,\delta) \ := \ 2^{-1} \sup_{z \in \Omega \,:\, \|z - x\| \le \delta}
		\bigl\| H(z) - H(x) \bigr\| .
\]
Consequently, for any unit vector $u \in \R^d$ and any vector $v \in \R^d$ with sufficiently small norm $r := \|v\| > 0$,
\begin{align*}
	u^\top g(x+v)
		& - u^\top g(x) \\
	= \ r^{-1} \bigl(
		& (ru)^\top g(x + v) - (ru)^\top g(x) \bigr) \\
	= \ r^{-1} \Bigl(
		& f(x + v + ru) - f(x + v) - 2^{-1} (ru)^\top H(x) (ru) - \rho(x,x+v,ru) \\
	 	& - \ f(x + ru) + f(x) + 2^{-1} (ru)^\top H(x) (ru) + \rho(x,x,ru) \Bigr) \\
	= \ r^{-1} \Bigl(
		& f(x + v + ru) - f(x + v) - f(x + ru) + f(x) \\
		& \quad - \ \rho(x,x+v,ru) + \rho(x,x,ru) \Bigr) \\
	= \ r^{-1} \bigl(
		& f(x + v + ru) - f(x + v) - f(x + ru) + f(x) \bigr)
			+ o(r) \quad\text{as} \ r = \|v\| \to 0 ,
\end{align*}
because
\[
	|\rho(x,x+v,ru)| + |\rho(x,x,ru)|
	\ \le \ 2 r^2 R(x,2r)
	\ = \ o(r^2) .
\]
If we write each term $f(x + w)$ as $f(x) + g(x)^\top w + 2^{-1} w^\top H(x) w + r(x,x,w)$, then elementary algebra shows that
\begin{align*}
	r^{1} \bigl(
		& f(x + v + ru) - f(x + v) - f(x + ru) + f(x) \bigr) \\
		&= \ u^\top H(x) v
			+ r^{-1} \bigl( \rho(x,x,v + ru) - \rho(x,x,v) - \rho(x,x,ru) \bigr) \\
		&= \ u^\top H(x) v
			+ o(r) \quad\text{as} \ r \to 0 ,
\end{align*}
because
\[
	|\rho(x,x,v + ru)| + |\rho(x,x,v)| + |\rho(x,x,ru)|
	\ \le \ 4 r^2 R(x,2r) + 2 r^2 R(x,r)
	\ = \ o(r) .
\]\\[-5ex]
\end{proof}

\subsection{Further Proofs for Section~\ref{sec:G-Convexity}}

\begin{proof}[\bf Proof of Lemma~\ref{lem:smoothness1}]
If $f$ is differentiable, then for arbitrary $\Sigma \in \Rqqsympd$ and $\Delta \in \Rqqsym$ with $\Sigma + \Delta \in \Rqqsympd$,
\[
	f(\Sigma + \Delta) \ = \ f(\Sigma) + \langle \Delta,\nabla f(\Sigma)\rangle
		+ o(\|\Delta\|)
	\quad\text{as} \ \Delta \to 0.
\]
This implies that for $B \in \Rqqns$ and $A \in \Rqqsym$,
\begin{align*}
	f(B \exp(A) B^\top) \
	&= \ f(BB^\top + BAB^\top + O(\|A\|^2)) \\
	&= \ f(BB^\top) + \langle BAB^\top, \nabla f(BB^\top)\rangle + o(\|A\|) \\
	&= \ f(BB^\top) + \langle A, B^\top \nabla f(BB^\top) B \rangle + o(\|A\|)
\end{align*}
as $A \to 0$. Hence Condition~(S1.ii) is satisfied with $G(B) = B^\top \nabla f(BB^\top) B$.

If $f$ satisfies Condition~(S1.ii), then for arbitrary $\Sigma \in \Rqqsympd$ and $\Delta \in \Rqqsym$ with $\Sigma + \Delta \in \Rqqsympd$,
\[
	f(\Sigma + \Delta)
	\ = \ f \bigl( \Sigma^{1/2} \exp(A) \Sigma^{1/2} \bigr)
\]
with
\[
	A \ := \ \log(I_q + \Sigma^{-1/2} \Delta \Sigma^{-1/2})
	\ = \ \Sigma^{-1/2}\Delta\Sigma^{-1/2} + O(\|\Delta\|^2)
\]
as $\Delta \to 0$, whence
\begin{align*}
	f(\Sigma + \Delta) \
	&= \ f(\Sigma) + \langle A, G(\Sigma^{1/2})\rangle + o(\|A\|) \\
	&= \ f(\Sigma) + \langle \Delta, \Sigma^{-1/2} G(\Sigma^{1/2}) \Sigma^{-1/2} \rangle
		+ o(\|\Delta\|)
\end{align*}
as $\Delta \to 0$. Thus $f$ is differentiable with gradient $\nabla f(\Sigma) = \Sigma^{-1/2} G(\Sigma^{1/2}) \Sigma^{-1/2}$ at $\Sigma$.
\end{proof}

\begin{proof}[\bf Proof of Lemma~\ref{lem:smoothness2}]
Suppose first that $f$ is twice continuously differentiable. This implies that for $\Sigma \in \Rqqsympd$ and $\Delta \in \Rqqsym$ with sufficiently small norm $\|\Delta\|$,
\begin{equation}
\label{eq:Taylor2.trad}
	f(\Sigma + \Delta)
	\ = \ f(\Sigma) + \langle \Delta, \nabla f(\Sigma)\rangle
		+ 2^{-1} Q(\Delta, \Sigma) + o(\|\Delta\|^2)
\end{equation}
with the quadratic form $Q(\Delta,\Sigma) := \langle \Delta, D^2 f(\Sigma) \Delta\rangle$.
This implies that for $B \in \Rqqns$ and $A \in \Rqqsym$,
\begin{align*}
	f(B \exp(A) B^\top) \
	=& \ f(BB^\top + BAB^\top + 2^{-1} BA^2B^\top + O(\|A\|^3)) \\
	=& \ f(BB^\top) + \langle BAB^\top, \nabla f(BB^\top)\rangle
		+ 2^{-1} \langle BA^2B^\top, \nabla f(BB^\top)\rangle \\
	 & + \ 2^{-1} Q(BAB^\top, BB^\top) + o(\|A\|^2)
\end{align*}
as $A \to 0$. Hence Condition~(S2.ii) is satisfied with
\[
	H(A,B) \ := \ \langle A^2, G(B)\rangle + Q(BAB^\top, BB^\top) .
\]

Now suppose that $f$ satisfies Condition~(S2.ii). Then for arbitrary $\Sigma \in \Rqqsympd$ and $\Delta \in \Rqqsym$ with $\Sigma + \Delta \in \Rqqsympd$,
\[
	f(\Sigma + \Delta)
	\ = \ f \bigl( \Sigma^{1/2} \exp(A) \Sigma^{1/2} \bigr)
\]
with
\[
	A \ := \ \log(I_q + \Sigma^{-1/2} \Delta \Sigma^{-1/2})
	\ = \ \Sigma^{-1/2}\Delta\Sigma^{-1/2}
		- 2^{-1} \Sigma^{-1/2} \Delta \Sigma^{-1} \Delta \Sigma^{-1/2}
		+ O(\|\Delta\|^3)
\]
as $\Delta \to 0$, whence
\begin{align*}
	f(\Sigma + \Delta) \
	=& \ f(\Sigma) + \langle A, G(\Sigma^{1/2})\rangle
		+ 2^{-1} H(A, \Sigma^{1/2}) + o(\|A\|^2) \\
	=& \ f(\Sigma)
		+ \langle \Delta, \Sigma^{-1/2} G(\Sigma^{1/2}) \Sigma^{-1/2} \rangle
		- 2^{-1} \langle \Delta \Sigma^{-1}\Delta,
			\Sigma^{-1/2} G(\Sigma^{1/2}) \Sigma^{-1/2} \rangle \\
	 & + \ 2^{-1} H(\Sigma^{-1/2}\Delta\Sigma^{-1/2}, \Sigma^{1/2} \rangle
		+ o(\|\Delta\|^2)
\end{align*}
as $\Delta \to 0$. Hence $f$ admits a Taylor expansion \eqref{eq:Taylor2.trad} with $\nabla f(\Sigma) = \Sigma^{-1/2} G(\Sigma^{1/2}) \Sigma^{-1/2}$ and
\[
	Q(\Delta,\Sigma) \
	:= \ H(\Sigma^{-1/2}\Delta\Sigma^{-1/2}, \Sigma^{1/2})
		- \langle \Delta\Sigma^{-1}\Delta, \nabla f(\Sigma)\rangle .
\]
Moreover, this is continuous in $\Sigma \in \Rqqsympd$ for any fixed $\Delta$. Now we may conclude from Lemma~\ref{lem:2nd.order.Taylor} with $d = q(q+1)/2$ and $\Omega = \Rqqsympd$ that $f$ is indeed twice continuously diffrentiable.
\end{proof}

\subsection{Further Proofs for Section~\ref{sec:Regularized.scatter}}

\begin{proof}[\bf Proof of Proposition~\ref{prop:existence}]
We use essentially the same arguments as \cite{Duembgen_etal_2015}. With $h(t) := \rho(e^t)$ we may write $h'(t\,+) = \psi(e^t)$ and
\[
	\rho(x^\top \Sigma^{-1} x) - \rho(\|x\|^2)
	\ = \ \int_{\log(\|x\|^2)}^{\log(x^\top \Sigma^{-1}x)} h'(t\,+) \, dt
	\ = \ \int_{\log(\|x\|^2)}^{\log(x^\top \Sigma^{-1}x)} \psi(e^t) \, dt
\]
for $x \in \R^q \setminus \{0\}$. Since $\lambda_{\rm max}(\Sigma)^{-1} \le x^\top \Sigma^{-1} x / \|x\|^2 \le \lambda_{\rm min}(\Sigma)^{-1}$, we may conclude that
\[
	\bigl| \rho(x^\top \Sigma^{-1} x) - \rho(\|x\|^2) \bigr|
	\ \le \ \log(\lambda) \psi(\lambda \|x\|^2)
\]
with $\lambda := \max\{\lambda_{\rm min}(\Sigma)^{-1}, \lambda_{\rm max}(\Sigma)\} \ge 1$. This shows that Condition~\eqref{eq:existence} is sufficient for integrability of $\rho(x^\top \Sigma^{-1} x) - \rho(\|x\|^2)$ with respect to $Q$.

On the other hand, suppose that $\int \bigl| \rho(x^\top \Sigma^{-1} x) - \rho(\|x\|^2) \bigr| \, Q(dx) < \infty$ for any $\Sigma \in \Rqqsympd$. This implies that $\rho(\lambda_2 \|x\|^2) - \rho(\lambda_1 \|x\|^2)$ is integrable with repect to $Q$ for arbitrary $\lambda_2 > \lambda_1 > 0$. But
\[
	\rho(\lambda_2 \|x\|^2) - \rho(\lambda_1 \|x\|^2)
	\ = \ \int_{\log(\lambda_1 \|x\|^2)}^{\log(\lambda_2 \|x\|^2)} h'(t\,+) \, dt
	\ \ge \ \log(\lambda_2/\lambda_1) \psi(\lambda_1 \|x\|^2) ,
\]
so \eqref{eq:existence} has to be satisfied.

Finally, if $\rho'(\cdot\,+)$ is non-increasing on $\R_+$, then $\psi(\lambda s) = \lambda s \rho'(\lambda s\,+) \le \lambda s \rho'(s\,+) = \lambda \psi(s)$ for arbitrary $\lambda \ge 1$ and $s > 0$. Thus Condition~\eqref{eq:existence} is equivalent to $\int \psi(\|x\|^2) \, Q(dx)$ being finite.
\end{proof}

\begin{proof}[\bf Proof of Theorem~\ref{thm:Mfunc}]
It follows from Proposition~\ref{prop:existence} that $L_\rho(\Sigma,Q)$ is well-defined in $\R$ for arbitrary $\Sigma \in \Rqqsympd$.  For any fixed $R > 1$, the inequalities $\lambda_{\rm min}(\Sigma) \ge R^{-1}$ and $\lambda_{\rm max}(\Sigma) \le R$ imply that
\[
	\bigl| \rho(x^\top \Sigma^{-1} x) - \rho(\|x\|^2) \bigr|
	\ \le \ \log(R) \psi(R \|x\|^2) .
\]
Hence, by dominated convergence, $L_\rho(\Sigma,Q)$ is continuous in $\Sigma \in \Rqqsympd$. Geodesic convexity of $L_\rho(\cdot,Q)$ follows from examples~\ref{ex0} and \ref{exlog}.

Now the question is under which conditions on $\rho$, $Q$, $B \in \Rqqns$ and $\gamma \in \R^q \setminus \{0\}$, the function $t \mapsto L_\rho(B D(e^{-t\gamma}) B^\top,Q)$ is strictly convex on $\R$. With $Q_B = \LL(B^{-1} X)$, $X \sim Q$, one may write
\begin{align*}
	L_\rho(& B D(e^{-t\gamma}) B^\top,Q) - L_\rho(BB^\top,Q) \\
	&= \ \int_{\R^q \setminus \{0\}}
		\bigl[ \rho(x^\top D(e^{t\gamma})x) - \rho(\|x\|^2) \bigr] \, Q_B(dx)
		- t \sum_{i=1}^q \gamma_i .
\end{align*}
Moreover, with $h(t) := \rho(e^t)$ and $g_x(t) := \log(x^\top D(e^{t\gamma})x)$ for fixed $x \ne 0$,
\[
	\rho(x^\top D(e^{t\gamma}) x) \ = \ h(g_x(t)) .
\]
As mentioned in example~\ref{exlog}, $g_x$ is convex with
\[
	g_x'(t) \ = \ \sum_{i=1}^q p_i \gamma_i
	\quad\text{and}\quad
	g_x''(t) \ = \ \sum_{i=1}^q p_i \gamma_i^2
		- \Bigl( \sum_{i=1}^q p_i \gamma_i \Bigr)^2 ,
\]
where $p_i = x_i^2 e^{t\gamma_i} \big/ \sum_{j=1}^q x_j^2 e^{t\gamma_j}$. Hence $g_x$ is strictly convex unless $x$ belongs to
\[
	\V(\gamma_o) \ := \ \{x \in \R^q : x_i = 0 \ \text{if} \ \gamma_i \ne \gamma_o\}
\]
for some value $\gamma_o \in \{\gamma_1,\ldots,\gamma_q\}$. In the latter case, $g$ is linear with slope $\gamma_o$.

As to part~(a), suppose that $\rho(s)$ is strictly g-convex in $s > 0$, which is equivalent to $h$ being strictly convex and strictly increasing. Then $h \circ g_x$ is strictly convex unless $g_x$ is constant, i.e.\ $x \in \V(0)$. Consequently, $L_\rho(B D(e^{-t\gamma}),Q)$ is strictly convex in $t \in \R$, unless $Q(B\V(0)) = 1$. But $\gamma \ne 0$ implies that $\dim(B\V(0)) < q$. On the other hand, suppose that $Q(\V) = 1$ for some linear subspace $\V \subset \R^q$ with dimension $d < q$. If we choose $B = [b_1,\ldots,b_q]$ such that $b_1, \ldots, b_d$ form a basis of $\V$ and $\gamma := (1_{[i > d]})_{i=1}^q$, then $L(B D(e^{-t\gamma}) B^\top, Q)$ is linear in $t \in \R$.

As to part~(b), it suffices to consider matrices $B$ with $\det(B) = \pm 1$ and vectors $\gamma \in \R^q \setminus \{0\}$ with $\sum_{j=1}^q \gamma_j = 0$. Here $h' \equiv q$, so the function $h \circ g_x$ is strictly convex if, and only if, $g_x$ is strictly convex. The latter condition is true, unless $x$ lies in the union of the linear subspaces $\V(\gamma_o)$, $\gamma_o \in \{\gamma_1,\ldots,\gamma_q\}$. Hence $L_\rho(B D(e^{-t\gamma}),Q)$ is strictly convex in $t \in \R$, unless $Q \bigl( \bigcup_{\gamma_o} B \V(\gamma_o) \bigr) = 1$. The latter condition implies that $Q(\V \cup \W) = 1$ with $\V := B \V(\gamma_o)$ and $\W := B(\V(\gamma_o)^\perp)$ and $\gamma_o$ an arbitrary number in $\{\gamma_1,\ldots,\gamma_q\}$. On the other hand, suppose that $Q(\V \cup \W) = 1$ for linear subspaces $\V, \W \subset \R^q$ with respective dimensions $d, e \in [1,q)$ such that $\V \cap \W = \{0\}$. Now we take $B = [b_1, \ldots, b_q]$ such that $\V = \mathrm{span}(b_i : 1 \le i \le d)$, $\W = \mathrm{span}(b_i : d < i \le d+e)$ and $\det(B) = 1$. Further let $\gamma_i := 1_{[i \le d]}/d - 1_{[d < i \le d+e]} / e$. Then $L(B D(e^{-t\gamma}) B^\top, Q)$ is linear in $t \in \R$ while $\det(B D(e^{-t\gamma}) B^\top) \equiv 1$.
\end{proof}

\begin{proof}[\bf Proof of Proposition~\ref{prop:g-coercivity}]
We argue similarly as in the proof of Proposition~5.5 in \cite{Duembgen_etal_2015}. Note first that $L_\rho(\exp(tA), Q) = L_\rho(D(e^{-t\gamma},Q_U)$ with the transformed distribution $Q_U = \LL(U^\top X)$, $X \sim Q$. Thus it suffices to consider the case $A = D(-\gamma)$ and $U = I_q$, so $\V_j = \{x \in \R^q : x_i \ \text{for} \ i > j\}$. For real numbers $t < u$,
\begin{align*}
	& \frac{L_\rho(D(e^{-u\gamma}),Q) - L_\rho(D(e^{-t\gamma}),Q)}{u - t} \\
	& \ = \ \int_{\R^q \setminus \{0\}}
		\frac{\rho(x^\top D(e^{u\gamma}) x) - \rho(x^\top D(e^{t\gamma}) x)}{u - t}
			\, Q(dx) - \sum_{j=1}^q \gamma_j .
\end{align*}
For any fixed $x \in \R^q \setminus \{0\}$ we may write $\rho(x^\top D(e^{t\gamma})x) = h(g_x(t))$, where $h(t) := \rho(e^t)$ and $g_x(t) := \log(x^\top D(e^{t\gamma})x)$. As mentioned in the proof of Theorem~\ref{thm:Mfunc}, the function $h \circ g_x$ is convex. Thus
\[
	\frac{h(g_x(u)) - h(g_x(t))}{u - t}
	\ \in \ \bigl[ h(g_x(t)) - h(g_x(t-1)), h(g_x(t+1)) - h(g_x(t)) \bigr]
\]
for $u \in (t,t+1]$, and
\[
	\eta_x(t) \ := \ \lim_{u \to t\,+} \frac{h(g_x(u)) - h(g_x(t))}{u - t}
\]
is well-defined and non-decreasing in $t \in \R$. Hence by dominated convergence and monotone convergence,
\[
	\lim_{t \to \infty} \ \lim_{u \to t\,+} \,
		\frac{L_\rho(D(e^{-u\gamma}),Q) - L_\rho(D(e^{-t\gamma}),Q)}{u - t}
	\ = \ \int_{\R^q \setminus \{0\}}
		\lim_{t \to \infty} \eta_x(t) \, Q(dx) - \sum_{j=1}^q \gamma_j .
\]
Now we partition $\R^q \setminus \{0\}$ as $\bigcup_{j=1}^q \V_j \setminus \V_{j-1}$. For $x \in \V_j \setminus \V_{j-1}$,
\[
	x^\top D(e^{t\gamma}) x \ = \ \sum_{i=1}^j x_i^2 e_{}^{t\gamma_i}
	\ \to \ \begin{cases}
		\infty & \text{if} \ \gamma_j > 0 \\
		\sum_{i=1}^j x_i^2 1_{[\gamma_i = 0]} & \text{if} \ \gamma_j = 0 \\
		0 & \text{if} \ \gamma_j < 0
	\end{cases}
\]
and
\[
	g_x'(t)
	\ = \ \sum_{i=1}^j x_i^2 e_{}^{t\gamma_i} \gamma_i
		\big/ \sum_{i=1}^j x_i^2 e_{}^{t\gamma_i}
	\ \to \ \gamma_j
\]
as $t \to \infty$. Hence
\begin{align*}
	\lim_{t \to \infty} \eta_x(t) \
	&= \ \lim_{t \to \infty} \bigl\{ h'(g_x(t) \, +) g_x'(t)^+
		- h'(g_x(t) \, -) g_x(t)^- \bigr\} \\
	&= \ \lim_{t \to \infty} \bigl\{ \psi(x^\top D(e^{t\gamma}) x \, +) g_x'(t)^+
		- \psi(x^\top D(e^{t\gamma}) x \, -) g_x(t)^- \bigr\} \\
	&= \ \psi(\infty) \gamma_j^+ - \psi(0 \, +) \gamma_j^- .
\end{align*}
All in all we obtain the asserted limit \eqref{eq:g-coercivity}.

With $\gamma_0 := 0$ we may write $\gamma_j^+ = \sum_{k=0}^{j-1} (\gamma_{k+1}^+ - \gamma_k^+)$, and all summands $\gamma_{k+1}^+ - \gamma_k^+$ are non-negative. Hence in the special case that $\psi(0\,+) = 0$ the limit \eqref{eq:g-coercivity} equals
\begin{align*}
	\sum_{j=1}^q & Q(\V_j\setminus\V_{j-1}) \psi(\infty) \gamma_j^+
		- \sum_{j=1}^q \gamma_j \\
	&= \ \sum_{j=1}^q Q(\V_j\setminus\V_{j-1}) \psi(\infty) \gamma_j^+
		- \sum_{j=1}^q \gamma_j^+
		+ \sum_{j=1}^q \gamma_j^- \\
	&= \ \sum_{j=1}^q Q(\V_j\setminus\V_{j-1}) \psi(\infty)
		\sum_{k=0}^{j-1} (\gamma_{k+1}^+ - \gamma_k^+)
		- \sum_{j=1}^q \sum_{k=0}^{j-1} (\gamma_{k+1}^+ - \gamma_k^+)
		+ \sum_{j=1}^q \gamma_j^- \\
	&= \ \sum_{k=0}^{q-1} (1 - Q(\V_k)) \psi(\infty) (\gamma_{k+1}^+ - \gamma_k^+)
		- \sum_{k=0}^{q-1} (q - k) (\gamma_{k+1}^+ - \gamma_k^+)
		+ \sum_{j=1}^q \gamma_j^- \\
	&= \ \sum_{k=0}^{q-1} \bigl( (1 - Q(\V_k)) \psi(\infty) - q + k \bigr)
			(\gamma_{k+1}^+ - \gamma_k^+)
		+ \sum_{j=1}^q \gamma_j^- .
\end{align*}
In the special case of $\rho(s) = q \log s$ for $s > 0$, $\psi \equiv q$ on $\R_+$, so the limit \eqref{eq:g-coercivity} equals
\begin{align*}
	\sum_{j=1}^q & Q(\V_j\setminus\V_{j-1}) q \gamma_j
		- \sum_{j=1}^q \gamma_j \\
	&= \ q \sum_{k=1}^{q-1} Q(\V_k) \gamma_k
		- q \sum_{k=1}^{q-1} Q(\V_k) \gamma_{k+1}
		+ q \gamma_q
		- \sum_{j=1}^q \gamma_j
		 - q Q(\{0\}) \gamma_1 \\
	&= \ - q \sum_{k=1}^{q-1} Q(\V_k) (\gamma_{k+1} - \gamma_k)
		+ \sum_{j=1}^{q-1} (\gamma_q - \gamma_j)
		- q Q(\{0\}) \gamma_1 \\
	&= \ - q \sum_{k=1}^{q-1} Q(\V_k) (\gamma_{k+1} - \gamma_k)
		+ \sum_{j=1}^{q-1} \sum_{k=j}^{q-1} (\gamma_{k+1} - \gamma_k)
		- q Q(\{0\}) \gamma_1 \\
	&= \ - q \sum_{k=1}^{q-1} Q(\V_k) (\gamma_{k+1} - \gamma_k)
		+ \sum_{k=1}^{q-1} k (\gamma_{k+1} - \gamma_k)
		- q Q(\{0\}) \gamma_1 \\
	&= \ q \sum_{k=1}^{q-1} (k/q - Q(\V_k)) (\gamma_{k+1} - \gamma_k)
		- q Q(\{0\}) \gamma_1 .
\end{align*}\\[-5ex]
\end{proof}

\begin{proof}[\bf Proof of Theorem~\ref{thm:g-coercivity}]
We start with part~(a). According to Lemma~\ref{lem:g-coercivity} and Proposition~\ref{prop:g-coercivity}~(a), $L_\rho(\cdot,Q)$ is g-coercive on $\Rqqsympd$ if, and only if, it satisfies the following inequalities: For any $U = [u_1,\ldots,u_q] \in \Rqqorth$ and $\gamma \in \R^q \setminus \{0\}$ with $\gamma_1 \le \cdots \le \gamma_q$,
\begin{equation}
\label{ineq:g-coercivity1}
	\sum_{k=0}^{q-1} \bigl( (1 - Q(\V_k)) \psi(\infty) - q + k \bigr)
		(\gamma_{k+1}^+ - \gamma_k^+)
	+ \sum_{j=1}^q \gamma_j^-
	\ > \ 0 ,
\end{equation}
where $\V_0 := \{0\}$ and $\V_j := \mathrm{span}(u_1,\ldots,u_j)$, $1 \le j \le q$, and $\gamma_0 := 0$. If we choose $\gamma = (1_{[i > k]})_{i=1}^q$ for a fixed index $k \in \{0,\ldots,q-1\}$, then the left hand side of \eqref{ineq:g-coercivity1} equals $(1 - Q(\V_k))\psi(\infty) - q + k$ which is positive if, and only if, $Q(\V_k) < 1 - \{q - k\}/\psi(\infty)$. Note also that all differences $\gamma_{k+1}^+ - \gamma_k^+$ are non-negative. This shows that \eqref{ineq:g-coercivity1} is satisfied for arbitrary nonzero vectors $\gamma$ with non-decreasing components if, and only if,
\[
	Q(\V_k) \ < \ 1 - \frac{q - k}{\psi(\infty)}
	\quad\text{for} \ 0 \le k < q .
\]
But since $u_1, u_2, \ldots, u_q$ is an arbitrary orthonormal basis of $\R^q$, these considerations show that g-coercivity of $L_\rho(\cdot,Q)$ is equivalent to \eqref{eq:g-coercivity1} for arbitrary linear subspaces $\V \subset \R^q$ with $0 \le \dim(\V) < q$.

By virtue of Lemma~\ref{lem:existence.minimizers}, g-coercivity of $L_\rho(\cdot,Q)$ guarantees the existence of a minimizer $\Sigma \in \Rqqsympd$ of $L_\rho(\cdot,Q)$. It remains to be shown that this minimizer is unique in case of $\psi$ being strictly increasing on the interval $\{s \ge 0 : \psi(s) < \psi(\infty)\}$.

If the latter interval equals $[0,\infty)$, then the function $\rho(s)$ is strictly g-convex in $s > 0$, so it follows from Theorem~\ref{thm:Mfunc} and Condition~\eqref{eq:g-coercivity1} for arbitrary linear subspaces $\V$ of $\R^q$ with $\dim(\V) < q$ that $L_\rho(\cdot,Q)$ is strictly g-convex. Hence the minimizer $\Sigma$ is unique, see Corollary~\ref{cor:uniqueness.minimizer}.

Now suppose that $\psi(s_o) = \psi(\infty)$ for some $s_o \in \R_+$. Writing $\Sigma = BB^\top$ with $B \in \Rqqns$, it suffices to show that for any fixed $\gamma \in \R^q \setminus \{0\}$, the function $f : \R \to \R$ with
\[
	f(t)
	\ := \ L_\rho(B D(e^{-t\gamma}) B^\top, Q) - L_\rho(BB^\top, Q)
	\ = \ L_\rho(D(e^{-t\gamma}), Q_B)
\]
has a unique minimum at $t = 0$, where $Q_B = \LL(B^{-1} X)$, $X \sim Q$. As shown in the proof of Theorem~\ref{thm:Mfunc}, $f$ is convex, and optimality of $\Sigma = BB^\top$ implies that $f \ge f(0) = 0$. It remains to be shown that
\begin{equation}
\label{eq:optimality.f0}
	f(t) > 0 \ \ \text{whenever} \ t \ne 0 .
\end{equation}
Recall that
\[
	f(t)
	\ = \ \int_{\R^q \setminus \{0\}} \bigl[ h(g_x(t)) - h(g_x(0)) \bigr] \, Q_B(dx)
		- t \sum_{i=1}^q \gamma_i
\]
with $h(u) := \rho(e^u)$ and $g_x(t) := \log(x^\top D(e^{t\gamma})x)$. Since $\psi(s) > 0$ for all $s > 0$, the function $h$ is convex and strictly increasing. Moreover, $g_x$ is strictly convex unless $x$ is an eigenvector of $D(\gamma)$. Thus $f$ is strictly convex, unless
\begin{equation}
\label{eq:optimality.f1}
	Q_B \Bigl( \bigcup_{\gamma_o \in \{\gamma_1,\ldots,\gamma_q\}} \V(\gamma_o) \Bigr)
	\ = \ 1 ,
\end{equation}
where $\V(\gamma_o) := \{x \in \R^q : x_i = 0 \ \text{if} \ \gamma_i \ne \gamma_o\}$. Since $f \ge f(0) = 0$, strict convexity of $f$ implies \eqref{eq:optimality.f0}.

Suppose that \eqref{eq:optimality.f1} is true. Then we may write $f(t) = \sum_{\gamma_o \in \{\gamma_1,\ldots,\gamma_q\}} f_{\gamma_o}(\gamma_o t)$ with
\[
	f_{\gamma_o}(u) \ := \ \int_{\V(\gamma_o)\setminus\{0\}}
		\bigl[ \rho(e^u \|x\|^2) - \rho(\|x\|^2) \bigr] \, Q_B(dx)
		- \dim(\V(\gamma_o)) u .
\]
Note that
\[
	f_{\gamma_o}(u) \ = \ L_\rho(D(e^{-u \tilde{\gamma}}), Q_B)
	\quad\text{with}\quad
	\tilde{\gamma} := (1_{[\gamma_i = \gamma_o]})_{i=1}^q ,
\]
so each function $f_{\gamma_o}$ is convex with $f_{\gamma_o} \ge f_{\gamma_o}(0) = 0$. Consequently it suffices to show that for any $\gamma_o \in \{\gamma_1,\ldots,\gamma_q\}$,
\begin{equation}
\label{eq:optimality.f2}
	f_{\gamma_o}(u) \ > \ 0 \quad\text{for any} \ u \ne 0 .
\end{equation}
Note that $f_{\gamma_o}(u) = 0$ for some $u \ne 0$ would imply that $f_{\gamma_o}'(v\,+) = f_{\gamma_o}'(w\,+) = 0$ for real numbers $v < w$. But
\[
	f_{\gamma_o}(t\,+)
	\ = \ \int_{\V(\gamma_o)\setminus\{0\}} \psi(e^t \|x\|^2) \, Q_B(dx)
		- \dim(\V(\gamma_o)) ,
\]
so
\[
	0 \ = \ \int_{\V(\gamma_o)\setminus\{0\}}
		\bigl[ \psi(e^w \|x\|^2) - \psi(e^v \|x\|^2) \bigr] \, Q_B(dx) .
\]
The strict monotonicity property of $\psi$ would imply that $\psi(e^v \|x\|^2) = \psi(\infty)$ for $Q_B$-almost all $x \in \V(\gamma_o) \setminus \{0\}$. Hence
\begin{align*}
	f_{\gamma_o}(v\,+) \
	&= \ \psi(\infty) Q_B(\V(\gamma_o) \setminus \{0\}) - \dim(\V(\gamma_o)) \\
	&= \ \psi(\infty) (1 - Q_B(\V(\gamma_o)^\perp) - \dim(\V(\gamma_o)) \\
	&> \ \psi(\infty) \frac{q - \dim(\V(\gamma_o)^\perp)}{\psi(\infty)}
		- \dim(\V(\gamma_o)) \\
	&= \ 0 ,
\end{align*}
a contradiction to $f_{\gamma_o}'(v\,+) = 0$. In the latter display we used \eqref{eq:optimality.f1} in the second and \eqref{eq:g-coercivity1} in the third step.

Concerning part~(b), Lemma~\ref{lem:g-coercivity} with the modifications mentioned in Section~\ref{subsec:Scale-invariance} and Proposition~\ref{prop:g-coercivity}~(b) imply that $L_\rho(\cdot,Q)$ is g-coercive on $\M^{(q)}$ if, and only if, it satisfies the following inequalities: For any $U = [u_1,\ldots,u_q] \in \Rqqorth$ and $\gamma \in \R^q \setminus \{0\}$ with $\gamma_1 \le \cdots \le \gamma_q$ and $\sum_{j=1}^q \gamma_j = 0$,
\begin{equation}
\label{ineq:g-coercivity0}
	\sum_{k=1}^{q-1} (k/q - Q(\V_k)) (\gamma_{k+1} - \gamma_k)
	\ > \ 0
\end{equation}
with $\V_k := \mathrm{span}(u_1,\ldots,u_k)$. If we choose $\gamma = (k/q - 1_{[i \le k]})_{i=1}^q$, then the left hand side of \eqref{ineq:g-coercivity0} equals $k/q - Q(\V_k)$. Note also that all differences $\gamma_{k+1} - \gamma_k$ are non-negative. Thus \eqref{ineq:g-coercivity0} is true for arbitrary vectors $\gamma \in \R^q \setminus \{0\}$ with non-decreasing components summing to zero if, and only if, $Q(\V_k) < k/q$ for $1 \le k < q$. Hence g-coercivity of $L_\rho(\cdot,Q)$ on $\M^{(q)}$ is equivalent to \eqref{eq:g-coercivity0} for arbitrary linear subspaces $\V \subset \R^q$ with $1 \le \dim(\V) < q$.

The latter condition implies the assumption in part~(b) of Theorem~\ref{thm:Mfunc}. Thus $L_\rho(\cdot,Q)$ has a unique minimizer on $\M^{(q)}$.
\end{proof}

\begin{proof}[\bf Graphical LASSO and g-convexity]
Note that g-convexity of $\pi(\Sigma) := \sum_{i < j} |(\Sigma^{-1})_{ij}|$ would be equivalent to g-gonvexity of $f(\Sigma) := \pi(\Sigma^{-1}) = \sum_{i<j} |\Sigma_{ij}|$. Now consider
\[
	B \ = \ \begin{bmatrix}
		B_o & 0 \\
		0 & I_{q-2}
	\end{bmatrix}
	\quad\text{with}\quad
	B_o \ = \ \begin{bmatrix}
		1 & -1 \\
		1 & 1
	\end{bmatrix}
\]
and $x = (a,-1,0,\ldots,0)^\top$ with $a > 1$. Then
\[
	f(B D(e^{tx}) B^\top)
	\ = \ |e^{at} - e^{-t}| .
\]
But for $t > 0$, the right hand side equals $h(t) = e^{at} - e^{-t}$ with $h''(t) = a^2 e^{at} - e^{-t} < 0$ for $t < 2 \log(a)/(a-1)$.
\end{proof}

\subsection{Further proofs for Section~\ref{sec:Algorithm}}

\begin{proof}[\bf Proof of Lemma~\ref{lem:algorithm}]
By definition, the sequence $(f(\Sigma_k))_k$ is non-increasing, and $(\Sigma_k)_k$ stays in the compact set $\{\Sigma \in \Rqqsympd : f(\Sigma) \le f(\Sigma_1)\}$. Suppose $(\Sigma_k)_k$ does not converge to $\Sigma_*$. Then there exists a subsequence $(\Sigma_{k(\ell)})_\ell$ with limit $\Sigma_o \ne \Sigma_*$. It follows from continuity of $f$ and monotonicity of $(f(\Sigma_k))_k$ that
\[
	f(\Sigma_o)
	\ = \ \lim_{\ell \to \infty} f(\Sigma_{k(\ell)})
	\ = \ \lim_{\ell \to \infty} f(\Sigma_{k(\ell) + 1})
	\ = \ \lim_{\ell \to \infty} f(\phi(\Sigma_{k(\ell)})) .
\]
But this contradicts our assumption of $\phi$, because
\[
	f(\Sigma_o)
	\ > \ \limsup_{\Sigma \to \Sigma_o} f(\phi(\Sigma))
	\ \ge \ \limsup_{\ell \to \infty} f(\phi(\Sigma_{k(\ell)})) .
\]\\[-5ex]
\end{proof}

\begin{proof}[\bf Proof of Lemma~\ref{lem:remainders}]
Recall that for any function $g \in \mathcal{C}^2([0,1])$,
\[
	g(1) - g(0) - g'(0) \ = \ \int_0^1 \bigl( g'(t) - g'(0) \bigr) \, dt
	\ = \ \int (1 - t) g''(t) \, dt ,
\]
whence
\[
	g(1) - g(0) - g'(0) - g''(0)/2
	\ = \ \int (1 - t) \bigl( g''(t) - g''(0) \bigr) \, dt .
\]
Note that $g(t) := f(B \exp(tA) B^\top)$ defines a function $g \in \mathcal{C}^2([0,1])$ with
\[
	g'(t) \ = \ \langle A, G(B \exp(tA/2)) \rangle
	\quad\text{and}\quad
	g''(t) \ = \ H(A, B \exp(tA/2)) .
\]
Moreover, $f(B \exp(A) B^\top) = g(1)$, $f(\Sigma) = g(0)$, $g'(0) = \langle A, G(B)\rangle$ and $g''(0) = H(A,B)$. But $B = \Sigma^{1/2} V$ for some orthogonal matrix $V \in \Rqq$, and
\begin{align*}
	H(A, B \exp(tA/2)) \
	&= \ H(A, \Sigma^{1/2} \exp(t VAV^\top/2) V) \\
	&= \ H(VAV^\top, \Sigma^{1/2} \exp(t VAV^\top/2)) \\
	&= \ \|A\|^2 H(\tilde{A}, \Sigma^{1/2} \exp(C/2))
\end{align*}
with $\tilde{A} := \|A\|^{-1} VAV^\top$ and $C := t \|A\| \tilde{A}$, so $\|\tilde{A}\| = 1$ and $\|C\| \le \|A\|$; see also Remark~\ref{rem:Orthogonal transformations}. Thus
\[
	g''(t) \ \le \ \|A\|^2 \Lambda_{\rm max}(\Sigma,\|A\|)
	\quad\text{and}\quad
	\bigl| g''(t) - g''(0) \bigr| \ \le \ \|A\|^2 N(\Sigma,\|A\|) .
\]\\[-5ex]
\end{proof}

\end{document}